\documentclass[natbib]{svjour3}
\usepackage{lkfextradef}
\numberwithin{equation}{section}
\setcitestyle{round,citesep{:},aysep{,},yysep{,}}
\smartqed
\begin{document}
\title{Becoming Large, Becoming Infinite}
\subtitle{The Anatomy of Thermal Physics and Phase Transitions in Finite Systems}
\author{David A. Lavis, Reimer K\"uhn\\ and Roman Frigg}
\institute{King's College London,
Department of Mathematics,
 London WC2R 2LS, U.K. (for RK and DAL);
 London School of Economics,
 Centre for Philosophy of Natural and Social Science,
 London WC2A 2AE, U.K. (for RF and DAL).
 For correspondence: \email{david.lavis@kcl.ac.uk}}
 \date{\today}
%----------------------------------------------------------------------------------------------------
 \maketitle
%----------------------------------------------------------------------------------------------------
\begin{abstract} This paper presents an in-depth analysis of the anatomy of both thermodynamics
and statistical mechanics, together with the relationships between their constituent parts.
 Based on this analysis, using the renormalization
   group and finite-size scaling,  we give a definition of a large but finite
  system and argue that phase transitions are represented correctly, as incipient singularities in such systems.
We describe the role of the thermodynamic limit.  And we
explore the implications of this picture of critical phenomena for the questions of reduction and emergence.
\end{abstract}
%--------------------------------------------------------------------------
\keywords{scaling, renormalization, large systems, incipient singularities, reduction, emergence.}
%--------------------------------------------------------------------------------
\tableofcontents
%---------------------------------------------------------------------------------------------------------------------------------------
 \section{Introduction}\label{intro}
 Thermodynamics and statistical mechanics coexist in a collaborative relationship within the envelope of thermal physics.  In many presentations of the subject, particularly in undergraduate texts,
 it is heuristically advantageous to intermingle the macroscopic concepts of thermodynamics with the micro-picture provided by statistical mechanics.  And it is,
 of course, self-evident that statistical mechanics\footnote{At least in its application to physics, rather than to its more modern application to sociological and geographical problems.} needs the basic structure of thermodynamics with inter-theory connecting relationships defining the thermodynamic quantities like internal energy, temperature and entropy.
 On the other hand, there are some advantages, both aesthetic and mathematical, in producing an account of thermodynamics
which makes no reference to the underlying  microstructure  of the system, as would seem to be one of the aims of (among others) the books of \citet{giles1} and \citet{buch1} and the papers
  of \citeauthor{l&y2}.\footnote{\citet{l&y2} is the most comprehensive account of their work, with briefer versions in
   \citet{l&y1} and  \citet{l&y3}. The extension to non-equilibrium is given in  \citet{l&y4}.}  For \citeauthor{buch1} we have the first law implying the existence
   of the internal energy function $U$ and  \citeauthor{cara1}'s (\citeyear{cara1}) version of the second law yielding the entropy $S$ and temperature $T$;
   for \citeauthor{l&y2} three sets of axioms accomplish the same task. This, together with an account of the nature of adiabatic processes (as described, for example, in \citealt{buch1},
   Chaps. 5 and 6;  \citealt{l&y2},  Sect.\ 2.1;   \citealt{Lavis-ep}, Sect. 2.1.1) provides the basic
framework into which the models of statistical mechanics  are embedded.

This raises the question of how statistical mechanics and thermodynamics relate to each other. Attempts to answer this question run up against a problem. The neat labels  `statistical mechanics' and `thermodynamics' mask the fact that neither theory is a monolithic bloc. Indeed, each has a complicated internal structure with several layers of different theoretical postulates and assumptions. So the question of how statistical mechanics and thermodynamics relate ought to be interpreted as the more complex
question of (a) what the internal structure of each theory is and of (b) how the various parts of each theory relate to the various other parts of the other theory. The complexity of the internal structures of both theories, as well as the intricacy of their interrelations, seems to have been somewhat under appreciated in the philosophical literature on the subject, and so the first aim of this paper is to present an in-depth analysis of the anatomy of both theories and the connections between their parts.\footnote{For surveys of the philosophical discussions about statistical mechanics and thermodynamics see, for instance, \citet{skl1}, \citet{Uff3} and \citet{Frigg6}.}

\vspace{0.2cm}
%--------------------------------------------------------------------------------
 \begin{figure}[h]
\includegraphics[width=12cm]{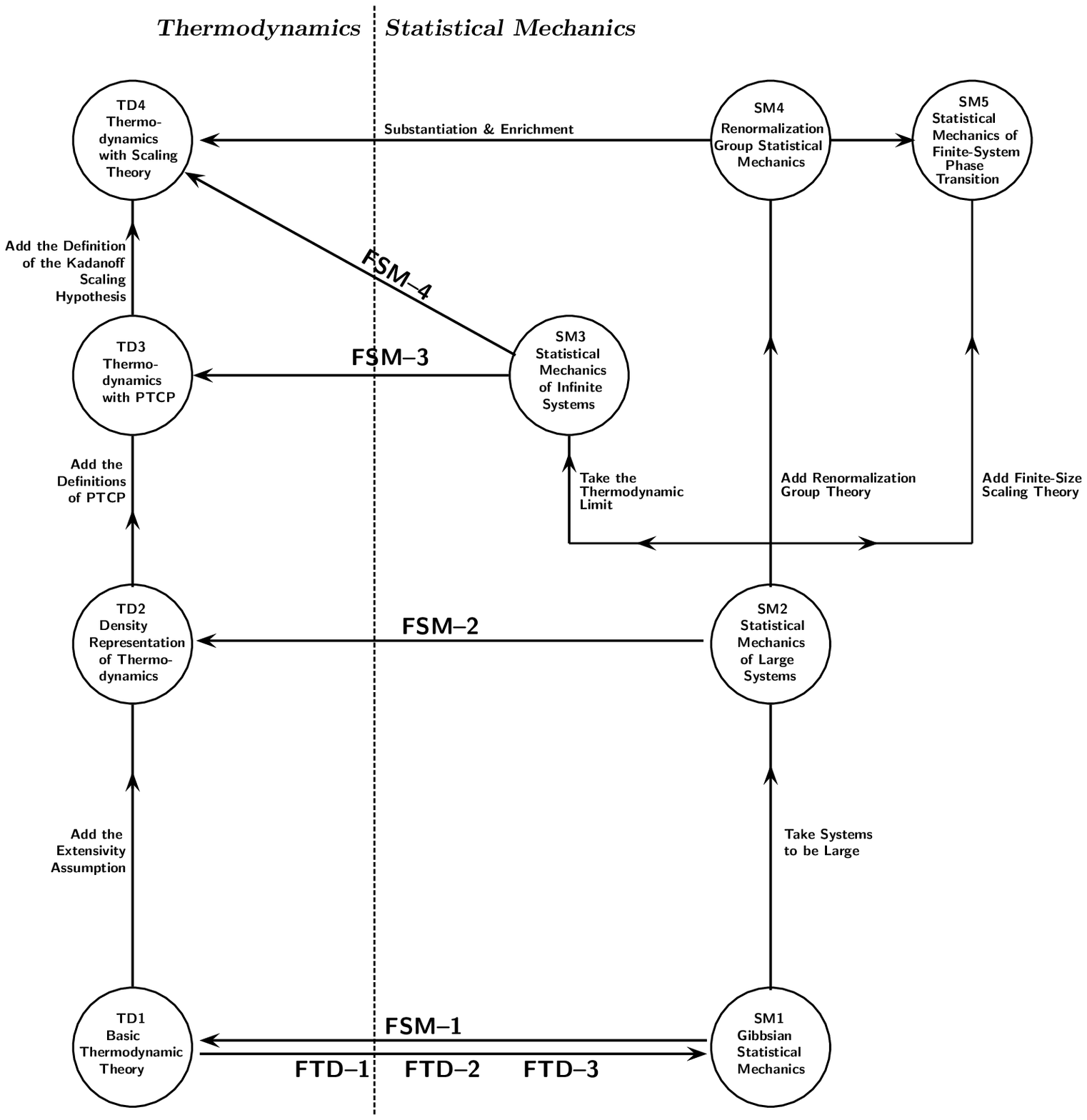}
\caption{Schematic representation of the relationship between thermodynamics and statistical mechanics.  }\label{Fig1}
\end{figure}
%---------------------------------------------------------------------------------------------------------------------------------------------------
\noindent Fig. \ \ref{Fig1} provides a schematic advance summary of the analysis that we develop in this paper.  It sees statistical mechanics and thermodynamics as
 parallel developments, each decomposed into separate levels representing the stages of theory-based development in which  features are added to the system.
  The cross-interactions between the levels in the two columns contain interventions integral to this development.
    On the left are  the levels for thermodynamics, as  described in detail  in   Sect.\ \ref{therm}.
These levels are related to each other by adopting special assumptions,
beginning at the bottom with  \emb{ basic thermodynamic theory} (labelled $\textsf{TD1}$).
 Adding the extensivity assumption to this theory takes us to the next level,
 the \emb{density representation of thermodynamics} (labelled $\textsf{TD2}$).
 Augmenting $\textsf{TD2}$ with the notion of phase transitions and critical phenomena
 (PTCP) gives \emb{thermodynamics with PTCP} (labelled $\textsf{TD3}$).
 Finally, supplementing $\textsf{TD3}$ with a version of the Kadanoff scaling hypothesis leads
 us to \emb{thermodynamics with scaling theory} (labelled $\textsf{TD4}$).

  The parallel development for statistical mechanics
 is represented on the right of Fig.\ \ref{Fig1}, as  described in detail in Sect.\ \ref{statmech}. The picture here
  is a little more complicated, involving, as we explain in our discussion, three different paths. At the bottom is the fundamental theory,
   which we here take to be \emb{Gibbsian statistical mechanics}
    (labelled $\textsf{SM1}$).\footnote{We set aside Boltzmannian statistical mechanics.
    For discussion of the relation between  Gibbsian and Boltzmannian statistical mechanics
     see \citet{lavis5} and \citet{Frigg&W3}.}
     Assuming that the systems to which the theory is applied are large leads us to the next layer, \emb{large statistical mechanical systems} (labelled $\textsf{SM2}$). This  marks a branching point in the structure of the theory: three different additions can be made to $\textsf{SM2}$, resulting in three different branches. Adding the thermodynamic limit to $\textsf{SM2}$ leads to the \emb{statistical mechanics of infinitely large systems} (labelled $\textsf{SM3}$). Adding renormalization group techniques to $\textsf{SM2}$ leads to the  \emb{renormalization group approach to statistical mechanics} (labelled $\textsf{SM4}$).
 Finally, adding the analysis of phase transitions for finite systems\footnote{Where, as described in Sect.\ \ref{ourp}, phase transitions are defined   in a way which
  avoids the involvement of singularities.} to $\textsf{SM2}$ leads to  the \emb{statistical mechanics of finite-system phase transitions} (labelled $\textsf{SM5}$).

\vspace{0.2cm}

\noindent It is our aim in this work to keep the developments of thermodynamics and statistical mechanics as separate as possible,
 in order to make visible  the internal structure of each separate theory. However, as indicated above, on close examination
 it becomes evident that there are in fact some `messages', both implicit and explicit, sent  \emb{\underline{f}rom \underline{s}tatistical \underline{m}echanics} (\sfF\sfS\sfM), that is to say from  the microstructure, to thermodynamics, which provides  the macrostructure. These are spelled out in  \FSM--\refsf{fsm-01}, \FSM--\refsf{fsm-02}, \FSM--\refsf{fsm-03}, \FSM--\refsf{fsm-04}.   In the other direction  the  connecting relationships  \emb{\underline{f}rom \underline{t}hermo\underline{d}ynamics} (\sfF\sfT\sfD), labelled  \FTD--\refsf{ftd-01}, \FTD--\refsf{ftd-02}, \FTD--\refsf{ftd-03},
 identify quantities in statistical mechanics with thermodynamic variables. As we shall see \FSM--\refsf{fsm-01} also plays a role in the
connecting  process and can be seen as in dialogue
 with \FTD--\refsf{ftd-03}.  The remaining  interventions  \FSM--\refsf{fsm-02}, \FSM--\refsf{fsm-03}, \FSM--\refsf{fsm-04}, can be viewed as an aid to the
 clarification of a number of important issues. We discuss these links between elements of both theories in the appropriate places in Sects.\ \ref{therm} and \ \ref{statmech}.

\vspace{0.2cm}

\noindent Much of the recent interest in the relationship between thermodynamics and statistical mechanics has concentrated on PTCP.
 It is the second aim of this paper to revisit the issue of PTCP in the light of our analysis of the internal structure of the two theories
  and their interrelations.  Doing so will lead us to some unexpected, and we think important, conclusions.

 In the modern theory of critical phenomena, dating from the middle of the 1960s,\footnote{For an historic account see \citet{domb1}.}
 critical exponents, which classify the type of singular behaviour in the approach to a critical region, play an important role.
 In our development of thermodynamics in Sect.\ \ref{therm} scaling theory is the final destination with scaling laws relating
  these critical exponents.  However, as already indicated and as described below, thermodynamics is a structured shell into
  which particular models are embedded, either by the assumption of a phenomenological
form for the entropy function or from statistical mechanics.  In the absence of such an embedding it is not possible
 to calculate values for critical exponents, nor to discuss universality.
 This is the idea  \citep{2-k2} that all critical situations\footnote{Of which there may be more than one in any model.}
  can be divided into universality classes, characterized by the values of their critical exponents and
differentiated by  a small number of properties of which the most important are  the (physical) dimension
 $d$ of the system and  the symmetry group of the  order parameter.
The first, but not the second, of these plays an important role in our
discussions,\footnote{For an account of the role of the order parameter in critical phenomena see, for example,
\citet[][Sect. 1.2]{2-b47}.} in particular in  the case of the Ising model, which we shall use
as an illustrative example throughout this work.
This, the most  well-known and thoroughly investigated model in the statistical mechanics of lattice systems,
 is briefly described in Appen.\ \ref{tim}. With the list of critical exponents given there for $d\cequals2$,
  $d\cequals3$ and $d\ge 4$, it provides an example of the dependence of these exponents and hence
   the universality class on the dimension of the system. The dimension $d$ is also of importance, in our
    discussion of scaling theory in Sect. \ref{scathr}, of finite-size scaling in Sect.\ \ref{fssc}
and of phenomenological renormalization in Sect.\ \ref{wtfn} (c).

\vspace{0.2cm}

\noindent  These observations concerning universality classes  together  with the inter-theory connecting relationships  \FSM--\refsf{fsm-02}, \FSM--\refsf{fsm-03}, \FSM--\refsf{fsm-04},
  provide the impetus to investigate, and clarify a number of important issues relating to PTCP.  These are (not necessarily
  in the order in which they arise in the discussion):
\begin{enumerate}[(i)]
\item Are infinite systems really  necessary in thermodynamics or statistical mechanics and:
\begin{enumerate}[(a)]
\item If so what for?
\item  If they are, is this solely because extensivity is not exactly true  in most cases in statistical mechanics?
\item Is the thermodynamic limit irrelevant to thermodynamics or has it already  been
implicitly applied?\footnote{It is an interesting observation that discussions of PTCP in the context purely of thermodynamics
 (e.g. \citeauthor{1-p1}, \citeyear{1-p1}, Chap. 9; \citeauthor{2-b2}, \citeyear{2-b2}) rarely if ever feel the need to
invoke or even refer to the thermodynamic limit.}
\item Is the thermodynamic limit in statistical mechanics necessary for the implementation of the procedures of the renormalization group?
\item Is there a  meaningful way to represent PTCP in finite systems?
\end{enumerate}
\item Given that, in thermodynamics, critical behaviour involves discontinuities in densities and singularities in response functions,
 is this necessarily still the case in statistical mechanics?
 \item Are the ideas of enrichment and substantiation helpful in describing the relationship between thermodynamics and statistical mechanics?
\item Where do reduction and emergence feature in the accounts of the relationship between thermodynamics and statistical mechanics?
\end{enumerate}
 As indicated, in the title of this work and by the progression between levels in the statistical mechanical column in Fig.\ \ref{Fig1},  we will discuss these issues with a special focus on
 large systems and infinite systems. In particular we shall address the question
  as to where realism is to be found,  in the study of  large systems, because real systems are finite but large (in the sense that they typically have $\sim 10^{23}$
constituents), or in the thermodynamic limit of an infinite system,  because singular behaviour (in susceptibilities and compressibilities) is believed to be experimentally
 observed, and in theories this arises only in the thermodynamic limit.
 This broad categorization of  large systems  is refined in Sect.\ \ref{ourp}.
The process of taking the thermodynamic limit is the determination of the asymptotic properties of a system as it becomes infinitely large.  In general this will involve
taking $d$ limits in each of the linear dimensions of the system and such a $d$-dimensionally infinite system, which where appropriate we call a \emb{fully-infinite system}, is implicitly the object of investigation by scaling theory in Sect.\ \ref{scathr}.\footnote{ \label{exthlt}Underlying this description is, of course, the question of the {\em existence} of the thermodynamic limit
and whether it depends on the boundary conditions  of the erstwhile finite system. For a discussion of these questions
see, for example, \citet{griff3}, \citet[][10--41]{griff2}  and  \citet[][Chaps.\  2 and 3]{1-r2}.}
 However, relevant to our discussions is the case of a \emb{partially-infinite system}, where the limit is taken in only $\frakd<d$ dimensions.
  Here it is $\frakd$ rather than $d$ which should count for the critical behaviour as the dimension of the system.
The idea underlying our approach to PTCP is that reality lies with \emb{fully-finite systems} ($\frakd=0$) and that the judgment as to whether
the large system will show behaviour which in practical terms is indistinguishable from singular behaviour is based
on comparing the behaviour of systems of ever increasing size to see whether their properties indicate  convergence towards those of the infinite
system.  In principle, as described in Sect.\ \ref{ourp} this limiting process is in all $d$ dimensions.  In practice, as we see in our discussion of $d=2$ transfer matrix calculations
in Sect.\ \ref{extls}, it also has relevance to the case where one limit has already been taken and increasing size is in the remaining dimension.

\vspace{0.2cm}

 \noindent Thus, as we have indicated, Sects.\ \ref{therm} and \ref{statmech} trace the steps in our developments of
 thermodynamics and statistical mechanics with the inter-theory connections between them; with  Sect.\  \ref{tlnism} addressing
 different proposed resolutions to the contradiction between the finiteness of real systems and the perceived necessity of phase
 transitions being portrayed as singularities in infinite systems. Sect.\ \ref{ptfs} discusses the proposal of \citet{main2}
 for representing the occurrence of phase transitions in finite systems.  Using the account of finite-size scaling in Sect.\ \ref{fssc}
 we propose in Sect.\ \ref{ourp} our alternative quantitative account  for phase transitions in finite {\em large} systems.
 Sect.\ \ref{emred} contains some after-thoughts on enrichment, substantiation, reduction and emergence and  our conclusions are
 in Sect.\  \ref{concl}.
%----------------------------------------------------------------------------------------------------------------------------------
\section{From Classical Thermodynamics to Scaling Theory}\label{therm}
%--------------------------------------------------------------------------------------------------------------
Accounts of thermodynamics range from those designed for the practical needs of engineers to those which aim for a degree of mathematical rigour. However, all
 share some common features and assumptions some of which are at variance with the  insights gained in statistical mechanics.  As indicated above, we flag these differences in the form
 of messages from statistical mechanics (\FSM--\refsf{fsm-01} to \FSM--\refsf{fsm-04}).
 %------------------------------------------------------------------------------------------------------------------------------------------------------------
\subsection{The Structure of Thermodynamics}\label{tstruth}
%----------------------------------------------------------------------------
All accounts of thermodynamics contain (in some form or another)  the first law, which establishes the existence of the internal energy function $U$ and the second law which
  establishes the existence of the entropy $S$ and temperature $T$. Details are not necessary for the present discussion.
   The only thing we need to carry forward is the fundamental thermodynamic differential
form.  Given a thermodynamic system with:
\begin{enumerate}[(i)]
\item One mechanic extensive/intensive\footnote{The extensive variables $U$, $X$ and $N$ scale with the size of the system, intensive variables $T$, $\xi$ and $\mu$ are invariant with respect to such scaling.  }
  conjugate variable pair $(X,\xi)$, where $X$ could stand for the volume $V$ or magnetic moment $\pM$ with conjugate
intensive variables,  which in the case of $V$ is the (negative) pressure $-P$ and  in the case of $\pM$ is the magnetic field $\pH$;
\item A (dimensionless) extensive variable $N$ which counts the number of units of mass in the system with a conjugate (intensive) energy $\mu$, called the chemical potential, carried by
 each unit of mass;\footnote{In most presentations of thermodynamics $N$ is simply taken to be the number of  particles in the system. Our usage is designed to avoid
 reference to the microstructure of the system and to allow $N$ to have non-integer values.}
\end{enumerate}
 for a differential change in the space  $\Xi_0$ of the variables $(U,X,N)$ the differential change in the entropy $S$ satisfies\footnote{At this point it is convenient:
\begin{enumerate}[(i)]
\item To clarify the dimensionality of the thermodynamic variables. It is straightforward
to show that, by scaling with respect to suitable constants,  $T$ and $\xi$ can be made of the
 dimensions of energy ($\rmJ\cequals\rmm^2\rmk\rmg\,\rms^{-2}$) and $U$, $S$ and $ X$
  made dimensionless. In the case of $U$ this is achieved by factoring out an energy constant $\varepsilon>0$. This is
the field--extensive variable representation of \citet[][Sect. 1.1]{lavisnew}, where scaling for $S$ and $T$ is effected using Boltzmann's constant $k_\tB$.
The further change to the coupling--extensive variable representation is achieved by taking ratios of $\varepsilon$, $\xi$ and $\mu$ with respect to $T$ as shown in (\ref{althsp1a}).
\item To observe that the generalization to more than one mechanical variable pair is straightforward.
\item To emphasise that this differential form should not be understood as some sort of equilibrium process in the space $\Xi_0$ \citep{Lavis-ep}.
\end{enumerate}}\label{fevr}
\begin{equation}
\rmd S=\zeta_1\rmd U -\zeta_2\rmd X -\zeta_3\rmd N,
\label{althsp1}
\end{equation}
where
\begin{equation}
\zeta_1\cequals\varepsilon/T,\tripsep\zeta_2\cequals\xi/T,\tripsep\zeta_3\cequals\mu/T,
\label{althsp1a}
\end{equation}
are couplings. It is clear that the couplings are intensive {\em and} dimensionless.
 That the  variables $(U,X,N)\in\Xi_0$ appear as differentials on the right of (\ref{althsp1}) should be understood as signifying that they are independent
variables. This means that the system is thermally, mechanically   and chemically isolated with $U$, $X$ and  $N$ fixed by an experimenter.
 Legendre transformations can be used to replace $U$ and $X$ successively as independent variables
by $\zeta_1$ and $\zeta_2$. Firstly, with Helmholtz free energy
\begin{equation}
\Phi_1\cequals \zeta_1 U-S,
\label{althsp30}
\end{equation}
we have
\begin{equation}
\rmd\, \Phi_1=U\rmd\zeta_1+\zeta_2 \rmd X +\zeta_3 \rmd N,
\label{althsp40}
\end{equation}
so that the independent variables are $(\zeta_1,X,N)\in\Xi_1$.  The system is in contact with a source of thermal energy at temperature $T=\varepsilon/\zeta_1$.
Secondly, with Gibbs free energy
\begin{equation}
\Phi_2\cequals \zeta_1  U-\zeta_2 X-S,
\label{althsp31}
\end{equation}
we have
\begin{equation}
\rmd\, \Phi_2=U\rmd \zeta_1- X\rmd\zeta_2  +\zeta_3 \rmd N,
\label{althsp41}
\end{equation}
so that the independent variables are $(\zeta_1,\zeta_2,N)\in\Xi_2$. The system is now, through $\zeta_2$, also in mechanical contact with its environment, be it a fluid system subject to a pressure $P$
or a magnetic system subject to a field $\pH$.  The couplings $\zeta_1$ and $\zeta_2$ are referred to as the {\em thermal and field (or mechanical) couplings} respectively.

It is tempting to suppose that this process could be taken one step further, interchanging the roles of $N$ and $\zeta_3$.  However, it is not difficult to see that the Legendre
transformation implementing this would involve a free energy $\Phi_3$ which is constant and can thus without loss of generality be taken to be identically zero.
 A viable form of thermodynamics must retain (at least) one extensive variable (here we choose that to be $N$, although we could have used $X$) which registers the size of the system.

Observing that in thermodynamics the uncontrolled variables remain  constant when the corresponding  controlled variables are held constant,
this is now the point for the first message  from statistical mechanics:
%--------------------------------------------------------------------------------------------------
\[\mbox{\mydbox{
\vspace{-0.3cm}
  \begin{fsm}
  \spacecol
Unlike in thermodynamics, extensive variables in statistical mechanics
that  are uncontrolled quantities {\em fluctuate} even when the corresponding controlled  variables are kept constant.
(In $\Xi_1$ the energy corresponding to the internal energy $U$ fluctuates, and in $\Xi_2$ the variable corresponding to $X$,
be it the volume or the magnetic moment, fluctuates. This is  born out by experiment \citep{MacD}.)
The variances of the fluctuations are given in terms of response functions and are $\mcO(N)$. This means that standard deviations of fluctuations
are $\mcO(\sqrt{N})$ and become negligibly small compared to $\mcO(N)$ variables only in the thermodynamic limit $N\to \infty$.
\end{fsm}
\vspace{-0.2cm}
}}\]
\customlabel{fsm-01}{\arabic{fsm}}
%-------------------------------------------------------------------------------------------------------------------------
  For fixed $N$ let $(U,X,N)\ato(U',X',N)$ denote an adiabatic process.    It can be shown \citep{Lavis-nt}, from
\citeauthor{cara1}'s first version of the second law\break \citep{cara1},\footnote{Or from some extensions to the approach of \citet{l&y2}.}
 that thermodynamic systems are of four types according to whether  the adiabatic process gives $U\le U'$ or   $U\ge U'$ and $S\le S'$  or $S\ge S'$
corresponding,  respectively, to the possibilities of  the temperature and heat capacity being positive or
  negative.\footnote{The Clausius version of the second law needs modification to include negative temperatures \citep{ramsey1,lands3}
and both the Kelvin-Planck and Clausius versions need modification to accommodate negative heat capacities (\citeauthor{Lavis-nt}, op. cit.).}
 Standard accounts of thermodynamics concentrate solely on the case where both internal energy and entropy increase, which is
  the situation where both temperature and heat capacity are positive. We shall restrict out attention to that case.\footnote{
 It is a matter of dispute \citep[see, for example,][and references therein]{Lavis-nt} whether statistical mechanical models support the
existence of negative temperatures, and the experimental evidence is also questioned.  The same is the case for negative heat capacities.}
%-------------------------------------------------------------------------------
\subsection{Extensivity and the Thermodynamic   Limit}\label{eatl}
%----------------------------------------------------------------------------------------------------------------------------------
Departing from the formulation \textsf{TD1} of the structure of thermodynamics we  ascend the left-hand column in Fig.\ \ref{Fig1},
where it is now useful to consider the embedding of particular models.
In this context they are of two types, ones which posit a phenomenological equation of state and ones derived from some microstructure according to the procedures
of statistical mechanics.  Most examples in the first category,  the perfect gas equation, the Weiss-field equation for ferromagnetism and the van der Waals equation\footnote{And a number
of lesser known relationships like the   Redlich--Kwong   and  Dieterici equations of state.}
 introduce the models in terms of an equation relating the mechanical variable pair $(X,\xi)$ and $N$ to the temperature.  However, it is more consonant with our approach to begin with a {\em defining} relationship for the entropy
surface $S(U,X,N)$, from which $T$, $\xi$ and $\mu$, or equivalently the couplings $\zeta_1$, $\zeta_2$ and $\zeta_3$  can be calculated using (\ref{althsp1}). Thus:
\begin{enumerate}[$\bullet$]
\item For the \emb{perfect gas}
\begin{equation}
S(U,V,N)\cequals Nc+\tfrac{3}{2}N\ln\Big(\ssfrac{U}{N}\Big)+N\ln\Big(\ssfrac{V}{N}\Big),
\label{prefl2}
\end{equation}
for some constant $c$,\footnote{Which can be evaluated using statistical mechanics but whose value is unimportant here.}
giving\footnote{Remember that $\zeta_2\cequals -P/T$.}
\begin{equation}
  T=\ssfrac{2U\varepsilon}{3N},\pairsep P=\ssfrac{NT}{V}.
\label{prefl1}
\end{equation}
\item For the \emb{van der Waals fluid}
\begin{equation}
S(U,V,N)\cequals Nc+\tfrac{3}{2}N\ln\Big(\ssfrac{U}{N}+\ssfrac{ N}{V}\Big)+N\ln\Big(\ssfrac{V}{N}-1\Big),
\label{tperdl11}
\end{equation}
giving
\begin{equation}
T= \ssfrac{2}{3}\varepsilon\Big(\ssfrac{U}{N}+\ssfrac{N}{V}\Big),\tripsep
P= \ssfrac{NT }{V- N}-\ssfrac{\varepsilon N^2}{V^2}.
 \label{tperdl13}
\end{equation}
\end{enumerate}
The entropy (\ref{prefl2})  is a concave function of $(U,V)$,  but for (\ref{tperdl11}) it is necessary to take the concave envelope.  This is, of course,
equivalent in the case of the \citet{1-v2} fluid and other phenomenological equations of state to the application of Maxwell's equal areas rule
\citep{3-m1}, which avoid the inclusion of unstable states and leads to a first-order gas-liquid phase transition (see Sect.\ \ref{scalt}).

It will be noted that, for both the perfect gas and van der Waals fluid with densities $u\cequals U/N$ and $v\cequals V/N$,
there exists an entropy density $s$ satisfying
 \begin{equation}
 s\cequals\frac{S(u N,vN,N)}{N}= s(u,v),
\label{entdenc}
\end{equation}
for all $N>0$, which avoids any reference to the size $N$ of the system.
But, of course,  these are rather special models and the question arises as to whether entropy, in general,
when $X$ replaces $V$ and $x\cequals X/N$ replaces $v$, satisfies
 \begin{equation}
 s\cequals\frac{S(u N,{x} N,N)}{N}= s(u,{x}),\pairsep\forall\vsmallsep N>0.
\label{entden}
\end{equation}
For this question the following result is important:
\begin{theorem}\label{exttl}
\spacecol\,\,
(\ref{entden}) is true iff
\begin{equation}
S(\lambda U,\lambda X,\lambda N)=\lambda S(U,X,N),\pairsep\forall\vsmallsep \lambda>0,
\label{extra1.1}
\end{equation}
is true.
\end{theorem}
\begin{proof}
That (\ref{entden}) follows from (\ref{extra1.1}) is easily seen by taking $\lambda\cequals 1/N$ and defining $s(u,x)\cequals S(u,x,1)$.

\vspace{0.2cm}

\noindent In the reverse direction, this last relationship $s(u,x)=S(u,x,1)$ in fact follows from (\ref{entden}) by setting $N=1$. Then
from (\ref{entden})   $S(U,X,N)= NS(U/N,X/N,1)$  and again setting $\lambda\cequals 1/N$ recovers (\ref{extra1.1}).
\qed
\end{proof}
Equation (\ref{extra1.1}) is the condition  that $S$ is an extensive function and it is easily shown from (\ref{althsp30}) and (\ref{althsp31}) that the free energies
$\Phi_1$ and $\Phi_2$ are extensive functions if and only if the entropy is an extensive function. But, as pointed out by  \citet[][Sect.\  2]{menon}
and show in Sect.\ \ref{extls},
%------------------------------------------------------------------------------------------------------------------------
\[\mbox{\mydbox{
\vspace{-0.3cm}
 \begin{fsm}
  \spacecol
The extensivity of entropy and of free energies assumed in thermodynamics is not
 exactly true for all systems in statistical mechanics, but is approximately true for large systems.
\end{fsm}
\vspace{0.2cm}
}}\]
\customlabel{fsm-02}{\arabic{fsm}}
%----------------------------------------------------------------------------------------------
For entropy the thermodynamic limit in statistical mechanics, assuming  it exists,\footnote{See footnote \ref{exthlt}.} is given by
\begin{equation}
\lim_{N\to\infty}\frac{S(u N,{x} N,N)}{N}= s(u,{x}).
\label{entdenl}
\end{equation}
 But for thermodynamics the corresponding formula is (\ref{entden}), without the need for the limiting process. Exact extensivity
in thermodynamics can be regarded as unnecessary or trivially true.

\vspace{0.2cm}

\noindent  Differentiating (\ref{extra1.1})  with respect to $\lambda$, and substituting from (\ref{althsp1}) gives
\begin{equation}
S=\zeta_1 U -\zeta_2 X -\zeta_3 N,
\label{efe2}
\end{equation}
when $\lambda$ is put equal to 1. From (\ref{althsp1}) and (\ref{efe2}),
\begin{equation}
u \rmd \zeta_1 -{x}\rmd \zeta_2 -\rmd \zeta_3 = 0\, ,
\label{efe3}
\end{equation}
which is a version of  the \emb{Gibbs-Duhem relationship}. In terms of densities (\ref{efe2}) becomes
\begin{equation}
s=\zeta_1 u -\zeta_2 {x} -\zeta_3,
\label{efe2x}
\end{equation}
and  substituting into (\ref{althsp1})--(\ref{althsp41})
{\jot=0.25cm
\begin{eqnarray}
 \rmd s&=&\zeta_1\rmd u-\zeta_2\rmd {x}-(s-\zeta_1 u+\zeta_2 {x}+\zeta_3)\rmd N/N\nonumber\\*
 &=&\zeta_1\rmd u-\zeta_2\rmd {x}.
 \label{dens1}
 \end{eqnarray}
 Then, for free-energy densities $\phi_1\cequals \Phi_1/N$ and $\phi_2\cequals \Phi_2/N$,
  \begin{eqnarray}
  \phi_1&=&\zeta_1 u-s=\zeta_2 {x} +\zeta_3,\smallsep
 \rmd\phi_1=u\rmd\zeta_1+\zeta_2\rmd {x},
 \label{dens2}\\
\phi_2&=&\zeta_1 u-\zeta_2{x}-s=\zeta_3,\smallsep
\rmd\phi_2=u\rmd\zeta_1-{x}\rmd\zeta_2.
 \label{dens3}
 \end{eqnarray}
These are the fundamental  size-free  thermodynamic  relationships in terms of density variables and density functions. They are exact in
 thermodynamics but  approximately true only for large systems in statistical mechanics.
The question of large systems and the thermodynamic limit in statistical mechanics is treated in Sects.\ \ref{extls},  \ref{tlnism} and \ref{ourp}.
%-------------------------------------------------------------------------------------------------------------------------
\subsection{Thermodynamics with PTCP}\label{scalt}
Having arrived at a formulation of thermodynamics in terms of densities and couplings the modern theory of  PTCP is
 largely concerned with an investigation and classification of the singular properties of  systems \citep[see e.g.][]{2-b2}.
 Specifically the singularities which could occur  on the  hypersurface of the entropy density, or the appropriate free-energy density, which defines the state
 of the system. However we should be forewarned that the account of statistical mechanics in Sect.\  \ref{statmech} concludes that:
 %------------------------------------------------------------------------------------------------------------------------
\[\mbox{\mydbox{
\vspace{-0.3cm}
 \begin{fsm}
  \spacecol
  The association of PTCP with  singularities in the entropy and free-energy densities which is made in thermodynamics can be made   in statistical mechanics only for infinite systems.
\end{fsm}
\vspace{0.2cm}
}}\]
\customlabel{fsm-03}{\arabic{fsm}}
%----------------------------------------------------------------------------------------------
\noindent The association of PTCP with singularities in both \textsf{TD3} and \textsf{SM3}
leads to a tendency for them to be mistakenly conflated. (We shall discuss this in more detail in relation to limit reduction in Sect.\ \ref{limred}.)

\vspace{0.2cm}

\noindent    We now consider  three thermodynamic spaces,   $\tXi_0$,  $\tXi_1$ and  $\tXi_2$,  which correspond respectively to the spaces  $\Xi_0$,  $\Xi_1$ and $\Xi_2$
defined in Sect.\ \ref{tstruth}     except that now densities
   replace extensive variables.  In reverse order, since this is more heuristically  transparent:
 %-------------------------------------------------------------------------------------------------------------------------------------------
\begin{figure}{t}
  \includegraphics[width=4cm]{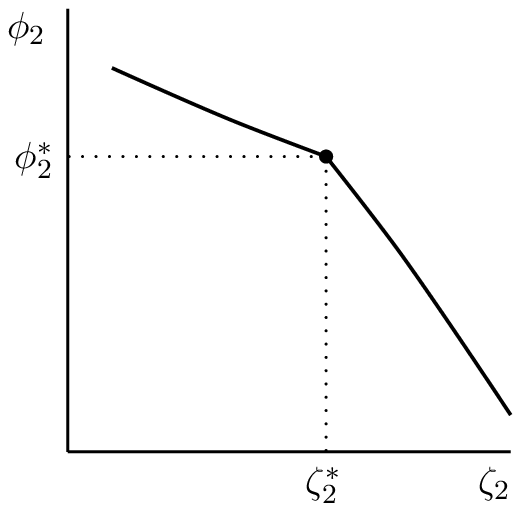}
\caption{A first-order
transition showing as a discontinuity of slope in an isothermal section ($\zeta_1=\zeta_1^\star$)
of  $\phi_2=\zeta_3$  plotted against $\zeta_2$.}\label{Fig2}
\includegraphics[width=4cm]{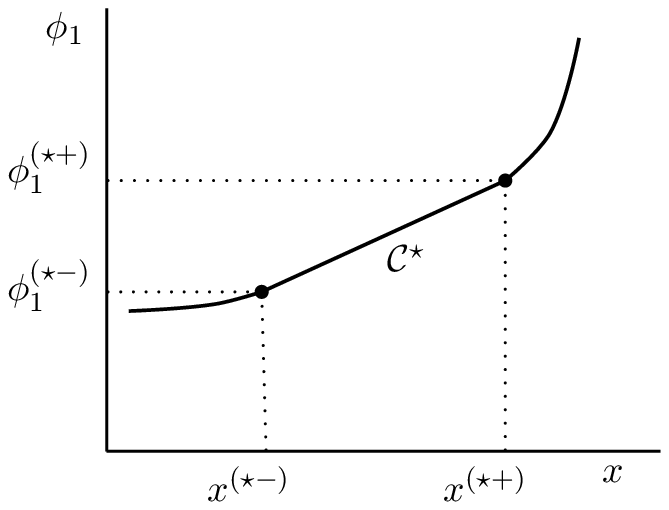}
\caption{A first-order transition showing as the linear section $\mcC^\star$  in an isothermal section
of the $\phi_1$ surface.}\label{Fig3}
  \includegraphics[width=4cm]{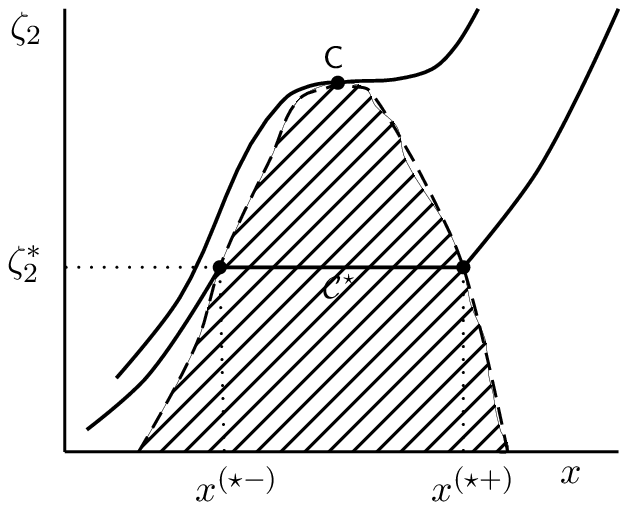}
\caption{A first-order transition showing as a horizontal part $\mcC^\star$ of an isotherm of
$\zeta_2$ plotted against ${x}$ together with the isotherm through the critical point $\sfC$.  As $\zeta_1$ varies the ends
of $\mcC^\star$ trace the boundary of the coexistence region (shaded). }\label{Fig4}
\includegraphics[width=6cm]{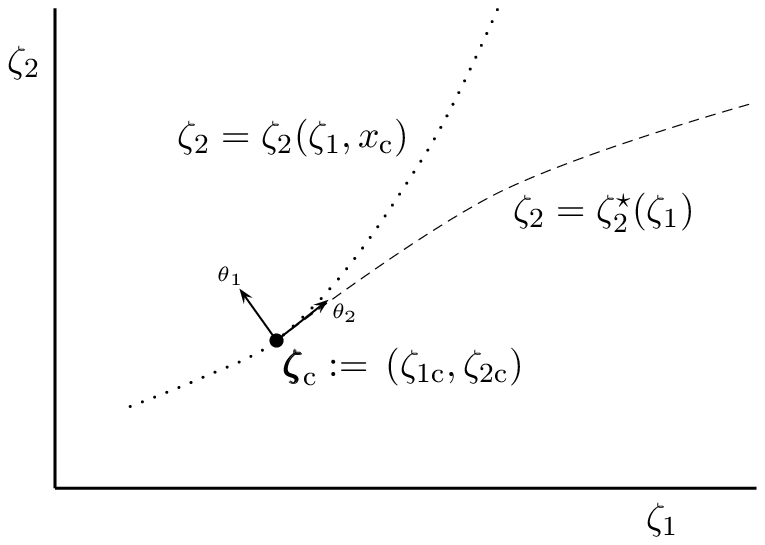}
\caption{A critical point $(\zeta_{1\rmc} , \zeta_{2\rmc} )$ in
$\tXi_2$. The first-order transition
(coexistence curve) $\zeta_2 = \zeta_2^\star (\zeta_1 )$
is represented by a broken line and the critical isochore, along which the density $x$ takes its critical value $x=x_\rmc$
 by a dotted line.  The directions of the axes of the two relevant scaling fields at the critical point, as described in Sect.\ \ref{scathr},
are shown. }\label{Fig5}
\end{figure}
%-----------------------------------------------------------------------------------------
   \begin{enumerate}[(i)]
  \item \underline{In the space $\tXi_2$  of the vector   $\bzeta\cequals(\zeta_1,\zeta_2)$}  the free-energy density $ \phi_2(\zeta_1,\zeta_2)$
is a surface with normal in the direction $(1,-u,{x})$ and
phases are separated by lines of transitions. The simplest example is a  line $\mcL^\star$ across which there is
 a discontinuity of the gradient $\nabla\phi_2=(u,-{x})$; an isothermal section ($\zeta_1$ constant) of this surface is shown in Fig.\ \ref{Fig2}.
The point $\bzeta^\star\cequals(\zeta_1^\star,\zeta_2^\star)\in\mcL^\star$, with $\zeta_3^\star=\phi_2(\bzeta^\star)$.
$\mcL^\star$  can be regarded as representing the \emb{coexistence of two phases} with different
densities.  As $\bzeta$ is varied across $\mcL^\star$  through $\bzeta^\star$  there is a
 \emb{first-order phase transition} where the densities change discontinuously.
 In the case of both fluid and magnetic systems a first-order transition will involve a discontinuity of the internal energy density $u$. In a fluid system there will be a discontinuity
 of (physical) density as the system changes between a liquid and a gas.  In a magnetic system there will be a discontinuity in the magnetization (or equivalently the magnetization density)
 as shown for the Ising model in Fig.\ \ref{Fig9}.
 \item \underline{In the space $\tXi_1$ of the vector $(\zeta_1,{x})$}  the free-energy density $\phi_1(\zeta_1,{x})$ is a surface convex
with respect to ${x}$ with normal   in the direction $(1,-u,-\zeta_2)$, as shown by an  isothermal ($\zeta_1=\zeta_1^\star$) section in Fig.\ \ref{Fig3}.
  A first-order transition corresponds to  the part of the isotherm, labelled $\mcC^\star$,   which is  linear with respect to  ${x}$.
 At the ends of   $(\zeta_1^\star,x^{(\star +})$  and  $(\zeta_1^\star,x^{(\star -)})$ of  $\mcC^\star$  all three couplings $\zeta_1$,  $\zeta_2$ and $\zeta_3$ have the same values
as is otherwise shown in Fig.\ \ref{Fig2}.  Typically, as $\zeta_1^\star$ varies along $\mcL^\star$ the ends of $\mcC^\star$ converge to a \emb{critical point} where
the system exhibits a \emb{second-order transition}. There the densities are continuous
but one or more of the response functions (that is to say the curvature components of the free-energy surface) is singular.\footnote{Such critical points can also occur as lines.
A line of first-order transitions can terminate on a line of second-order transitions at a point called a \emb{critical end-point},
or be continued as a line of second-order transitions  at a point   called a \emb{tricritical point}.} A projection
of the linear coexistence region in Fig.\ \ref{Fig3} is shown in Fig.\ \ref{Fig4}, and the situation where the  corresponding transition line $\mcL^\star$
terminates is shown in Fig.\ \ref{Fig5}.
\item \underline{The space $\tXi_0$ of the vector $(u,{x})$,} in which the entropy density $s(u,x)$ is a concave surface
is similar to that for $\phi_1(\zeta_1,{x})$,\footnote{That convexity is replaced by concavity is clear from the negative sign of $s$ in (\ref{dens2}).}
 except that now the linear generator $\mcC^\star$ of the coexistence region
has endpoints $(u^{(\star +)},x^{(\star +})$  and  $(u^{(\star -)},x^{(\star -)})$.
As $\bzeta^\star$ varies along $\mcL^\star$,  $\mcC^\star$ traces out the boundary of a
ruled\footnote{A  ruled surface (like, for example, the surface of a cylinder) is one densely covered by a set of straight lines.}
region  on the entropy surface with $\mcC^\star$ converging in one direction  to the critical point described in (ii).
 \end{enumerate}
\emb{Critical exponents} at the critical point are associated with
the curvature of the coexistence curve in $\tXi_1$ and the coexistence line in $\tXi_2$,
and the asymptotic singular behaviour
of the (per particle) heat capacities $c_{x}$ and $c_\xi$ at constant density and field respectively and  a response function
   $\varphi_{\vtmT}$, which in a fluid corresponds to the compressibility
  and in a magnet to the susceptibility. It will also be useful to include the coefficient of thermal expansion $\alpha_\xi$.
   These are defined together with their critical exponents in Appen. \ref{refnce}.
The heat capacities $c_{x}$ and $c_\xi$
are normally positive and from (\ref{cp43}) it follows that, if $\varphi_{\vtmT}>0$, then $c_\xi$ dominates
both $c_{x}$ and $\alpha^2_\xi/\varphi_{\vtmT}$ as ${T}\to {T}_{\rmc}$.
For the critical exponents $\cursigma$ and $\cursigma'$ characterizing the singularity of $c_x$ on approach to the critical point from
 above and below $T_\rmc$, and the analogously defined critical exponents $\curalpha$ and $\curalpha'$ characterizing the
 singularity of $c_\xi$, and $\curgamma$ and $\curgamma'$ characterizing the singularity of $\varphi_{\vtmT}$, as
  well $\curbeta$ characterizing the curvature of the coexistence curve, this means that
\begin{equation}
\cursigma \ge \curalpha\, ,\tripsep \cursigma' \ge
 \curalpha',\pairsep
\cursigma'+2\curbeta+\curgamma'\ge 2\, .\label{cp601}
\end{equation}
The condition $\varphi_{\vtmT}>0$ is true
for a magnetic system and
in this case the third  inequality in  (\ref{cp601})
 was first established by
\citet{2-r3}. The stronger condition
\begin{equation}
\curalpha'+2\curbeta+\curgamma'\ge 2\label{cp602}
\end{equation}
was obtained by \citet{2-g3} for both magnetic
and fluid systems
using the convexity properties of the free energy.
In fact it is a consequence of scaling theory (Sect.\ \ref{scathr}) that,
for systems with a special symmetry which is present in magnetic systems where, as for the Ising model in Appen.\ \ref{tim},
 the coexistence curve coincides with the zero field axis,
$\cursigma'=\curalpha'$ and inequalities
(\ref{cp601}) and (\ref{cp602}) become identical. Otherwise $\cursigma'=\curgamma'$.
\citet{2-g3}\label{2-g3-4-002} also derived a number of other inequalities.
In particular
\begin{equation}
\curgamma'\ge \curbeta(\curdelta-1),
\label{cp603}
\end{equation}
where $\curdelta$, given by (\ref{cp49}),  is the exponent characterizing the (critical) equation of state.
%------------------------------------------------------------------------------------
\subsection{Thermodynamics with Scaling Theory}\label{scathr}
In view of our aim to keep as distinct as possible the developments of thermodynamics and statistical mechanics, we choose here
to present  scaling theory as a mathematical axiomatization of the properties of PTCP in thermodynamics. Although, as we see below, it has deep
roots in, and is substantiated by, statistical mechanics, in particular renormalization group theory,\footnote{The assumption (\ref{eu3})   that, near to a critical region,  the
 free-energy density can be divided  into smooth and singular parts  is  justified in terms
of the form of the Hamiltonian  \citep[][p.\ 3519]{2-h3}.}   where, in almost all cases,\footnote{An exception
being the one-dimensional Ising model  \citep[see e.g.][Sect.\ 15.5.1]{lavisnew}.} the realization of this
picture of scaling  involves approximations and yields scaling forms of only local validity.

  Originating in the work of (among others)
\citet{2-w2,2-w1} and  \citet{2-k1} our approach is essentially that of  \citet{2-h3}. Given here in brief
 outline\footnote{For  a more detailed account see, for example, \citet[][Chap. 4]{lavisnew}.} it  is
sufficient for an analysis of power-law singularities in the critical region.\footnote{For statistical mechanical systems like
 the Ising model with $d = 2$ which exhibit logarithmic singularities, it has been shown by
\citet{2-n1}\label{2-n1-4-001} that a slight generalization  needs to be used.}

\vspace{0.2cm}

\noindent Suppose we have the free-energy density of a system in terms of its maximum number of independent couplings.  In the discussion above that maximum number
was two, but for the moment  we generalize to $n$ couplings so the free-energy density is $\phi_n(\bzeta)$, where $\bzeta\cequals(\zeta_1,\zeta_2,\ldots,\zeta_n)$,
which is represented as a hypersurface of dimension $n$  in the $(n+1)$-dimensional space $(\phi_n,\bzeta)$.
Now suppose that there is a critical region $\mcC$  of dimension $n-s$. Although $\phi_n(\bzeta)$
 itself is continuous and finite across and within $\mcC$ it may have discontinuous first-order derivatives, meaning that $\mcC$ is a region of phase coexistence
with a first-order transition when, as is shown in Fig.\ \ref{Fig2}, the phase point crosses through $\mcC$, or it may have singular second-order derivatives  in  $\mcC$,
as is the case in the situation described above where a line of first-order transition terminates at a critical point.\footnote{Or there may be discontinuities
or singularities in higher-order derivatives; but we shall for simplicity concentrate solely on cases involving first- and second-order derivatives.}

With respect to some origin  $\bzeta^\circ\in \mcC$    a system of  orthogonal curvilinear coordinates $\theta_1,\theta_2,\ldots,\theta_n$
called \emb{scaling fields} is constructed. These are  smooth functions of the couplings which parameterize $\mcC$ so that
$\theta_1 = \cdots = \theta_s = 0$ within $\mcC$.
The scaling fields in this subset are called \emb{relevant}  with those in the  remaining subset  $\theta_{s+1},\theta_{s+2},\ldots\theta_n$, called \emb{irrelevant},
acting as a local set of coordinates within $\mcC$.\footnote{ `Relevance' here refers to their role in an understanding of the nature of the criticality in $\mcC$.}
The free-energy density   $\phi_n(\bzeta)$  is separated into two parts
\begin{equation}
\phi_n(\bzeta) = \phi_{\smth}(\bzeta) + \phi_{\sing}(\triangle\bzeta)\, ,
\label{eu3}
\end{equation}
where $\phi_{\smth}(\bzeta)$  is a regular function and, with $\triangle\bzeta\cequals \bzeta-\bzeta^\circ$,   $\phi_{\sing}(\triangle\bzeta)$, for which     $\phi_{\sing}(\bzero)=0$,
contains  all the non-smooth parts of $\phi_n(\bzeta)$ in $\mcC$. It is now assumed  that
  $\phi_{\sing}(\triangle\bzeta)$ can be re-coordinated in terms of the scaling fields so
  that it is a generalized homogeneous function satisfying
the \emb{Kadanoff scaling hypothesis}\footnote{The physical dimension $d$ of the system  is not something which plays a significant role in most of thermodynamics. It is
included here to bring compatibility with the discussion of statistical mechanics. It could be removed by redefining $\lambda$.}
\begin{equation}
{\phi}_{\sing}({\lambda}^{y_1}{\theta}_1,\ldots,{\lambda}^{y_n}{\theta}_n) = {\lambda}^d {\phi}_{\sing}
({\theta}_1,\ldots ,{\theta}_n)\, ,
\label{egs4}
\end{equation}
for all real  $\lambda>0$, where  $d$ is the physical dimension of the system,
 and $y_j$, $j=1,2,\ldots,n$ are \emb{scaling exponents} satisfying
\begin{equation}
y_j > 0,\vsmallsep  j = 1,\ldots,s,\pairsep
y_j < 0,\vsmallsep   j = s + 1,\ldots,n.
\label{egs5}
\end{equation}
The exponents in the first subset are, like the corresponding scaling fields,  called \emb{relevant}\,
and those in the  latter subset are called \emb{irrelevant}.\footnote{
 We have for the sake of simplicity excluded  the possibility of a zero exponent;  such an exponent is called
\emb{marginal}.   Marginal  exponents are associated in renormalization group theory with an `underlying' parameter of the system,
often resulting in lines of fixed points as we see in our treatment of the one-dimensional Ising model in Sect.\ \ref{wtfn}.
 It will also be assumed that no exponent is complex. In practice this is not always the case
\citep[see e.g.][Sect.\ 15.5.2]{lavisnew}, but situations arising from complex exponents are
not difficult to interpret in particular examples.}
Of the assumptions made here, that scaling fields can be derived
 is not particularly demanding; at the very least it is usually straightforward to obtain their linear parts near to the origin.
And the division of the free-energy density (\ref{eu3}) into smooth and singular parts  has very little content until we explore in more detail the consequences of the scaling hypothesis (\ref{egs4})
which we now do for the case of a critical point terminating a coexistence curve.

\vspace{0.2cm}

\noindent There are many general accounts of scaling theory, treating a variety of critical phenomena. Here we restrict attention to
 the case of a critical point terminating a line of first-order transitions, as shown in Fig.\  \ref{Fig5}.  So we have two critical regions.
 The first is the critical point with two relevant scaling fields and scaling exponents  with axes chosen perpendicular to and along the coexistence curve.
For this we shall show that  the critical exponents defined in Appen.\ \ref{refnce}, can be expressed in terms of the two scaling exponents.
 The second is the coexistence curve which has one relevant and one irrelevant scaling field constructed with respect to some chosen origin (not shown in Fig.\ \ref{Fig5}) on the coexistence curve.

  For the sake of further simplifying our presentation we restrict attention to a simple ferromagnetic system with $\xi\cequals \pH$, the magnetic field,  $X\cequals \pM$,
the magnetization and $x\cequals m=\pM/N$, the magnetization density. The coupling $\zeta_1$ is the thermal coupling so we relabel it as $\zeta_\vtmT=\varepsilon/T$ and $\zeta_2$ is the field coupling which we relabel as $\zeta_{\vtpH}=\pH/T$.
 This model, of which an example in statistical mechanics is the Ising model described in Appen.\ \ref{tim},
  has the advantage of having the special symmetry that the coexistence curve lies along the zero-field axis in an interval $T\in[0,T_\rmc]$
with $\pH_\rmc=m_\rmc=0$. This axis with $T>T_\rmc$ is the critical isochore.
Thus (referring to Fig.\ \ref{Fig5})  the  coexistence curve lies along the $\zeta_\vtpH=0$ axis in an interval   $[\zeta_{\vtmT\rmc},\infty)$.
This same phase diagram for the Ising model, now plotted with respect to the temperature $T$ and the magnetic field $\pH$, is shown in Fig.\ \ref{Fig8}.

\vspace{0.2cm}

\noindent We consider separately the critical point and the coexistence curve,
beginning with the critical point where  we can take the scaling fields to be
\begin{equation}
 \theta_\vtmT\cequals\zeta_\vtmT-\zeta_{\vtmT\rmc}=\varepsilon\left(\ssfrac{1}{T}-\ssfrac{1}{T_{\rmc}}\right)\ge 0,
 \pairsep
\theta_\vtpH\cequals\zeta_\vtpH=\ssfrac{\pH}{T}.
\label{scflh}
\end{equation}
The scaling hypothesis (\ref{egs4}) becomes
\begin{equation}
{\phi}_{\sing}(\lambda^{y_\vtmT}\theta_\vtmT,\lambda^{y_\vtpH}\theta_\vtpH) = {\lambda}^d {\phi}_{\sing}(\theta_\vtmT,\theta_\vtpH)\, ,
\label{egs4m}
\end{equation}
and, from  (\ref{eu3})  and (\ref{cpden}),
\begin{eqnarray}
m=- \frac{\partial  \phi_{\smth}}{\partial \zeta_\vtpH} & -& \frac{\partial \phi_{\sing}}{\partial\theta_\vtpH},
\label{egsx1}\\
\frac{\partial{\phi}_{\sing}}{\partial\theta_\vtpH}({\lambda}^{y_\vtmT}{\theta}_\tmT,{\lambda}^{y_\vtpH}{\theta}_\vtpH) &=& {\lambda}^{d-y_\vtpH}
\frac{\partial{\phi}_{\sing}}{\partial\theta_\vtpH}({\theta}_\tmT,{\theta}_\vtpH).
\label{egs4x}
\end{eqnarray}
Since $m_\rmc=0$,  ${\partial  \phi_{\smth}}/{\partial \zeta_\vtpH}=0$ at the critical point.  For an approach to the critical point   along the coexistence curve $\theta_\vtpH=0$ and setting
$\lambda\cequals \theta_\vtmT^{-1/y_\vtmT}$ in (\ref{egs4x}) and substituting into (\ref{egsx1}) gives
\begin{equation}
m \simeq -\theta_{\vtmT}^{(d-y_\vtpH)/y_\vtmT}\frac{\partial{\phi}_{\sing}}{\partial\theta_\vtpH}(1,0) \sim (T_\rmc-T)^{(d-y_\vtpH)/y_\vtmT},
\label{betai1}
\end{equation}
which, when comparing with (\ref{cp48}) establish the identification
\begin{equation}
\curbeta=(d-y_\vtpH)/y_\vtmT.
\label{betai2}
\end{equation}
At this point we could carry out a similar procedure for the response functions in (\ref{cdrsfn}) and (\ref{cp43}) to determine the critical exponents defined  in (\ref{cp60})--(\ref{cp55}).
However,  the analysis can be shortened by a closer examination of the way that the expression (\ref{betai2}) for $\curbeta$ was obtained. From this we see that the scaling exponent  $y_\vtpH$
in the numerator indicates  that differentiation was once with respect to $\zeta_\vtpH$. And that the approach was in the direction of varying $\zeta_\vtmT$
is indicated by the scaling exponent $y_\vtmT$ in the denominator. So with the same reasoning it follows from (\ref{cp49}) that
\begin{equation}
\curdelta= y_\vtpH/(d-y_\vtpH),
\label{betai3}
\end{equation}
and bearing in mind that the analysis yields singularities  for response functions so  $\phi_{\smth}$ can play no role, from (\ref{cp51}),
\begin{equation}
\curgamma=\curgamma^\prime=(2y_\vtpH-d)/y_\vtmT.
\label{betai4}
\end{equation}
When we come to consider $c_\xi\cequals c_\vtpH$, given by (\ref{cdrsfn}), the situation becomes a little more complicated, since there are three terms and we need to know which dominates
as the critical point is approached. This will depend on the relative magnitudes of $y_\vtmT$ and $y_\vtpH$ and it can be shown (\citeauthor{lavisnew}, ibid, Sect. 4.5.1) that, in general for a critical point terminating a line of first-transitions, the exponent associated with approaches tangential to the coexistence curve is smaller (less relevant) than that associated with an approach at a non-zero angle to this curve. These are called  respective  \emb{weak and strong approaches} and in the present context we have $y_\vtpH>y_\vtmT$, these being respectively the weak and strong exponents.
Returning to the formula for $ c_\vtpH$ in (\ref{cdrsfn})  we see that the third term on the right-hand side would be the one that dominates meaning that, from (\ref{cp55}),
$\cursigma=\cursigma^\prime= \curgamma$. However, because of the symmetry of the magnetic model $\zeta_{2\rmc}\cequals \zeta_{\vtpH\rmc}=0$ and the only remaining term is the first,
meaning that
\begin{equation}
\cursigma=\cursigma^\prime=(2y_\vtmT-d)/y_\vtmT.
\label{betai5}
\end{equation}
Finally we need to determine the asymptotic form for $c_x\cequals c_m$ using  (\ref{cp43}).  Here the situation need a more detailed analysis, when it can be shown
 (\citeauthor{lavisnew}, ibid, Sect. 4.5.4) that, {\em whether or not the magnetic symmetry applies}  cancellation of coefficients  leads to an asymptotic form equivalent to
 that of a second-order derivative with respect to $\zeta_\vtmT$; that is,
\begin{equation}
\curalpha=\curalpha^\prime=(2y_\vtmT-d)/y_\vtmT.
\label{betai6}
\end{equation}
This means that it is the asymptotic form of the heat capacity with constant intensive variable (pressure or magnetic field)  which is dependent on symmetry.
In the magnetic system the exponent is the same as that of the heat capacity with constant extensive variable (the magnetization) and in a fluid, where there
is no symmetry it is equal to that of $\varphi_{\vtmT}$, which is the compressibility.
Equations (\ref{betai2})--(\ref{betai6}) are formulae
for the exponents $\curalpha$, $\curbeta$, $\curgamma$
and $\curdelta$ in terms of $y_\vtmT$ and $y_\vtpH$. They are, therefore,
not independent and two relationships exist between them. These can be
expressed in the form $\curalpha + 2\curbeta + \curgamma= 2$,
called the \emb{Essam--Fisher scaling law}, which
is a strengthening of the inequality (\ref{cp602})
and $\curgamma'  =  \curbeta (\curdelta - 1)$,
called the \emb{Widom scaling law}, which is a strengthening of the inequality (\ref{cp603}).

For the coexistence curve,  scaling fields,  chosen with respect to some arbitrary origin $\zeta_\vtmT=\zeta^\circ_\vtmT$, $\zeta_\vtpH=0$   are
\begin{equation}
 \theta^\prime_\vtmT\cequals\zeta_\vtmT-\zeta^\circ_{\vtmT},
 \pairsep
\theta^\prime_\vtpH\cequals\zeta_\vtpH=\pH/T,
\label{scflhcc}
\end{equation}
with $y^\prime_\vtmT$ and $y^\prime_\vtpH$ irrelevant and relevant exponents respectively.
In general it can be shown that relevant exponents are less than or equal to $d$  meaning in this case  that
 $0<y^\prime_\vtpH\le d$.  With primes attached to the exponents and fields (\ref{egsx1})
 and (\ref{egs4x}) continue to applied to the magnetization density. If $y^\prime_\vtpH< d$
\begin{equation}
\frac{\partial{\phi}_{\sing}}{\partial\theta_\vtpH}(0,0)=0
\label{egs4x1}
\end{equation}
and  $m$ is continuous at  the origin; there is no first-order phase transition.
If $y^\prime_\vtpH = d$ then (\ref{egs4x1}) does not necessarily hold.
There may be a contribution to (\ref{egsx1}) from the
derivative of ${\phi}_{\sing}$.
This will be the only way in which the magnetization can be discontinuous across the coexistence curve.
So a scaling exponent equal to $d$ is a necessary, but not sufficient condition for a first-order
transition.   An example of such a first-order transition with an exponent of $d$ is at zero temperature in the one-dimension Ising model (Sect.\ \ref{wtfn}(a)).
Discontinuities in higher-order derivatives can be treated in a similar way.
%-------------------------------------------------------------------------------------------------------------------------------------
\subsection{Dimensionality and Phase Transitions}\label{dapt}
Although, as we have seen,  thermodynamics, and particularly its treatment of PTCP, assumes that the system is infinite, the dimension $d$
of the system entered into the discussion in Sect.\ \ref{scathr}.  And once dimensionality has entered then finiteness has also appeared.
Thus, for example, a two-dimensional system can be viewed as a three-dimensional system of   `thickness' one in the third dimension and it is only
a small step from there to increase the thickness to two.  In Sect.\ \ref{intro}  we referred to the classification of singularities in terms of universality classes.
This, as we  asserted, can be discussed
 only in the context of statistical mechanics, with $d$ one of the factors determining the universality
class of an occurrence of singular behaviour.  If the number of directions in which the system is infinite  is increased, then its critical behaviour
will change from one universality class to another.
This is an example of what in scaling and renormalization group theory is called `cross-over'.\footnote{Of course, such
a change of universality class is counter-factual \citep{HKT1}, in the sense that one cannot change the dimension or extensivity properties of a real system.}
The dimension of the system affects not just the universality class of singular behaviour but whether it occurs at all. However, that dimension is not $d$
but $\frakd\le d$, the number of directions in which the system is infinite.\footnote{The connection between the thermodynamic limit and extensivity is retained in a partially-infinite system
with $N_k$ sites in the $k$-direction and $N_1N_2\cdots N_d=N$,
when,  in the case, for example, of entropy, (\ref{extra1.1}) is replaced by
\[S(\lambda^\prime U,\lambda^\prime X,\lambda_1 N_1,\lambda_2N_2,\ldots,\lambda_\frakd N_\frakd, N^{(\frakd)})=\lambda^\prime S(U,X,N_1,N_2,\ldots,N_\frakd, N^{(\frakd)})\, ,\]
where $N^{(\frakd)}\cequals N_{\frakd+1}N_{\frakd+2}\cdots N_d$ and  $\lambda'\cequals \lambda_1\lambda_2\cdots \lambda_\frakd$.}
  And the final message sent from statistical mechanics to thermodynamics is that:
%-----------------------------------------------------------------------------
\[\mbox{\mydbox{
\vspace{-0.2cm}
  \begin{fsm}
  \spacecol
 There exists a {\em lower-critical dimension} $d_{\tL\tC}$ such that, if $\frakd\le d_{\tL\tC}<d$ singular behaviour can occur
in the fully-infinite system but not in the partially-infinite system. If $d>\frakd>d_{\tL\tC}$ then singular behaviour can occur in both, but in different universality classes.
\end{fsm}
\vspace{0.2cm}
}}\]
\customlabel{fsm-04}{\arabic{fsm}}
%---------------------------------------------------------------------------------------------------------------------------------
\section{From Gibbsian Statistical Mechanics to the Renormalization Group}\label{statmech}
The move from thermodynamics to statistical mechanics is, we shall argue, an {\em enrichment} and {\em substantiation} of the picture we have of
any system under investigation.  This operates at two levels. The first is structural, where renormalization group theory embedded
in statistical mechanics provides a fuller picture in terms of renormalization group transformations and fixed points than
scaling theory embedded in thermodynamics.  The second is in the provision of specific models which arise from assumptions about
the microstructure of the system.  We now consider the development represented by the right-hand column in
Fig.\ \ref{Fig1}, beginning with the basic structure of statistical mechanics.
%-----------------------------------------------------------------------------------------
\subsection{Inter-Theory Connecting Relationships}\label{brirel}
Let the microstate of the system be given by a value of the vector variable $\bsigma$ in the phase space $\Gamma$.
In the case of a fluid system $\bsigma$ will be a set of values for the positions and momenta of all the particles; for a spin system
on a lattice, like the Ising model in Appen.\ \ref{tim},  $\bsigma$ will be the set of values of all the spin variables.  The microscopic and macroscopic
structure of the system is then determined by the Hamiltonian. This is an explicit function of the independent couplings
with the independent extensive variables  imposing constraints on $\bsigma$.  Thus we have three cases:
\begin{enumerate}[(i)]
\item  When $(U,X,N)\in\Xi_0$  are the independent variables the Hamiltonian is $\hH_0(\bsigma;X,N)$, with values
constrained by
\begin{equation}
\hH_0(\bsigma;X,N)=U,
\label{microcan}
\end{equation}
and $\bsigma$ constrained, according to the nature of the particular model by $X$ and $N$.\footnote{$N$ fixes the number of particles in
$\bsigma$. If  the system is a fluid, with $X\cequals V$ the volume of it container,  then this constrains the range of the configuration component
of $\bsigma$.  Rather less physically achievable, if the system is a magnet,  with $X\cequals M$ the magnetization,  this will constrain
the spin configuration of the microsystems.}
\item  When $(\zeta_1,X,N)\in\Xi_1$  are the independent variables the Hamiltonian\break $\hH_1(\bsigma;\zeta_1,X,N)$ is a linear function of $\zeta_1$. The constraint (\ref{microcan})
is removed but $\bsigma$ remains constrained by $X$ and $N$.
\item  When $(\zeta_1,\zeta_2,N)\in\Xi_2$  are the independent variables the Hamiltonian\break $\hH_2(\bsigma;\zeta_1,\zeta_2,N)$ is a linear function of $\zeta_1$ and $\zeta_2$. The only
remaining constraint is from $N$.
\end{enumerate}
Connecting relationships are now invoked in three stages:
%-----------------------------------------------------------------------------------------------------------
\[\mbox{\mydbox{
\vspace{-0.3cm}
  \begin{ftd}
  \spacecol\,\,
The independent variables in $\Xi_0$, $\Xi_1$ and $\Xi_2$ are endowed with their thermodynamic meanings.
\end{ftd}
\vspace{-0.2cm}
}}\]
\customlabel{ftd-01}{\arabic{ftd}}
%------------------------------------------------------------------------------------------------------
 To proceed to the next stage of the inter-theory connecting process  we need to give a
form in cases (i), (ii) and (iii), respectively,  for   the entropy,  and the free energies $\Phi_1$ and $\Phi_2$.   Case (i) gives the \emb{microcanonical distribution}\footnote{Where  there is still
some dispute about the appropriate form for the entropy \citep[see, for example,][and references therein]{Lavis-nt}.} and  cases (ii)  and (iii) give, respectively, the \emb{canonical distribution}
and the \emb{constant pressure or  magnetic field  distribution}.   For the sake of simplicity we concentrate exclusively on case (iii), where the Gibbs free energy is {\em defined} by
\begin{equation}
\Phi_2(\zeta_1,\zeta_2,N)\cequals -\ln\{Z_2(\zeta_1,\zeta_2,N)\},
\label{gibfe}
\end{equation}
where
\begin{equation}
Z_2(\zeta_1,\zeta_2,N)\cequals \sum_{\{\bsigma\}} \exp\{-\hH_2(\bsigma;\zeta_1,\zeta_2,N)\},
\label{gibpf}
\end{equation}
is the Gibbs partition function.\footnote{Of course, according to the nature of the system this could involve an integral rather than a sum.}
Then
%-----------------------------------------------------------------------------------------------------------
\[\mbox{\mydbox{
\vspace{-0.3cm}
  \begin{ftd}
  \spacecol\,\,
$\Phi_2$ is endowed with its thermodynamic properties and, using
 (\ref{althsp30})--(\ref{althsp41}),
\begin{equation}
U=\frac{\partial \Phi_2}{\partial\zeta_1},\tripsep X=-\frac{\partial \Phi_2}{\partial\zeta_2},\tripsep \zeta_3=-\frac{\partial \Phi_2}{\partial N},\label{bX}
\end{equation}
\begin{equation}
 \Phi_1=\Phi_2+\zeta_2 X,\pairsep S=\zeta_1 U-\Phi_1,\label{bY}
\end{equation}
establishes the connection between $U$, $X$, $\zeta_3$,  $\Phi_1$ and $S$ and their thermodynamic equivalents.
\end{ftd}
\vspace{-0.2cm}
}}\]
\customlabel{ftd-02}{\arabic{ftd}}
%------------------------------------------------------------------------------------------------------
This completes a {\em sufficient} set of the connecting relationships. However, we can make some further links.  Suppose that
\begin{equation}
\hH_2(\bsigma;\zeta_1,\zeta_2;N)\cequals \hU(\bsigma)\zeta_1-\hX(\bsigma)\zeta_2.
\label{hamform}
\end{equation}
Then, from (\ref{gibfe})--(\ref{bY}),
\begin{equation}
U=\langle \hU(\bsigma)\rangle,\pairsep  X=\langle \hX(\bsigma)\rangle.
\label{exval}
\end{equation}
%-----------------------------------------------------------------------------------------------------------
\[\mbox{\mydbox{
\vspace{-0.3cm}
  \begin{ftd}
  \spacecol\,\,
$U$ and $X$ are identified respectively as the expectation values of $\hU(\bsigma)$ and $\hX(\bsigma)$ with respect to the probability distribution with
density
\begin{equation}
\rho(\bsigma;\zeta_1,\zeta_2)\cequals \frac{ \exp[-\hH_2(\bsigma;\zeta_1,\zeta_2;N)]}{Z_2(\zeta_1,\zeta_2,N)}.
\label{gibprob}
\end{equation}
\end{ftd}
\vspace{-0.2cm}
}}\]
\customlabel{ftd-03}{\arabic{ftd}}
And it further follows from (\ref{gibfe})--(\ref{gibprob}) that
\begin{equation}
\Var[\hX(\bsigma)]=\frac{\partial^2\Phi_2}{\partial \zeta_2^2} = N\varphi_{\vtmT},
\label{flures}
\end{equation}
where $\varphi_{\vtmT}$ is the response function given by (\ref{cdrsfn}).  This is an example of a {\em fluctuation--response function} relationship.  Similar
relationships apply to $\hU(\bsigma)$ and  all uncontrolled extensive variables.
%----------------------------------------------------------------------------------------------------------------------------------------------------------
\subsection{Correlation Function  and Correlation Length}\label{corcoln}
As is already evident, thermodynamics is a `black-box' theory  with a set of macro-variables some of which are independent and controllable and
others whose values change in response to the changes in the independent variables. The only concession made to internal structure was,
 in Sect.\ \ref{tstruth},  to allow a counting of the number $N$ of mass units of the system.
 Now with the `enrichment' provided by statistical mechanics we are able to record the microstate $\bsigma$ of the system,
which is simply the aggregate of  the states of the individual microsystems.

\vspace{0.2cm}

\noindent Suppose that  we take the $d$-dimensional hypercubic lattice\footnote{The restriction of this presentation to a hypercubic lattice is in the interests of simplicity.
It can easily be generalized to other lattices.}  $\mcN_d$ with sites $\br\cequals(n_1,n_2,\ldots,n_d)\fraka$,
for $n_k=1,2,\ldots,N_k$ with $N=N_1N_2\cdots N_d$,   where $\fraka$ is the lattice spacing.\footnote{It is convenient for our discussions to
suppose that the microsystems are confined to the sites of a lattice.  It is, of course, the case that a whole area of statistical mechanics concerned
with fluid systems treats the case of  microsystems/molecules moving in a continuum.}
Then, given that the states of the microsystems on sites $\br$ and $\br'$ of the lattice  are $\sigma(\br)$ and $\sigma(\br')$, respectively, how does the state
of one effect the state of the other; that is to say,  how are their states {\em correlated}?  More specifically, how is the correlation between $\sigma(\br)$ and $\sigma(\br')$  affected by:
\begin{enumerate}[(i)]
\item  the distance $|\br-\br'|$ between the sites?
\item  the closeness of the thermodynamic state of the system to a critical region?
\end{enumerate}
To begin to answer these questions suppose that, as for the   Hamiltonian (\ref{hamform})  in the Ising model in Appen.\    \ref{tim},
 $\hX(\bsigma)$ is a linear
sum of the states on the sites of $\mcN$. And (temporarily) suppose that the coupling $\zeta_2$ takes different values
$\zeta_2(\br)$ at the sites. Then, denoting the set of couplings $\zeta_2(\br)$ by the vector $\bzeta_2$,
\begin{equation}
\hH_2(\bsigma;\zeta_1,\bzeta_2;N)\cequals \hU(\bsigma)\zeta_1-\sum_{\{\br\}}\sigma(\br)\zeta_2(\br)
\label{hamforn}
\end{equation}
and from (\ref{gibprob}), the expectation values of $\sigma(\br)$ is
\begin{equation}
\langle\sigma(\br)\rangle = \sum_{\{\bsigma\}} \sigma(\br) \rho(\bsigma;\zeta_1,\bzeta_2)=-\frac{\partial\Phi_2}{\partial\zeta_2(\br)}.
\label{expsig}
\end{equation}
If the states $\sigma(\br)$ and $\sigma(\br')$ are uncorrelated  $\langle\sigma(\br)\sigma(\br')\rangle$ will factor into $\langle\sigma(\br)\rangle\langle\sigma(\br')\rangle$.
So
\begin{equation}
\curGamma(\br,\br';\zeta_1,\{\zeta_2(\br)\})\cequals \langle\sigma(\br)\sigma(\br')\rangle-\langle\sigma(\br)\rangle\langle\sigma(\br')\rangle
=-\frac{\partial^2\Phi_2}{\partial\zeta_2(\br)\partial\zeta_2(\br')},
\label{prcfn}
\end{equation}
called the \emb{pair correlation function} is a measure of the degree of correlation between  $\sigma(\br)$ and $\sigma(\br')$.
If  all the couplings $\zeta_2(\br)$ are set equal to $\zeta_2$, it follows from (\ref{cdrsfn}) that
\begin{equation}
\sum_{\{\br,\br'\}}\curGamma(\br,\br';\zeta_1,\zeta_2)=N\varphi_{\vtmT},\label{flrs}
\end{equation}
which is a \emb{fluctuation-response function relationship}.
If translational invariance is assumed, then
$\curGamma(\br,\br';\zeta_1,\zeta_2)=\curGamma(\obr;\zeta_1,\zeta_2)$,  where $\vsmallsep\obr\cequals\br-\br'$
and
\begin{equation}
\sum_{\{\obr\}}\curGamma(\obr;\zeta_1,\zeta_2)=\varphi_{\vtmT}\smallsep\mbox{with}\smallsep\curGamma^\star(\bzero;\zeta_1,\zeta_2)=\varphi_{\vtmT},
\label{flrsm}
\end{equation}
where $\curGamma^\star(\bk;\zeta_1,\zeta_2)$ is the Fourier transform of $\curGamma(\obr;\zeta_1,\zeta_2)$.\footnote{Given by
\[\curGamma^\star(\bk;\zeta_1,\zeta_2)\cequals\sum_{\{\obr\}}\curGamma(\obr;\zeta_1,\zeta_2)\exp(-\rmi\,\bk\cdot\obr),\smallsep
\curGamma(\obr;\zeta_1,\zeta_2)=\frac{1}{N}\sum_{\{\bk\}}\curGamma^\star(\bk;\zeta_1,\zeta_2)\exp(\rmi\,\bk\cdot\obr).\]}
The correlation length $\curxi(\zeta_1,\zeta_2)$, given by\footnote{The prefactor $c(\frakd)$ is dependent on the number of dimensions $\frakd$, in which the system is
infinite.  It can be show from Ginzburg-Landau  theory that $c(\frakd)=1/(2\frakd)$ \citep[][Sect.\ 5.6]{lavisnew}.\label{cfrakd}}
\begin{equation}
\curxi^2(\zeta_1,\zeta_2)\cequals c(\frakd)\frac{\sum_{\{\obr\}}|\obr|^2\curGamma(\obr;\zeta_1,\zeta_2)}{\sum_{\{\obr\}}\curGamma(\obr;\zeta_1,\zeta_2)}
= - c(\frakd) \frac{\nabla^2_{\bk}\curGamma^\star(\bzero;\zeta_1,\zeta_2)}{\curGamma^\star(\bzero;\zeta_1,\zeta_2)},
\label{corlen}
\end{equation}
is a measure of distance over which microscopic degrees of freedom are statistically correlated.

\vspace{0.2cm}

\noindent  We are now able to augment the scaling theory, described in Sect.\ \ref{scathr}, by applying it to the correlation function and correlation length.
Again adopting the magnetic model used in of Sect.\ \ref{scathr},
suppose that near a critical point
these functions can be re-expressed in terms of the scaling fields $\theta_\vtmT$ and $\theta_\vtpH$;   $\obr$ and $\bk$ can also be treated as scaling fields which, on dimensional grounds will
have exponents $-1$ and $+1$ respectively. Then the relationships (\ref{flrsm}) between the correlation function and the response function $\varphi_{\vtmT}$,  together with the formula
(\ref{excf}) derived from Ginzburg-Landau theory suggests a scaling form\footnote{The exponent of minus one for $\obr$ is chosen on dimensional grounds. It is also equivalent to the rescaling of
length in the renormalization group (item (iii) in Sect.\ \ref{tgs}).}
\begin{equation}
\curGamma(\lambda^{-1}\obr;\lambda^{y_\vtmT}\theta_\vtmT,\lambda^{y_\vtpH}\theta_\vtpH)=\lambda^{\cureta+d-2}\curGamma(\obr;\theta_\vtmT,\theta_\vtpH),
\label{sccfn}
\end{equation}
for the correlation function, and,  hence
\begin{equation}
\curGamma^\star(\lambda\bk;\lambda^{y_\vtmT}\theta_\vtmT,\lambda^{y_\vtpH}\theta_\vtpH)=\lambda^{\cureta-2}\curGamma^\star(\bk;\theta_\vtmT,\theta_\vtpH),
\label{sccfnft1}
\end{equation}
for its Fourier transform. Then, from (\ref{corlen}), the scaling form for the correlation length is
\begin{equation}
\curxi(\lambda^{y_\vtmT}\theta_\vtmT,\lambda^{y_\vtpH}\theta_\vtpH)=\lambda^{-1}\curxi(\theta_\vtmT,\theta_\vtpH).
\label{corlsc}
\end{equation}
From (\ref{flrsm}), (\ref{sccfnft1}) and (\ref{cdrsfn}),  $d-2y_\vtpH=\cureta-2$ and, setting $\lambda=|\theta_2|^{-1/y_\vtmT}$ in (\ref{corlsc})  gives, from (\ref{excl})
\begin{equation}
\curnu=\curnu'=1/y_\vtmT.
\label{nurel}
\end{equation}
Then, from (\ref{betai4})  and (\ref{betai6}), $\curnu(2-\cureta)=\curgamma$,
which is the \emb{Fisher scaling law} \citep{2-f1}  and $d\,\curnu = 2 - \curalpha$,
which is the \emb{Josephson hyper-scaling law}  \citep{2-j1}.\footnote{This is the only scaling law which involves the
dimension $d$ of the system.  For reasons which become evident if Ginzburg-Landau theory is used in the Gaussian approximation \citep[][Sect.\ 5.6]{lavisnew}
it becomes invalid when $d>d_{\tU\tC}$, the upper-critical dimension.  This is the dimension such that, when $d\ge d_{\tU\tC}$, critical exponents become dimensionally independent
with the classical  values given by, for example, the van der Waals fluid. For  the Ising and similar non-quantum systems (see Appen. \ref{tim})  $d_{\tU\tC}=4$.\label{ucd}}
%----------------------------------------------------------------------------------------------------------
\subsection{Transfer-Matrix Methods}\label{extls}
 As we have already shown $S$, $\Phi_1$ and $\Phi_2$  are all extensive functions  of their extensive variables or none of them is.   The message
 \FSM--\refsf{fsm-02} sent
from statistical mechanics to thermodynamics is that the latter is the case, and in particular that
\begin{equation}
\phi_2\cequals \frac{\Phi_2(\zeta_1,\zeta_2,N)}{N}=\phi_2(\zeta_1,\zeta_2)
\label{extphi1}
\end{equation}
is true only as an approximation for large systems.\footnote{In fact the Sackur-Tetrode formula for the entropy of a perfect gas given by
(\ref{prefl2}) and treated there as an assumption is, when derived from statistical mechanics, also not completely extensive.  This condition
is achieved only when $N$ is large and the Stirling formula for $N!$ is applied.}
   We shall now substantiate this claim by considering a particular way to develop statistical mechanical models, namely the method of transfer matrices.
   Although, of course, statistical mechanics can model systems of microsystems (molecules) moving, as in a fluid, through a continuum of points, transfer
   matrix methods are restricted to microsystems confined to the points of a lattice. In principle lattices of any dimension can be considered, but we shall, for easy of presentation,
   consider only the two-dimensional case.  A virtue of this development is that it can be clearly seen how it unfolds as  the two lattice directions in which
    the system gets larger and then  infinite are applied   separately.

\vspace{0.2cm}

\noindent   Consider a square lattice, of lattice spacing $\fraka$, with $N_{\tH}$ sites in the horizontal direction, $N_{\tV}$ in the vertical direction, so that $N=N_\tH N_\tV$.
This situation is like the one considered for finite-size scaling in Sect.\  \ref{fssc}, when extensivity can be considered separately in the two directions.
  Periodic boundary conditions are applied so that the lattice forms a torus with  horizontal rings
of $N_\tH$ sites and  rings in a vertical plane of $N_\tV$ sites.\footnote{The point we are establishing with respect to extensivity is even more evident in systems with open boundaries.}
 We suppose that the sites of the lattice are occupied by identical microsystems having $\nu$ possible states.\footnote{The Ising model of Appen.\ \ref{tim} is an example of
 such a model with $\nu=2$.}
 The state of the whole system is $\bsigma\cequals(\tbsigma_1,\tbsigma_2,\ldots,\tbsigma_{N_\tH})$,
where $\tbsigma_i$,  the state of the $i$-th vertical ring of sites, has one of $N_\tR\cequals\nu^{N_\tV}$ values.
Given that contributions to the Hamiltonian arise (at least in the horizontal direction) only between first-neighbour sites
the Hamiltonian can be decomposed  into interactions between neighbouring rings of sites and within rings.
 The latter can be distributed between interacting pairs of rings so that the Hamiltonian takes the form
of the sum of contributions of interactions between rings and it  is straightforward to show that  the partition function is expressible in the form
\begin{equation}
 Z_2(\zeta_1,\zeta_2,N)=\Trace\{\bV^{N_\tH}\},
 \label{ztrm}
 \end{equation}
 where $\bV$ is the $N_\tR$-dimensional transfer matrix with elements  consisting of the exponentials of the negatives of the inter-ring interactions.
Assuming that $\bV$ is diagonalizable,\footnote{The condition for this to be the case is that $\bV$ is {\em simple} \citep[][p.\ 146]{3-l11}.}
 it is an elementary algebraic result that its trace is equal to the sum of its eigenvalues, which in decreasing order of magnitude we denote as
 $\Lambda^{(\ell)}(\zeta_1,\zeta_2,N_\tV)$, $\ell=1,2,\ldots,N_\tR$.  Then,
  from (\ref{gibfe}) and  (\ref{ztrm}),
\begin{equation}
 \Phi_2(\zeta_1,\zeta_2,N)=-\ln\{[\Lambda^{(1)}(\zeta_1,\zeta_2,N_\tV)]^{N_{\tH}} +\cdots+[\Lambda^{\left(N_\tR\right)}(\zeta_1,\zeta_2,N_\tV)]^{N_{\tH}}\}.
 \label{ztrmev}
 \end{equation}
As we can see the factors $N_\tH$ and $N_\tV$ of $N$  are `buried' at different places in this expression and it is clear that the extensivity condition (\ref{extphi1}) is not satisfied
and the negative aspect of the message \FSM--\refsf{fsm-02}  from statistical mechanics to thermodynamics is justified.
However,  we can make some progress because, if all the elements of $\bV$ are strictly positive, as will usually be the case, an important theorem of
\citet{3-p5} (see also,  \citealt{2-g11}, p. 64; \citealt{lavisnew}, p.\ 673) states that the largest eigenvalue of $\bV$  is real,  positive and non-degenerate.  This means that, in the approximation
when $N_\tH$ becomes large,
\begin{equation}
 \Phi_2(\zeta_1,\zeta_2,N)\simeq-   N_\tH\ln\{\Lambda^{(1)}(\zeta_1,\zeta_2, N_\tV)\}
 \label{largnh}
 \end{equation}
with extensivity achieve in the horizontal direction. Two strategies emerge at this point:

\vspace{0.2cm}

\noindent The first is to calculate an expression of the form
\begin{equation}
 \Lambda_1(\zeta_1,\zeta_2, N_\tV)\cequals [\psi(\zeta_1,\zeta_2)]^{N_\tV},
 \label{exres}
 \end{equation}
 valid in the limit $N_\tV\to\infty$ and giving
\begin{equation}
 \phi_2(\zeta_1,\zeta_2)=-   \ln\{\psi(\zeta_1,\zeta_2)\}
 \label{largnhs}
 \end{equation}
in the limit $N\to\infty$.  If this calculation can be carried out it is an effective proof of the existence of the thermodynamic limit,\footnote{Although, of course, the {\em current}
absence of such a calculation is not a proof of the contrary assertion.} which achieves complete extensivity, with free-energy density given
by (\ref{largnhs}).
It is, however, a strategy that has been successfully applied in only a few cases, of which  \citeauthor{1-o2}'s (\citeyear{1-o2})  solution of the
  two-dimensional zero-field Ising model and \citeauthor{2-b7}'s (\citeyear{2-b7})  solution of the eight-vertex model are the most well-known instances.

\vspace{0.2cm}

\noindent In the absence of a complete solution as represented by (\ref{largnhs}), the strategy most often adopted is to treat $N_\tV$ as
 a parameter indexing a sequence of models. That is
\begin{equation}
\Psi^{(N_\tV)}(\zeta_1,\zeta_2)\cequals \Lambda^{(1)}(\zeta_1,\zeta_2, N_\tV)
  \label{exres1}
 \end{equation}
 and
\begin{equation}
 \phi^{(n)}_2(\zeta_1,\zeta_2)\simeq-  \frac{\ln\{\Psi^{(n)}(\zeta_1,\zeta_2)\}}{n}.
 \label{largnhi}
 \end{equation}
In the case of the Ising and similar semi-classical models it can be shown by a method due to \citet{2-p8}  that  $\phi^{(n)}_2(\zeta_1,\zeta_2)$
is   a smooth function for all $n>0$ which exhibits maxima in response functions. A quantitative analysis using finite-size scaling theory  (see Sect.\ \ref{fssc})
shows that such maxima become increasingly steep for increasing values of $n$, with
 convergence to the singularity associated with the transition in the two-dimensionally infinite system as $n\to\infty$.
 In particular to the corresponding singularities in \citeauthor{1-o2}'s solution of the two-dimensional zero-field Ising model.
 However, in view of the discussion later in this work it should be noted that the limiting process is singular.  Although the maxima in the
finite-$N_\tV$ models converge to the singularities in the $N_\tV=\infty$ model they
 remain of a different (non-singular) character however large $N_\tV$ becomes.

 \vspace{0.2cm}

\noindent The pair correlation function and correlation length were defined in Sect. \ref{corcoln}.
In terms of this transfer matrix formulation
 it can be shown \citep[][Sect.\ 11.1.3]{lavisnew} that in the limit $N_\tH\to\infty$
\begin{equation}
\curxi(\zeta_1,\zeta_2,N_\tV) \simeq - \fraka{\left\{\ln\left|\Omega_2(\zeta_1,\zeta_2,N_\tV)\right|\right\}}^{-1},
\label{atm43}
\end{equation}
where $\fraka$, the lattice spacing,  is now the distance between neighbouring rings of sites,
\begin{equation}
\Omega_2(\zeta_1,\zeta_2,N_\tV)\cequals{\Lambda^{(2)}(\zeta_1,\zeta_2,N_\tV)}/{\Lambda^{(1)}(\zeta_1,\zeta_2,N_\tV)}
\label{atm43a}
\end{equation}
and
\begin{equation}
\curGamma_2(\br,\br';\zeta_1,\zeta_2,N_\tV)\sim \exp\{-|\br-\br'|/\curxi(\zeta_1,\zeta_2,N_\tV)\},
\label{cftme}
\end{equation}
in the limit $|\br-\br'|\to \infty$, where $\br$ and $\br'$ lie on the same vertical ring of sites
which establishes an asymptotic form for $f_d(|\obr|/\curxi)$ in (\ref{excf}).

\vspace{0.2cm}

\noindent  The situation where $N_\tH\to\infty$ and $N_\tV$ is finite corresponds to that to be discussed in Sect. \ \ref{fssc},
below, for finite-size scaling, where here $\frakd\cequals 1$ and the thickness of the lattice $\aleph\cequals N_\tV$,
with a maximum in  $\varphi_\vtmT$ and in other response functions signalling  an incipient singularity.\footnote{In Sect.\ \ref{fssc}
we are primarily concerned with fully-finite systems, although as we indicated there, the analysis also applies to cases where, like here, $0<\frakd\le d_{\tL\tC}<d$.}
The eigenvalue ratio $\Omega_2(\zeta_1,\zeta_2,N_\tV)$  can also be used as a means of detecting
an incipient singularity, but in a slightly different way. Since, in \citeauthor{1-o2}'s  solution  for the Ising model, the largest eigenvalue
 is degenerate along the first-order transition line below the critical temperature \citep[][p.\ 194]{1-d4}, we expect that $\Omega_2(\zeta_1,\zeta_2,N_\tV)$
will begin, as $N_\tV$ is increased, to form a `plateau'
 with small (negative) slope for small temperature temperatures. The end of this plateau, where the negative curvature is a
 maximum can then be construed as the location of an incipient singularity.\footnote{Similar arguments apply to the three-state Potts
model \citep[][Sect.\ 11.32]{lavisnew}.}
 The finite-size scaling argument  of Sect.\ \ref{fssc} can  be applied to
all these quantities showing that  the maxima converge towards the infinite-system critical value as $N_\tV$ increases.
However, of course, for finite $N_\tV$ we cannot expect these locations to exactly coincide. These perceptions are given
 further weight by the phenomenological renormalization group procedure described in Sect.\ \ref{wtfn} (c).

 As we have already indicated, the use of transfer matrix methods to determine exact solutions for infinite systems leads into our discussion in Sect.\ \ref{nftl}
 of the thermodynamic limit. In a similar way our account of incipient singularities resulting from an analysis of systems with $N_\tV$ finite leads into our discussion
 of phase transitions in finite systems is Sect.\ \ref{ptfs}.
%----------------------------------------------------------------------
\subsection{The  Renormalization Group Method}\label{thrgrp}
Once it became evident, around the turn of the twentieth century that the exponents associated
 with a critical point, both in experimental systems and theoretical models
were not those derived from classical models, like van der Waals equation, an interest developed
 in determining their exact values, in experimental systems and also
in theoretic models, where of course it was also necessary in many cases  to derive the critical temperature.
Before the advent of renormalization group methods the most successful way to do this
was by using high and low temperature series.  These were very successful in obtaining critical temperatures and exponents
at second-order critical points.  However, although they can be adapted to deal with first-order transitions, this is not their
main strength and they are also not designed to map out the whole  picture of phase transition curves
in thermodynamic space.
This contrasts with the renormalization group methods developed in the late
 sixties -- early seventies. They are able (when they work) not only to deal with critical points
  but also curves of first-order and second-order transitions. However, any account of these methods should be proceeded
   by some words of warning, like those of John \citeauthor{2-c7}. As he says \citep[][pp.\ 28--29]{2-c7}:
 \begin{quote}
 ``Not only are the words `renormalization'\footnote{A carry-over on the part of \citet{2-w13} from his work on the high-energy behaviour of renormalized quantum electrodynamics.}
 and `group'\footnote{It is in fact a semi-group since it has no inverse.}  examples of unfortunate terminology, the use of the definite article `the' which usually precedes them
 is even more confusing. It creates the   misleading  impression that the renormalization group is a kind of universal machine through which any problem may be processed, producing neat
 tables of critical exponents at the other end. This is quite false. It cannot be stressed too strongly that the renormalization group is merely a framework, a set of ideas, which has
 to be adapted to the nature of the problem at hand. In particular, whether or not a renormalization group approach is quantitatively successful depends to a large extent on the nature of the
 problem, but lack of success does not necessarily invalidate the qualitative picture it provides.
\end{quote}
  Here we shall concentrate solely on the approach to the renormalization group  which is usually
referred to as happening in `real space';  in contrast to the approach initiated by \citet{2-w13}
where renormalization is performed in wave-vector space resulting in expansions in the parameter $\epsilon\cequals d-4$.\footnote{A comprehensive collection of
articles on both types of renormalization group methods is contained
in the articles in \citet{2-d16}   and on real-space methods in the
volume edited by \citet{2-b17}. For a description of renormalization
methods in wave-vector space the reader is referred to \citet{2-m3} and
\citet{2-a1}. Accounts of both approaches are given by \citet{3-g10},
 \citet{2-b47} and \citet{2-c7}.}

 \vspace{0.2cm}

 \noindent  The core of real-space renormalization group (RSRG) methods is the construction of a semi-group of transformations on the independent couplings, or functions thereof.
 There is a variety of procedures for doing this.  Many are based on
the block-spin method of \citet{2-k1}, and another popular technique is decimation, where  the states of a proportion of the microsystems  is summed out of the
partition function.  In fact decimation applied to the one-dimensional Ising model, or related models like the Potts model, \citep[see, e.g.][Sect.\  15.5.1, and Sect.\ \ref{wtfn} (a) below]{lavisnew}
is one of the few examples of an exact  RSRG transformation.  Most transformation involve approximations, which thus means that the critical exponents are approximations
with, in many cases no obvious way to make improvements, unlike series methods where, in principle and often with a great deal of labour, improvements are made
by extending the series.

 In essence the RSRG transformation involves some fractional reduction in the number of degrees of freedom.  It would, therefore, seem to follow that there must
 have been a prior application of the thermodynamic limit. Whether this is required for the renormalization group and, more generally, whether it
  is needed at all in the statistical mechanics of critical phenomena is a question that we return to in Sect. \ref{tlnism}, following  a brief account of the ideas involved in
the RSRG.
%----------------------------------------------------------------------------------------------------------
\subsubsection{General Theory}\label{tgs}
Underlying the semigroup of transformations on couplings, which is the real-space renormalization group, is a mapping from a lattice $\mcN$  to a lattice $\wtmcN$.  For the sake of simplicity
we suppose that both  are hypercubic lattices with periodic boundary conditions.  Then:
\begin{enumerate}[(i)]
\item  The number of sites $N$ and $\wtN$ of $\mcN$ and $\wtmcN$ are related by $\wtN= N/\lambda^d$, where $\lambda>1$.
\item The lattice spacings $\fraka$ and $\tfraka$ of $\mcN$ and $\wtmcN$ are related by $\tfraka=\lambda \fraka$.
\item The size of $\wtmcN$ is reduced by a length scaling $|\tbr|= |\br|/\lambda$.
\end{enumerate}
The renormalization group is constructed by imposing onto the lattice transformation a statistical mechanical
transformation. To do this we modify the Hamiltonian (\ref{hamform}) to
\begin{equation}
\hH^\prime_2(\bsigma;\zeta_0,\bzeta;N)\cequals N\zeta_0 +\hH_2(\bsigma;\bzeta;N),
\label{newhamform}
\end{equation}
where, for reasons that will become evident below, we have  added a term including
a trivial coupling $\zeta_0$ and, as in the presentation of scaling theory at the beginning of  Sect.\ \ref{scathr},  generalized the number of non-trivial couplings from two to $n$,
with $\bzeta\cequals(\zeta_1,\zeta_2,\ldots,\zeta_n)$.\footnote{As we see in the example of the two-dimensional Ising model in Sect.\ \ref{wtfn}(b) below, the renormalization group
transformation will in many cases generate further couplings beyond the set dictated by the physics of the model under consideration.}  The terminology `trivial' signals the fact that, if  in (\ref{gibprob})  $\hH_2$ is replaced by $\hH^\prime_2$ and $Z_2$ by
\begin{equation}
Z^\prime_2(\zeta_0,\bzeta,N)\cequals \sum_{\{\bsigma\}} \exp\{-\hH^\prime_2(\bsigma;\zeta_0,\bzeta,N)\},
\label{gibpfh}
\end{equation}
then the probability density function is left unchanged and
\begin{equation}
\Phi_2(\bzeta,N)\cequals -\ln\{\exp(N\zeta_0)Z^\prime_2(\zeta_0,\bzeta,N)\}.
\label{gibfeh}
\end{equation}
Bearing in mind the remarks of \citeauthor{2-c7}, given above, a successful application of this method depends on being able
 to construct relationships between the couplings $\zeta_0,\bzeta$ in the system on $\mcN$ and the couplings $\tzeta_0,\tbzeta$
in the system on $\wtmcN$,  done in such a way that the values for the couplings for $\mcN$ place it in a critical region if and only the same is the case for
the values of the couplings for $\wtmcN$.   Since the critical properties of a system are contained within the partition function the invariance
\begin{equation}
Z^\prime_2(\tzeta_0,\tbzeta,\wtN)=Z^\prime_2(\zeta_0,\bzeta,N)
\label{inpf}
\end{equation}
of that function is a sufficient guarantee; and this is achieved by
the relationship
\begin{equation}
\exp\{-\hH^\prime_2(\tbsigma;\tzeta_0,\tbzeta,\wtN)=\sum_{\{\bsigma\}} w(\bsigma,\tbsigma)\exp\{-\hH^\prime_2(\bsigma;\zeta_0,\bzeta,N)\},
\label{rgt}
\end{equation}
where the \emb{weight function}  $w(\bsigma,\tbsigma)$ satisfies
\begin{equation}
\sum_{\{\tbsigma\}} w(\bsigma,\tbsigma)=1.
\label{wfcn}
\end{equation}
Running over the set of states $\tbsigma$  in (\ref{rgt}) will, {\em in principle}, produce \emb{recurrence relationships}\footnote{To be precise, the recurrence relationships are derived from (\ref{rgt}) as  relationships between the Boltzmann factors $\exp(\tzeta_j)$ and $\exp(\zeta_j)$,
$j=0,1,\ldots,n$ of the couplings.}
\begin{equation}
\tzeta_j=\pK_j(\bzeta),\tripsep j=1,2,\ldots,n,
\label{recr12}
\end{equation}
 and for $\tzeta_0$ a recurrence relationship which we choose, for convenience to express in the form
\begin{equation}
\tzeta_0=\lambda^d[\zeta_0+\pK_0(\bzeta)].
\label{ecr0}
\end{equation}
The `in principle' caveat entered here is important. As we shall see it is rarely possible to implement this programme and to choose a weight function without some kind
of approximation being applied.  And it is frequently the case that consistency can be achieved only by increasing the value of $n$ from its initial value.  When this
happens it is necessary, in order to apply repeated iterations, to back-track and for the extra couplings to be included from the start.

The importance of (\ref{ecr0}) is that it can be used, together with  (\ref{gibfeh}) and (\ref{inpf}) to obtain the relationship
\begin{equation}
\phi_2(\tbzeta)=\lambda^d\phi_2(\bzeta)-\lambda^d\pK_0(\bzeta).
\label{trgt7}
\end{equation}
between  the free-energy densities per lattice site at $\bzeta$ and $\tbzeta$.
Then, given that (\ref{recr12}) can be iterated to produce a sequence of points $\bzeta^{(0)}\to\bzeta^{(1)}\to\bzeta^{(3)}\to\cdots$
in $\Xi_2$,
\begin{equation}
\phi(\bzeta^{(0)}) = \sum_{s=0}^{\infty}{\frac{1}{\lambda^{sd}}}{\pK}_0(\bzeta^{(s)})\,
\label{fed4}
\end{equation}
is the free-energy density at an  initial point $\bzeta^{(0)}$.
Although this result seems to  imply the need for an infinite number of iterations, this is clearly not possible in practical computations. It is, therefore,
 fortunate that it is usually found that this series converges after a very few iterations, allowing densities and response functions to be calculated  (see the discussion  Sect.\ \ref{nftl}).

\vspace{0.2cm}

\noindent  A \emb{fixed point}  $\bzeta^\star$ of (\ref{recr12}) is associated with either a single-phase region or a critical region  $\mcC$  in $\Xi_2$.
To analyze its nature we linearize with $[\btriangle{\bzeta}^{(s)}]^\tT \cequals \bzeta^{(s)}-\bzeta^{\star}$ to give
$\btriangle\bzeta^{(s+1)}\simeq {\bL}^{\star}\btriangle\bzeta^{(s)}$,
where $\bL^{\star}$ is the fixed-point value of the matrix $\bL$ with elements $L_{ij}\cequals \partial\pK_i/\partial\zeta_j$.
In general $\bL^{\star}$ is not symmetric, with different left and right eigenvectors $\bw_j$ and $\bx_j$ for the eigenvalue
$\Lambda_j$.  It can then be shown\footnote{Assuming that $\bL^\star$ is a simple matrix.}
that in a neighbourhood of the fixed point there exist   scaling fields
$\theta_j = \theta_j(\btriangle\bzeta)$, $j=1,2,\ldots,n$ which are smooth functions of the couplings with
\begin{eqnarray}
\theta_j(0) = 0,\pairsep&&
\theta_j^{(s+1)} = \Lambda_j\theta_j^{(s)},
\label{fplrg10}\\
\theta_j \simeq\phantom{\sum_{j=1}^n}\hspace{-0.2cm} \bw_j\centerdot\btriangle\bzeta,
\pairsep&&
\btriangle\bzeta \simeq \sum_{j=1}^n \bx_j \theta_j.\label{fplrg12}
\end{eqnarray}
which is a realization of the relationship between scaling fields and couplings described in Sect.\  \ref{scathr}.

From (\ref{fplrg10})  $\theta_j^{(s+k)} = \Lambda^{k+s}_j\theta_j^{(0)}$ and the semi-group character of this transformation implies that
  $\Lambda_j\cequals \lambda^{y_j}$, for $j=1,2,\ldots,n$ and a set of exponents $y_1,y_2,\ldots,y_n$.
Then, in a neighbourhood of the fixed point $\bzeta^\star$
the couplings $\zeta_j$ and $\tzeta_j$ in (\ref{trgt7}) can be expressed as
\begin{equation}
\zeta_j=\zeta_j^\star +\sum_{i=1}^n x_i^{(j)}\theta_i,\pairsep
\tzeta_j=\zeta_j^\star +\sum_{i=1}^n x_i^{(j)}\lambda^{y_i}\theta_i,
\label{tzetoth}
\end{equation}
where $\bx_i\cequals(x_i^{(1)},x_i^{(2)},\ldots,x_i^{(n)})$.
In (\ref{trgt7})  the function $\pK_0(\bzeta)$ is regular.
So in a region around $\bzeta^\star$ the singular part $\phi_{\sing}(\triangle\bzeta)$
of $\phi_2(\bzeta)$,
with $\phi_{\sing}(0)=0$, can be re-expressed in terms of the scaling fields to give
\begin{equation}
{\phi}_{\sing}({\lambda}^{y_1}{\theta}_1,\ldots,{\lambda}^{y_n}{\theta}_n) = {\lambda}^d {\phi}_{\sing}
({\theta}_1,\ldots, {\theta}_n)\, ,
\label{egs4rg}
\end{equation}
which is a substantiation of (\ref{egs4}).
%-----------------------------------------------------------------------------------------------------------------------------------------------------------------------
\subsubsection{Finite-Size Systems}\label{fssc}
This treatment of criticality, which plays an important role in our understanding of PTCP in real systems (see Sect. \ref{ourp}),  was initiated by \citet{2-f7} and \citet{2-f8}.\footnote{For a review
see \citet{2-b19} and, for a collection of papers on finite-size scaling, \citet{2-c4}.}
For simplicity we suppose, as in Sect.\ \ref{extls},  that the system under consideration   consists of $N$ identical microsystems
on the sites of a $d$-dimensional  hypercubic lattice $\mcN_d$
with $N_k$ sites} in the $k$-direction and $N_1N_2\cdots N_d=N$.
A partially-infinite system of \emb{thickness} $\aleph\cequals [N^{(\frakd)}]^{1/(d-\frakd)}$, where
$N^{(\frakd)}\cequals N_{\frakd+1}N_{\frakd+2}\cdots N_d$, is obtained if the thermodynamic limit $N_k\to\infty$ is taken
 only for $k=1,2,\ldots,\frakd<d$.
In  a fully-finite system $\frakd=0$ and $N^{(\frakd)}=N$.   We denote the critical region
in the partially-infinite system,  when $\frakd>d_{\tL\tC}$,  by $\mcC(\frakd;\aleph)$, with $\mcC(d;\infty)=\mcC$. Finite-size scaling theory can be  applied both
 to a partially-infinite system, where there is the possibility of a critical region consisting of some kind of singular behaviour,  and a  fully-finite system where there
  is not.  In a fully-finite system or a partially infinite system with $\frakd\le d_{\tL\tC}$ the critical region is replaced by:
 \begin{definition}\label{qcritreg}
 For a fully-finite, or partially-infinite system with $\frakd\le d_{\tL\tC}$, a  region $\mcI\mcS(\frakd;\aleph)$ in the space of couplings is one of \emb{incipiently singularity},\footnote{It should be noted that this is a slightly different usage from that in \citet[][Chap.\ 11]{lavisnew}, where such occurrences are called
`incipient phase transitions'.} if in the limit $\aleph\to \infty$, it maps into a critical region $\mcC$ of the infinite system.
 \end{definition}
 Expressed in a slightly different way a system has an incipient singularity
at certain size-dependent  values  of it couplings if,  as the system size $\aleph$ is increased, those values converge to ones where thermodynamic functions exhibit {\em properties that have no finite limits}.

\vspace{0.2cm}

\noindent The basic assertion of finite-size scaling is that
$\theta_\aleph\cequals 1/{\aleph}$ can be treated as another scaling field with $y_\aleph=1$, meaning that $\theta_\aleph$ is a relevant scaling field,  and   $\theta_\aleph=0$ for the infinite system.
The only condition required for this is that the system is sufficiently large for the renormalization group transformation in the space of all the other couplings
to be unmodified by the finite size of the system. That is to say, that the renormalized couplings can be represented in the system.
For simplicity we confine our attention to the simple magnetic system used in  Sect.\ \ref{scathr}.   The critical region for the infinite system is just a critical point
$T=T_\rmc$, $\pH=0$
with scaling fields $\theta_\vtmT$ and $\theta_\vtpH$, given by (\ref{scflh}), measuring departures from this point.
When the system has finite thickness ($\theta_\aleph\ne 0$),  the incipient singularity is at a different temperature, but because of the symmetry of the system still with $\pH=0$.
Again, for simplicity,  attention will be restricted  to the zero-field axis
 where two  temperatures come into play:
\begin{enumerate}[(i)]
\item For  a system of finite thickness $\aleph$,  $\tiT(\aleph)$  is the \emb{shift temperature}
such that, as $\aleph\to\infty$, $\tiT(\aleph)\to T_\rmc$,  the temperature at which  {\em  the infinite system has a singularity}.  If  $\frakd>d_{\tL\tC}$
then $\tiT(\aleph)$ is also a critical temperature, but for the system of finite thickness.
 If $\frakd\le d_{\tL\tC}$,  and in particular when $\frakd=0$
and the system is fully-finite, $T=\tiT(\aleph)$ is   a  \emb{quasicritical temperature} \citep{2-f8} which is exhibited by
 a maximum in the susceptibility.\footnote{Another
response function like the heat capacity can replace the susceptibility, with a slightly different quasicritical temperature.} This temperature is an example of an \emb{incipient singularity}. In keeping with the other assumptions of scaling theory it is assumed that this convergence is algebraic, so, with
scaling field
\begin{equation}
 \titheta_\vtmT(T,\aleph)\cequals\varepsilon\left(\ssfrac{1}{T}-\ssfrac{1}{\tiT(\aleph)}\right),
 \label{shipfd}
 \end{equation}
 the condition
{\jot=0.25cm
\begin{eqnarray}
\titriangle(\aleph)&\cequals&\theta_\vtmT(T)-\titheta_\vtmT(T,\aleph)
=\varepsilon\left(\ssfrac{1}{\tiT(\aleph)}-\ssfrac{1}{T_\rmc}\right)\nonumber\\*
&=&\theta_\vtmT(\tiT(\aleph))=- \titheta_\vtmT(T_\rmc,\aleph)
\simeq C_\rms\aleph^{-\curchi}\smallsep\mbox{as $\aleph\to\infty$},
 \label{shifte}
 \end{eqnarray}
where $\curchi>0$ is the \emb{shift exponent}, is sufficient to ensure convergence.
\item $\scT(\aleph)$, called the \emb{rounding temperature} is an important, but rather more elusive, property of the system.  It is the temperature at which the susceptibility
first shows significant deviation from that of the fully-infinite system.  With
 \begin{equation}
 \sctheta_\vtmT(T,\aleph)\cequals\varepsilon\left(\frac{1}{T}-\ssfrac{1}{\scT(\aleph)}\right),
 \label{roundd}
 \end{equation}
it is supposed that
\begin{eqnarray}
\sctriangle(\aleph)&\cequals& \titheta_\vtmT(T,\aleph)- \sctheta_\vtmT(T,\aleph)=\varepsilon\left(\frac{1}{\scT(\aleph)}-\ssfrac{1}{\tiT(\aleph)}\right)
\nonumber\\*
&=&\titheta_\vtmT(\scT(\aleph),\aleph)=- \sctheta_\vtmT(\tiT(\aleph),\aleph)\simeq C_\rmr \aleph^{-\curtau},\smallsep\mbox{as $\aleph\to\infty$},
 \label{rounde}
 \end{eqnarray}
where} $\curtau>0$ is the \emb{rounding exponent}.
\end{enumerate}
Scaling around the infinite system critical point is shown in Fig. \ref{Fig6}.
%---------------------------------------------------------
\begin{figure}[t]
\includegraphics[width=6cm]{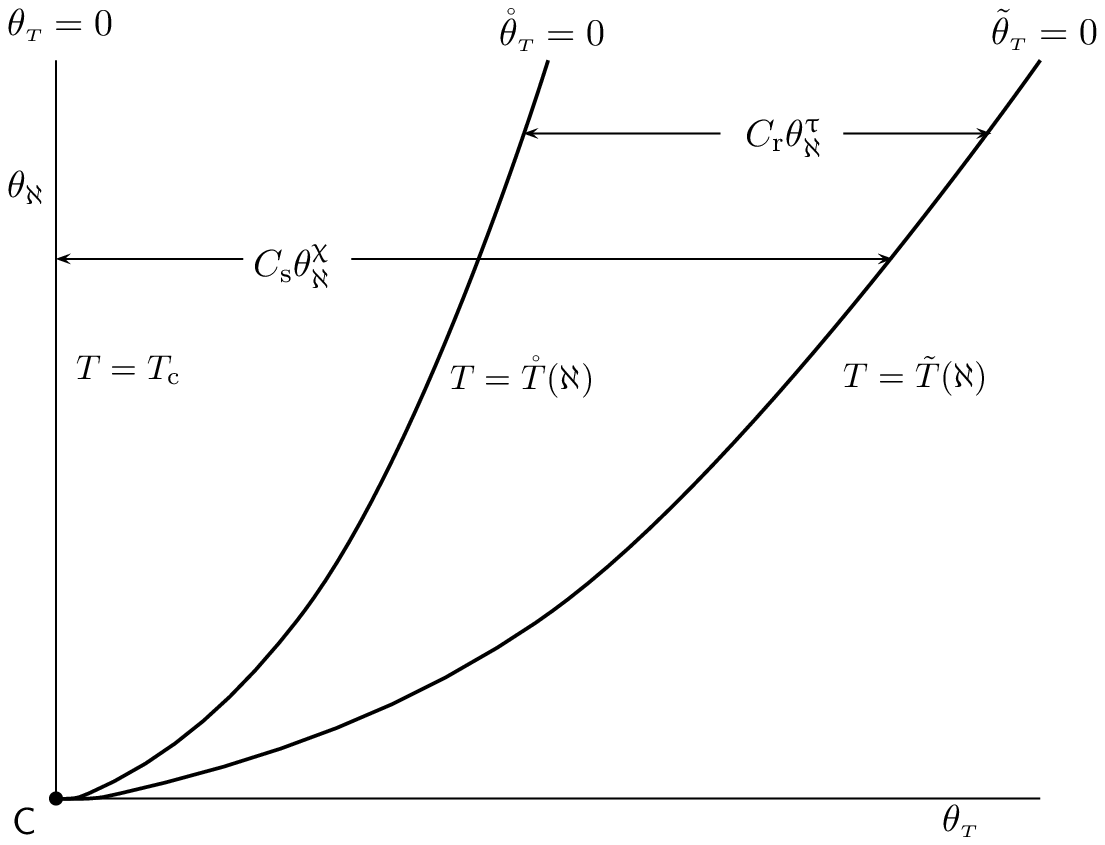}
\caption{Scaling around the critical point {\sfC}, showing the curves $\sctheta(T,\aleph)=0$ and $\titheta(T,\aleph)=0$ of rounding and shift temperatures.
}\label{Fig6}
\end{figure}
%----------------------------------------------------------
  Our interest in this work is in the occurrence of an incipient singularity; so henceforth  the assumption is that $\frakd\le d_{\tL\tC}$.\footnote{In Sect.\ \ref{extls} on transfer matrix methods,
and Sect.\     \ref{wtfn}(c)  on phenomenological renormalization,  $\frakd=d_{\tL\tC}=1$ and $d=2$. Our later discussion  in Sect.\ \ref{ourp}  is  concerned with phase transitions in fully-finite systems
where $\frakd=0$.} Thus we have three relevant scaling fields
with the critical region of the infinite system at the origin $(\theta_\vtmT,\theta_\vtpH,\theta_\aleph)=(0,0,0)$.  However, this is not the complete picture; in general there will be a number of irrelevant
scaling fields, which parametrize the critical region and affect its asymptotic properties.  For the sake of simplicity we just include the most nearly relevant.\footnote{That is $|y_\star|\cequals \min_{j\in s+1,\ldots,n} |y_j|$.} of these designated as $\theta_\star$,  with exponent $y_\star<0$. Then
 on the zero-field axis (\ref{egs4rg}) is replaced by
\begin{equation}
{\phi}_{\sing}(\lambda^{y_\vtmT}\theta_\vtmT,\lambda^{y_\star}\theta_\star,\lambda\theta_\aleph) = {\lambda}^d {\phi}_{\sing}(\theta_\vtmT,\theta_\star,\theta_\aleph)\, .
\label{egs4mf}
\end{equation}
As we have already seen, singular parts of
 thermodynamic functions like densities and response functions are  obtained by differentiations with respect to the scaling fields.
In particular, for the susceptibility $\varphi_\vtmT$, given by  (\ref{cp51}),
\begin{equation}
\varphi_\vtmT (\theta_\vtmT,\theta_\star,\theta_\aleph) = {\lambda}^{\curomega} \varphi_\vtmT(\lambda^{y_\vtmT}\theta_\vtmT , \lambda^{y_\star}\theta_\star,\lambda \theta_\aleph )\ ,\label{fss3p}
\end{equation}
with $\curomega \cequals  2y_\vtpH -d=\curgamma/\curnu$, where $\curgamma$ is given by (\ref{betai4}) and $\curnu\cequals 1/y_\vtmT$, given in  (\ref{nurel}), is the critical exponent of the correlation
length. Asymptotic behavior in a neighbourhood of the critical point, that is when $|\theta_\vtmT \ll 1$,  is then as usual exposed by choosing the
 scale parameter $\lambda\cequals|\theta_\vtmT|^{-1/y_\vtmT}$,  giving
\begin{equation}
\varphi_\vtmT (\theta_\vtmT,\theta_\star,\theta_\aleph) = |\theta_\vtmT|^{-\curgamma}  \varphi_\vtmT(\pm 1,\frakX_\star,\frakX_\aleph )\ ,
\label{scalY}
\end{equation}
where  the $\pm 1$ branches of  $\varphi_\vtmT(\pm 1,\frakX_\star,\frakX_\aleph )$ apply to the cases $\theta_\vtmT>0$ and $\theta_\vtmT< 0$, respectively, and
$\frakX_\star(T,\aleph)\cequals|\theta_\vtmT(T)|^{-y_\star\curnu}\theta_\star$,  $\frakX_\aleph(T,\aleph)\cequals |\theta_\vtmT(T)|^{-\curnu} \aleph^{-1}$ are scaling functions.
 In a similar way, with $\lambda\cequals\aleph$,
\begin{equation}
\varphi_\vtmT (\theta_\vtmT,\theta_\star,\theta_\aleph) = \aleph^\curomega \varphi_\vtmT(\frakX_\aleph^{-1/\curnu},\frakX_\star^{1/y_\star},1 )\ .
\label{scalYa}
\end{equation}
 {\em In the thermodynamic limit\/} $\aleph \to \infty$,  it follows from (\ref{scalY}) that the susceptibility  has the form
\begin{equation}
\varphi_\vtmT(\theta_\vtmT,\theta_\star,0) =A_\vtmT^{(\pm)}(\frakX_\star)  |\theta_\vtmT|^{-\curgamma},
\label{varphiThdLim}
\end{equation}
where the amplitudes
\begin{equation}
A_\vtmT^{(\pm)}(\frakX_\star)\cequals \varphi_\vtmT(\pm 1,\frakX_\star,0),
\label{defa}
\end{equation}
which are, in general,  different for $\theta_\vtmT>0$ and $\theta_\vtmT < 0$, are dependent on $\theta_\vtmT$ by virtue of the presence of the irrelevant scaling field $\theta_\star$.
This contribution  will become small, as $|\theta_\vtmT|^{-y_*\curnu} \to 0$ for $|\theta_\vtmT| \to 0$, eventually becoming negligible for sufficiently small $|\theta_\vtmT|$. The susceptibility will then
display an {\em asymptotic algebraic  singularity\/} of the form
\begin{equation}
\varphi_\vtmT \simeq A_\vtmT^{(\pm)}(0) |\theta_\vtmT|^{-\curgamma}\ ,\quad \mbox{as}~~~|\theta_\vtmT|\to 0\ .
\label{Ypowerlaw}
\end{equation}
 The singularity is a divergence, if $\curgamma > 0$, which is generally the case for response functions.

\vspace{0.2cm}

\noindent  Given that both (\ref{scalY}) and (\ref{scalYa}) are valid, and that a finite statistical mechanical system {\em cannot\/} exhibit non-analytic behaviour, whereas singular behaviour {\em does occur\/} at critical points in the limit of infinite system size, the scaling function $\varphi_\vtmT(\pm 1,\frakX_\star,\frakX_\aleph )$ in  (\ref{scalY}) {\em must\/} exhibit asymptotic behaviour of the form
\begin{equation}
\varphi_\vtmT(\pm 1,\frakX_\star,\frakX_\aleph) \simeq  B_\vtmT^{(\pm)}(\frakX_\star)\,  \frakX_\aleph^{-\curomega}.
\label{varpiAsympt}
\end{equation}
Since the susceptibility has  maxima along the curve $T=\tiT(\aleph)$  of shift temperatures in Fig.\ \ref{Fig6}   these maxima will be in one of the branches of
 $B_\vtmT^{(\pm)}(\frakX_\star)$    with the other branch being a monotonically decreasing function of $\frakX_\star$ in the vicinity of
 $\frakX_\star=0$.   Along the curve of shift temperatures, from (\ref{shifte}),   $\frakX_\star(\tiT(\aleph),\aleph)\simeq C_\rms^{-\curnu y_\star}\aleph^{\curchi\curnu y_\star}\theta_\star$
and $\frakX_\aleph(\tiT(\aleph),\aleph)\simeq  C_\rms^{-\curnu}\aleph^{\curchi\curnu-1}$.  On this curve  $\theta_\star\ne 0$, and if it is supposed that  the two shift functions have the same
asymptotic dependence on $\aleph$, the  shift exponent will be related to $y_\vtmT=1/\curnu$ and $y_\star<0$  by $\curchi=y_\vtmT/(1-y_\star)$ with the shift amplitude
$C_\rms\simeq [\frakX_\star(\tiT(\aleph),\aleph)/\frakX_\aleph(\tiT(\aleph),\aleph)\theta_\star]^\curchi$.

\vspace{0.2cm}

\noindent As already indicated finite-size corrections to the pure power-law behaviour of $\varphi_\vtmT$, as described by (\ref{Ypowerlaw}), will
 begin to be observed  whenever the system is finite
(with $\theta_\aleph\cequals\aleph^{-1} \ne 0$) at the rounding temperature $\scT(\aleph)$.  It has been argued \citep{2-f9}
 that  this is the temperature at which the size $\aleph$ of the system is of the same order as the correlation length $\xi(T)$.\footnote{This can quite easily
 be established
 explicitly  for a one-dimensional Ising model on a ring of $N$ sites, where the magnetization density $m(T,\pH,N)$ is given by
 $m(T,\pH,N)=\tanh[N/2\xi(T)]\times m(T,\pH,\infty)$.}
  It follows from (\ref{excl})
 that $|\titheta_\vtmT(\scT(\aleph),\aleph)|^{-\curnu} \aleph^{-1}\simeq C$, where $C$ is a  constant,
 which establishes, from  (\ref{rounde}), that $C\cequals C_\rmr$ and the rounding exponent $\curtau=1/\curnu=y_\vtmT$
  with $\curomega=\curgamma\curtau$. Thus on the basis of some plausible assumptions we have the condition $\curchi<\curtau$, which, for large systems, motivates the disposition of the curves in
  Fig.\ \ref{Fig6}.
%-----------------------------------------------------------------------------------------------------
\subsubsection{Renormalization Schemes}\label{wtfn}
The practical implementation of  the renormalization group procedure in Sect.\   \ref{tgs}
involves the choice of a weight function and leads to recurrence relationships between systems  related
by a size parameter $\lambda$, together with a method for calculating the free-energy density which satisfies
the scaling relationship. In (a) and (b) in this section we give examples of the implementation of  two of the most commonly used weight functions
and in (c) we briefly outline a different scheme  which, using transfer matrix methods, relates the correlation lengths of systems of different
sizes.

For $d$-dimensional  lattices, most weight functions are based on a division of
the lattice $\mcN$ into equal blocks of  $\lambda^d$ sites. The mapping from $\mcN$ to $\wtmcN$  is given
 by associating each lattice site $\tbr \in \wtmcN$ with a  blocks of sites in $\mcN$  denoted  by ${\mcB}(\tbr)$.
\begin{enumerate}
\newcounter{wf}
\itemsep=0.3cm
\usecounter{wf}
\def\labelenumi{{\bf (\alph{wf})}}
%---------------------------------------------------------------------------------------------------------------
\item {\bf The decimation weight  function}

\vspace{0.1cm}

\noindent For this weight function the sites of $\wtmcN$ consist of a subset of the sites of $\mcN$,
chosen so that $\wtmcN$ forms  a lattice which is isomorphic to $\mcN$.
So we can take  $\tbr\in\mcB(\tbr)$ with
\begin{equation}
w(\bsigma,\tbsigma)\cequals\prod_{\{\tbr\}}\Krdelta(\tsigma(\tbr)-\sigma(\tbr))\, .
\label{twf1}
\end{equation}
The effect of this is that the summation on  the right-hand side of (\ref{rgt})
is a partial sum over all the sites of the lattice $\mcN$ except those
of $\wtmcN$. For a range of one-dimensional models (including the Ising and Potts
models), which can be solved exactly using transfer matrix methods, exact RSRG decimation transformations
can also be obtained.

\vspace{0.2cm}

\noindent For the one-dimensional ferromagnetic case of the Ising model  it can be shown \citep{2-n12,lavisnew} that the
most convenient variables are not those given in Appen.\ \ref{tim} but rather  $\zeta_1\cequals \tanh([2J+H]/2T)$, $\zeta_2\cequals \exp(-2H/T)$,
and for $\lambda\cequals 2$, with the partial summation
in (\ref{rgt}) over alternate sites, the recurrence relationships
take the form
\begin{equation}
\tzeta_1= {\frac{4{\zeta_1}^2 -(1-{\zeta_2})({\zeta_1}^2-1)}
{4+(1-{\zeta_2})({\zeta_1}^2-1)}},\pairsep
\tzeta_2 = {\frac{{\zeta_2}^2{(1 +{\zeta_1})}^2+{(1-{\zeta_1})}^2}
{2(1+{\zeta_1}^2)}}.\label{odim13}
\end{equation}
It is then not difficult to show that there is a fixed point $\zeta_1=\zeta_2=1$  ($T=H=0$),
with both scaling exponents equal to $d=1$.  As we saw in the discussion of scaling theory
in Sect.\ \ref{scathr}, an exponent equal to the dimension of the system is indicative of the possibility
of a first-order transition. In this case the critical point is at zero temperature on the zero field line,
meaning that the first-order coexistence curve has contracted to a point coinciding with the critical
point at zero-temperature. At this point there is a first-order transition across the zero-field axis with a change of sign of the magnetization.
 It can also be shown that the curve
\begin{equation}
\zeta_2=\left(\frac{1-\zeta_1}{1+\zeta_1}\right)^2,
\label{odim14x}
\end{equation}
which corresponds to the interaction $J$ between microsystems being set to zero,
is invariant under (\ref{odim13}).  At  every point it has exponents $0$ and $-\infty$; the first of these is marginal,
which  indicates that the line consists of fixed points, and the latter that it is `infinitely attractive'  to points not on the line.
The end points of the line $\zeta_1=0$, $\zeta_2=1$  ($H=\infty$, $T=0$) and $\zeta_1=1$, $\zeta_2=0$ ($T=\infty$, $H=0$)
are fixed points in their own right in the invariant subspaces $T=0$ and $H=0$ respectively.  The phase diagram is shown
in Fig.\ \ref{Fig7}.
%---------------------------------------------------------
\begin{figure}[t]
\includegraphics[width=6cm]{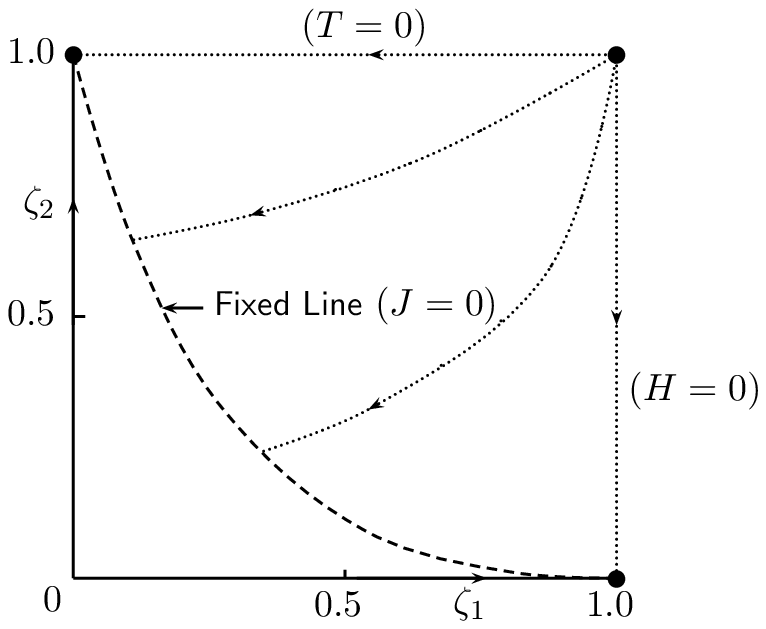}
\caption{The trajectory flows for the renormalization group
transformation of the one-dimensional Ising model.
}\label{Fig7}
\end{figure}
%----------------------------------------------------------

\vspace{0.2cm}

\noindent  Of course, for reasons just explained, the one-dimensional Ising model is less interesting than the two-dimensional model where
the ferromagnetic critical point is not at zero temperature. So, suppose that we try to carry out the same procedure in that case.  A possibility
is to choose blocks of two sites as shown in Fig.\ \ref{Fig8}.  The lattice $\wtmcN$ consists of the black sites and the partial summation
in (\ref{rgt}) is over the spin states on white sites.  But this will create an interaction between the {\em four} sites surrounding each white site.  So we would need to
back-track and increase $n$ from two to three,  inserting this interaction from the beginning.  But this would in turn generate an interaction between nine sites.  And so on.  This
proliferation of interactions is typical of the problems encountered with decimation.  The usual trick is to cut off the proliferation at a certain level.
Such an approximation for this model was investigated by \citet{2-w13}  with a rather poor outcome compared to the known exact results.
%-------------------------------------------------------------
\begin{figure}[t]
\includegraphics[width=8cm]{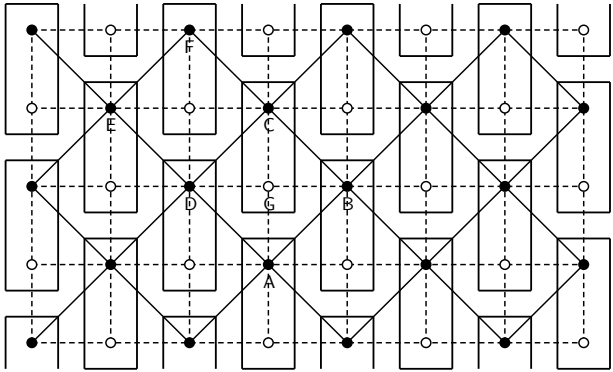}
\caption{Two site blocks for the first-neighbour Ising model
on a square lattice. The lattice $\mcN$ consists of both white and black sites and $\wtmcN$ of only black sites.}\label{Fig8}
\end{figure}
%---------------------------------------------------------------
\item {\bf The majority-rule weight function}
\vskip 0.1cm
This weight function was introduced by \citet{2-n13,2-n10}.
The first step in assigning $\tsigma(\tbr)$ for the block
$\mcB(\tbr)$ can be described in terms of the `winner takes
all' voting procedure used in some democracies. Given that each microsystem has $\nu$ states and that among the
sites of $\mcB(\tbr)$ one of the $\nu$ state occurs
more that any other, ${\tsigma}(\tbr)$ is assigned to have
this value. If $\nu\cequals 2$ and the number of sites $\lambda^d$ in a
block is odd this rule works; a case in point being the treatment of the  Ising model
on the triangular lattice with a block of nine sites ($\lambda\cequals 3$) by \citet{2-s16}.
But unless these conditions hold it is clear that the simple majority rule is not sufficient to determine
$\sigma(\tbr)$ for every configuration of the block. A `tie'
can occur in the voting procedure and a strategy must be adopted to
deal with such cases. One possibility is to assign to
$\sigma(\tbr)$ one of these predominating values on the
basis of equal probabilities.
In some cases this may not, however, be the most appropriate choice.
In their work on the Ising model using a square
first-neighbour block  ($\lambda\cequals 2$) \citet{2-n14} divided the
configurations with equal numbers of up and down spins between
block spins up and down with equal probabilities. The rule (one of four)
which they chose ensured that the reversal of all the spins in the
block reversed the block spin.
%----------------------------------------------------------------------------------------------------------
\item {\bf Phenomenological renormalization}
\vskip 0.1cm
The idea of finite-size scaling,  introduced in Sect.~\ref{fssc}, leads quite naturally \citep[][Sect.\ IV]{2-b19}
to the  RSRG method  developed by \citet{2-n16}.
The essential feature of finite-size scaling  is that, for a $d$-dimensional system,    infinite in $\frakd$
dimensions and of  thickness $\aleph$, the quantity $1/{\aleph}$ is treated as an additional scaling field $\theta_{\aleph}$.  If attention is restricted
 to the simple magnetic system with the two other scaling fields $\theta_\vtmT$ and $\theta_\vtpH$,
 the response function $\varphi_\vtmT$  satisfies the scaling relationship (\ref{fss3p}).
 A similar inclusion of $\theta_\aleph$ in the scaling relationship (\ref{corlsc})  for the correlation length gives
 \begin{equation}
\curxi(\lambda^{y_\vtmT}\theta_\vtmT,\lambda^{y_\vtpH}\theta_\vtpH,\lambda\theta_\aleph)=\lambda^{-1}\curxi(\theta_\vtmT,\theta_\vtpH,\theta_\aleph).
\label{corlsfc}
\end{equation}
With the slight change of notation  ${\curxi}^{({\aleph})}(\theta_\vtmT,\theta_\vtpH)\cequals
{\curxi}(\theta_\vtmT,\theta_\vtpH,\theta_\aleph)$,  (\ref{corlsfc}) can be regarded as relating
the correlation lengths of two similar systems denoted by ${\mcL}_{\frakd}({\aleph})$
and ${\mcL}_{\frakd}(\widetilde{\aleph})$ with couplings $\zeta_\vtmT,\zeta_\vtpH$ and $\tzeta_\vtmT,\tzeta_\vtpH$ and thicknesses
$\aleph$ and $\widetilde{\aleph}\cequals \aleph/\lambda$,   $\lambda>1$,  respectively.
 The relationship  (\ref{corlsfc}) can be reexpressed as
\begin{equation}
{\curxi}^{({{\aleph}})}(\theta_\vtmT,\theta_\vtpH)
=\lambda{\curxi}^{({\tilde{\aleph}})}({\tilde\theta}_\vtmT,{\tilde\theta}_\vtpH),
\label{pr2}
\end{equation}
where
\begin{equation}
{\tilde\theta}_\vtmT=\theta_\vtmT({\tilde\zeta}_\vtmT,{\tilde\zeta}_\vtpH)= \lambda^{y_\vtmT}\theta_\vtmT(\zeta_\vtmT,\zeta_\vtpH),\tripsep
{\tilde\theta}_\vtpH=\theta_\vtpH({\tilde\zeta}_\vtmT,{\tilde\zeta}_\vtpH)= \lambda^{y_\vtpH}\theta_\vtpH(\zeta_\vtmT,\zeta_\vtpH)
\label{pr3}
\end{equation}
relate the scaling fields  for ${\mcL}_{\frakd}({\aleph})$ and ${\mcL}_{\frakd}(\widetilde{\aleph})$.
These relationships
form the basis of  \citeauthor{2-n16}'s  phenomenological renormalization  method, where the correlation lengths for systems of the two widths are obtained
from transfer matrix calculations using  (\ref{atm43}).  In the case of one scaling field ($\pH=0$) the method yields the critical temperature fixed point
$\theta^\star_\vtmT\cequals\theta_\vtmT=\tilde{\theta}_\vtmT$
 and the thermal exponent for a number of different models \citep{2-n16,2-n17,2-s21,2-k14}, which in the case where exact results are know are at a
  high level of accuracy.\footnote{However, if  there is more than one coupling ($\pH\ne 0$) then substituting $\theta^\star_\vtmT\cequals\theta_\vtmT=\tilde{\theta}_\vtmT$
  and $\theta^\star_\vtpH\cequals\theta_\vtpH=\tilde{\theta}_\vtpH$ in  (\ref{pr2}), does not completely define the critical point.
 A number of methods for dealing with this case have been developed, \citep{2-s21,2-k14}}
\end{enumerate}
%-----------------------------------------------------------------------------------------------------
\subsection{The Thermodynamic Limit}\label{tlnism}
In the development of statistical mechanics represented by the right-hand column in Fig.\ \ref{Fig1} system size appears twice. Firstly in the passage for $\textsf{SM1}$ to $\textsf{SM2}$,
where the system becomes large yielding approximate extensivity.  This is needed for the discussion of finite-size phase transitions represented by $\textsf{SM5}$.
Secondly in the other branch from $\textsf{SM2}$, via the thermodynamic limit, to an infinite system represented by $\textsf{SM3}$.
  This entails the identification of the infinite statistical mechanical  system \textsf{SM3} with thermodynamics, or at least the version, labelled  \textsf{TD3} in Fig.\ \ref{Fig1},  of thermodynamics with
some PTCP defined.  But are \textsf{SM3} and \textsf{TD3} actually identical?  The answer is clearly `no'.
\textsf{TD3}  is the result of a development in the left-hand column of Fig.\ \ref{Fig1}, from the basic structure through the assumption of extensivity
to a grafting on of a picture of PTCP, in the manner of \citet{1-p1} or \citet{2-b2}.   On the other hand, as we have just indicated, \textsf{SM3} is the result of a statistical
mechanical development in the right-hand column in Fig.\ \ref{Fig1}.  It retains its microstructure with a probability distribution, and in most cases it is the result of the implementation of the thermodynamic
limit for a particular model, the most well-known examples being  the two-dimensional zero-field  Ising model and the eight-vertex model.
 Thus it should be recognised for later reference (see Sect.\ \ref{limred})  that
this way of understanding the relation between thermodynamics and statistical mechanics involves the unwarranted  \emb{conflation of two quite different pictures}.  Although one can argue that \textsf{SM3} is an enrichment of \textsf{TD3}, since the former has all the features
of the latter together with the extra ones provided by microstructure and precise results concerning critical values and exponents.
 That having been said, one may still question whether the thermodynamic limit is:\footnote{A third possibility that the thermodynamic limit is neither necessary nor useful can surely be discounted, with
  respect to usefulness, after a cursory survey of the corpus of work on statistical mechanics.}
\begin{enumerate}[(1)]
\item Necessary, in principle,  because statistical mechanics is not complete without it.
\item Useful because  calculations become much simpler in the thermodynamic limit and the relationship  \FSM--\refsf{fsm-03} of \textsf{SM3} to \textsf{TD3} makes it
easier to identify the order of phase transitions.
\end{enumerate}
Although both of these possibilities  deserve consideration it is the the first which has received the most attention, principally because of the  role of the thermodynamic limit
in the understanding of PTCP; this will be discussed in detail is Sect.\  \ref{nftl}.

In this work we propose, in Sect.\ \ref{ourp},  a particular view of the usefulness of the thermodynamic limit in the context of phase transitions in finite systems.  However, it
is pertinent  to note the range of possible circumstances calling for the use of the thermodynamic limit.  In particular one might suppose an additional
kind of necessity interposed between the two items in our list:

\vspace{0.2cm}

\noindent (1a)  Necessary in practice,  because calculations for particular models are not tractable without its use.

\vspace{0.2cm}

\noindent However, of course, tractability, and hence necessity in practice, is ephemeral, evolving (one might hope) with an increase in computing power and technical ingenuity  into mere usefulness.
%-----------------------------------------------------------------------------------------------------------------------------------
\subsubsection{Phase Transitions in Infinite Systems}\label{nftl}
The argument for the necessity  in principle of the thermodynamic limit for PTCP effectively  involves asserting the truth of the contradictory set of propositions:
 \begin{description}
\item[\textbf{P--IA:}]  PTCP occur in nature.
\item[\textbf{P--IB:}]  PTCP occur in nature as discontinuities in densities (first-order transitions)
and as singularities in response functions  (higher-order transitions).\footnote{This latter group also includes other sorts of weaker
singular behaviour.}
\item[\textbf{P--IIA:}] PTCP in thermodynamics are defined by  singularities in derivatives of first or higher order
in the free energies and are treated as such using scaling theory.
\item[\textbf{P--IIB:}] PTCP must necessarily be represented in thermodynamics by  singularities.
\item[\textbf{P--IIIA:}]  PTCP should be able to be modelled in statistical mechanics.
\item[\textbf{P--IIIB:}]  PTCP should be modelled in statistical mechanics in the same way that they are  in thermodynamics.
\item[\textbf{P--IV:}] Real systems are of finite size.
\item[\textbf{P--V:}] Thermodynamic functions for finite systems in statistical mechanics are regular functions.
\item[\textbf{P--VI:}] Thermodynamic functions for infinite systems in statistical mechanics can show singularities.
\end{description}
For later use it is relevant to compare this list with that of \citet[][p.\  589]{cal1} \citep[repeated by][pp. \ 13--14]{main2}:
\begin{description}
 \item[\textbf{CP--I:}] Real systems have finite [size].
 \item[\textbf{CP--II:}] Real systems display phase transitions.
 \item[\textbf{CP--III:}] Phase transitions occur when the partition function has a singularity.\footnote{In relation to this statement, see footnote \ref{zeros}.}
 \item[\textbf{CP--IV:}] Phase transitions are governed/described by classical or quantum statistical mechanics (through [the partition function]).
 \end{description}
A number of items in our list are indisputable and are not included in \citeauthor{cal1}'s list:
\begin{enumerate}[$\bullet$]
\item That PTCP are defined in thermodynamics  by singularities, can be confirmed
 by a visit to the thermodynamics section of any academic library ({\bf P--IIA} is true).
Whether  it is necessary for thermodynamics to be formulated in this way
 (that   {\bf P--IIB}  should be accepted),  given a possible denial  that   PTCP occur in nature as singularities
(that {\bf P--IB} is true)  is a different question.
\item The joint assertions that thermodynamic functions  are regular for finite systems but can have singularities for infinite systems
(included in our list as {\bf P--V} and {\bf P--VI}, respectively, but not contained  in   \citeauthor{cal1}'s list)  are facts
about the mathematical structure of statistical mechanics which cause the total list to be contradictory.
\end{enumerate}

\noindent And on \citeauthor{cal1}'s list:
\begin{enumerate}[$\bullet$]
\item It is difficult to argue that phase transitions do not occur in real systems (that  {\bf P--IA (CP--II)} is false), although it is plausible to deny
 that they arise as some kind of singularities (to argue that {\bf P--IB} (not in \citeauthor{cal1}'s list) is false),
 on the grounds that a first-order transition (say that between liquid water and water vapour) may look like a sudden change of density,
  but on closer observation would turn out to be a very steep continuous change.  Likewise, apparent singularities in compressibility
  in fluids and susceptibility in magnets may just  be very steep maxima.
\item It is also difficult to argue that real systems are not finite (that  {\bf P--IV  (CP--I)} is false),
 given that no system takes up the whole of the universe.\footnote{Which, in any event, may be finite.} A sort of argument  could be
 constructed on the basis that no system is completely isolated, but this would mean accepting
the need for computation, not with an infinite system as envisaged here, but with a system joined to a complicated and largely undetermined environment.
\item If the ability  to model PTCP  were not deemed to be a necessary part of statistical mechanics ({\bf P--IIIA   (CP--IV)} is rejected), then most of the work on statistical mechanics in the last half century
and more would be pointless. It is, however, relevant here to mention the work of the late Ilya Prigogine \citep[in particular,][]{prig7}. Although, in a sense he accepts {\bf P--IIIA},
it is a radically different form of statistical mechanics that he has in mind. From the assertion that ``[a]s long as we consider merely a few particles, we cannot say if they form
a liquid or gas'' (ibid, p.\ 45) he concludes that  ``[s]tates of matter as well as phase transitions are ultimately defined by the thermodynamic limit. $\ldots$
Phase transitions correspond to emerging properties. They are meaningful only at the level of populations and not of single particles'' (op. cit).
This entails for him the reformulation of statistical mechanics so that the underlying dynamics in not that of trajectories but of measure.\footnote{This being the approach that he
and his Brussels group also used to resolve the problem of irreversibility \citep[see, for example,][]{prig2}.}
\end{enumerate}
There remain {\bf P--IB}  and {\bf P--IIIB}, which together with {\bf P--IIB} is equivalent to {\bf CP--III}, and we now consider the consequences of denying one or both of them.
\begin{enumerate}[(i)]
\item If {\bf P--IB} is accepted, that is that PTCP in nature do occur as singularities,   then it is clearly necessary for thermodynamics to represent them in this way; {\bf P--IIB} must be accepted. Then we seem to be driven toward the conclusion that statistical mechanics should model them in the same way (that is the acceptance
of {\bf P--IIIB}) which leads back to the contradiction.  This  is avoided by  denying {\bf P--IIIB}. Then PTCP can be modelled in statistical mechanics without singularities, by, for example, transfer matrix methods,
while at the same time admitting that this is not the situation in reality.
\item   If {\bf P--IB} is denied then it can be argued either:
\begin{enumerate}[(a)]
\item That it  is not necessary for thermodynamics to model PTCP as singularities ({\bf P--IIB} is false).  In this case  {\bf P--IIIB} can be accepted,  with PTCP modelled without
singularities in statistical mechanics, with thermodynamics reformulated to do the same.\newline
\vspace{-0.25cm}
\hspace{-0.7cm}or
\vspace{0.5cm}
\item  That in statistical mechanics PTCP should be modelled without singularities, but because for large systems steep maxima in response functions and steep changes in densities
look very much like singularities and discontinuities, it is still necessary (on the grounds of tractability and simplicity) to model PTCP in thermodynamics as singular behaviour; {\bf P--IIB} is accepted and
 {\bf P--IIIB} is rejected.\footnote{To preview Sect.\ \ref{ourp}, this is the position we shall defend.}
\end{enumerate}
\end{enumerate}
 So given that all of  {\bf P--I} to {\bf P--VI} are accepted  is there any way out of the paradox?  One radical approach, which has already been noted,  is that  due to \citeauthor{prig7},
 where statistical mechanics is reformulated to `build in' the thermodynamic limit.\footnote{This involves an extension of the \citet{koop} formulation to a space beyond the Hilbert
 space in which it is set.}  Somewhat similar, but less radical, is the approach of Robert \citeauthor{bat9}, a philosopher of physics who has written extensively  on questions related
  to phase transitions, the renormalization group and the thermodynamic limit \citep{bat9,bat11,bat5,bat6,bat10,bat12}.  Rather than formulating a novel form of the mechanics underlying
 statistical mechanics,  his argument,  following the lead of \citet{kad2}, is that the renormalization group is itself a novel approach, revolutionary in the sense of \citet{kuhn},  which has the thermodynamic  limit built in.  His starting point  is that thermodynamics\footnote{For reference in the summary (a)--(f) of his position  on the renormalization group in Sect.\  \ref{tl&rg} the
  quotations from \citeauthor{bat11}'s work are given labels \bfA--\bfF.}
  \begin{quote}
 `` is correct to represent [phase transitions] mathematically as singularities.''   (\bfA: \citealt[][p.\ 234]{bat11}.)
\end{quote}

 And:
\begin{quote}
 ``Further, without the thermodynamic limit, statistical mechanics would completely fail to capture a genuine feature of the world. Without the thermodynamic limit,
   in fact, {\em statistical mechanics is incapable even of establishing the existence of distinct phases of systems.''} (\bfB: op. cit.)\,\,
\end{quote}
  If there is any doubt about his view of real systems, this is dispelled by
his forthright assertion that he wants
\begin{quote}
``to champion the manifestly outlandish proposal that despite the fact that real systems are finite, our understanding of them and their behaviour requires, in
a very strong sense, the idealization of infinite systems and the thermodynamic limit.'' (\bfC:  ibid, p. 231.)
\end{quote}
`Outlandish' or not his position is one which would appear, in our experience, to be that adopted implicitly or explicitly
 by many working physicists, including, albeit in a radical sense as indicated above, by  \citeauthor{prig7}, and \citet[][p.\ 238]{kadbk},\label{kadquote1} who asserts that the `` existence of a phase
 transition requires an infinite system. No phase transitions occur in systems with a finite number of degrees of freedom.'' \citeauthor{kad1} calls this the ``extended singularity theorem''
\citep[][pp.\  154--156]{kad2}  because ``these singularities have effects that are spread out over large regions of space'' \citep[][p.\  24]{kad1}.
Having asserted that
\begin{quote}
 ``the idea that we can find analytic partition functions that ``approximate'' singularities is mistaken, because the very notion of approximation required fails to make
sense when the limit is singular, [which it is in this case  because the]
 behaviour at the limit (the physical discontinuity, the phase transition) is {\em qualitatively} different from the limiting behaviour as that
limit is approached.'' (\bfD: ibid, p.\ 236)
\end{quote}
\citeauthor{bat11}'s  proposal for resolving the puzzle is to resort to the renormalization group. In the next section this possibility  is examined.
%---------------------------------------------------------------------------------------------------
\subsubsection{Infinite Systems and the Renormalization Group}\label{tl&rg}
`Infinity'  as it arises in accounts of renormalization group methods consists not so much in the limiting process, evident in, say,  \citeauthor{1-o2}'s solution
of the zero-field two-dimensional Ising model, whereby the dimensions of the system are taken to infinity, but rather in the perception that to make the method
intelligible one must be working with a system which is already infinite \citep{pala2,pala1}.  To spell this out, a renormalization group scheme consists of the following:
\begin{enumerate}[(i)]
\item In the space of couplings (or of functions thereof)  a semigroup of transformations is derived which generates recurrence relationships under which
any critical regions are invariant.
\item In this `dynamic system'  the critical regions are the basins of attraction of critical fixed points. And there are sinks associated with non-critical regions (phases)
of the system.
\item A critical fixed point determines the universality class of the system at each point in its basin of attraction,  with an associated set of critical exponents.
\item In general a system may be able to be in more than one universality class  determined by the symmetry group of the Hamiltonian when there is a particular relationship
between the couplings.\footnote{ For example, a renormalization scheme for the spin-1 Ising model will have fixed points for both the spin-$\frac{1}{2}$ Ising model
 and the 3-state Potts model.}
\end{enumerate}
It is clear that this way to do statistical mechanics is very different from the standard procedures (mean-field and other classical approximations,  series expansions and exact solutions).
So much so that, as we have already indicated,  it is characterized by \citet{kad2} as a \citeauthor{kuhn}-type revolution,  a view endorsed by \citet{bat10}.  The argument presented by \citeauthor{bat10} concerning the whole question of singularites/real singular systems/the thermodynamic limit  needs to be carefully rehearsed  and for this his  \citeyear{bat10} tribute to Leo \citeauthor{kad2} provides
the clearest account.

 He presents his view in contradistinction to that of Jeremy \citeauthor{butt2} who contends \citep[][p.\ 1077]{butt2} that: ``The use of the infinite limit $\ldots$ is justified,
despite $N$ being actually finite, by its being mathematically convenient and empirically correct (up to the required accuracy).''  For an understanding of \citeauthor{bat10}'s view
two quotes are particularly useful.  In the first he asserts that:
\begin{quote}
  ``the RG is not just a theory of the critical point, but rather it is a {\em theory of the critical region.} And this covers large but finite systems. So contrary
to the line of reasoning presented [by \citeauthor{butt2}] the explanation of the behaviour of real finite systems requires the use of mathematical infinities, but does not require there to
{\em be} infinite real systems.'' (\bfE:  \citealt[][p.\ 571]{bat10}.)
\end{quote}
At this point we have cause to be grateful to a referee of his paper, who objected that this quote was actually in line with ``the claims of those supporting the idea that real phase
transitions aren't sharp''.  In response to this \citeauthor{bat10} added a footnote in which he clarified his position in the following way:
\begin{quote}
 ``It seems to me that if one is going to hold that the use of the infinite limits is a convenience, then one should be able to say how (even if inconveniently) one might go about finding a fixed
point of {\em the RG transformation} without infinite iterations. I have not seen any sketch of how that is to be done. The point is that the fixed point, as just noted, determines the behaviour of the
flow in its neighbourhood. If we want to explain the universal behaviour of finite large systems using the RG, then we need to find a fixed point and,
 to my knowledge, this requires an infinite system.''(\bfF: op. cit.)
\end{quote}
So to summarize his view (using the labelled  quotes \bfA--\bfF,  given above):
\begin{enumerate}[(a)]
\item Phase transitions are real discontinuities in experimental systems (\bfA). [An acceptance of {\bf P--IA,B} and {\bf P--IIA,B}].
\item  The thermodynamic limit is needed in statistical mechanics to exhibit phase transitions (\bfB). [An implicit acceptance of {\bf P--V} and {\bf P--VI} and an endorsement of {\bf  P--IIIA,B}.]
\item Real systems are finite but in order to understand them we need the idealization of infinite systems and the thermodynamic limit (\bfC). [An endorsement of {\bf  P--IV}, and more.]
\item The idea that the study of large systems can play a role here is wrong because the properties of large systems and infinite systems are qualitatively different (\bfD).
\item  To represent the situation correctly we need to engage with mathematical singularities but not real infinite systems (\bfE).
\item  We need infinite iteration (of the RG transformation) to obtain fixed points (and all the information they provide) (\bfF).
\end{enumerate}
And to summarize the summary of \citeauthor{bat10}'s position:
\begin{quote}
\emb{Although phase transitions in real systems are accompanied by singular behaviour, and in statistical mechanical models this singular behaviour is exhibited
only by infinite systems, we don't need infinite systems, just the use of mathematical singularities, these being required to derive the fixed points in renormalization group calculations.}
\end{quote}
At this point we wish to challenge the last part of this statement by providing\footnote{Based on practical experience \citep{y&l,s&l1,s&l2,l&s&b,l&so}.}
the `sketch' that \citeauthor{bat10} (quote \bfF) requires of the means of the determination of renormalization group fixed points.

\vspace{0.2cm}

\noindent  The first thing to note is that the recurrence relationships (\ref{recr12}) and (\ref{ecr0})  are derived (almost always with some
approximations involved) between the couplings of two {\em finite} systems with
sizes $N$ and $\wtN$ with $N/\wtN=\lambda^d>1$.  Once this is done no point in the space of couplings is intrinsically associated with a system of a particular size and, by the same
token,  fixed points,  obtained from (\ref{recr12}) with  $\tzeta_j=\zeta_j$,   are not associated with infinite systems.
However, {\em if we were to  choose} to associate a particular system-size $N$ with the first point of a trajectory,
 it would be necessary to assume only that we are working with  a system  large enough to allow the required number of
 iterations.\footnote{Here we are agreeing with \citeauthor{bat10} that an infinite system is not required.} (Hence the inclusion of  \textsf{SM2} in the path from    \textsf{SM1}
  to \textsf{SM4} in Fig.\ \ref{Fig1}.)
As \citet[][p.\ 222]{norton2} says, fixed points are the ``limit points'' of the sequences generated by the recurrence relationships; the ``mathematical pegs on which to hang limit properties''  which are
never reached in a finite number of iterations. They do not arise from an investigation of the properties of infinite limit systems, and, although they are properties of the transformation, iteration is not always needed for their determination.
 In some simple cases, like the one-dimensional Ising model described in Sect.\ \ref{wtfn},  the fixed points can be extracted by direct analytic solution of the fixed point equations.  But, in more complicated cases numerical computation comes into play.  Although in principle iteration of the recurrence relationships starting from a point in the basin of attraction of a fixed point will generate a
sequence of points approaching the fixed point, this is not usually a viable strategy for their determination. Since those of greatest interest,  associated with
critical regions, have both irrelevant directions of attraction within the critical region (the basin of attraction) and relevant directions along which the trajectory is driven away from the critical region.
Except in special cases it is difficult to start a trajectory in a critical region, but nearby points are useful and possible.  Then the trajectory will hover near the critical fixed point before it moves away
to the sink associated with the phase containing the trajectory.  These `hover points' can be spotted by inspection of the computer output and used as initial guesses for a numerical solution of the fixed
point equations. These kinds of numerical techniques, used also to map the critical regions themselves, provide a good picture of the whole phase diagram.  And linearization of the recurrence
relationships about the fixed points allows the critical exponents to be obtained.
%--------------------------------------------------------------------------------------------------------------------------------------------
\subsection{Phase Transitions in Finite Systems: Mainwood's Proposal}\label{ptfs}
Given, as we have concluded in the previous section,  that the thermodynamic limit is not necessary to enable renormalization group calculations
to provide the PTCP structure, is it still useful  in other statistical mechanical treatments of PTCP?
  An assessment of usefulness, as distinct from necessity,  is obviously  heavily influenced by the position adopted with respect to whether PTCP occur in nature as singularities ({\bf P--IB}).
 If  it is false and real systems, by virtue of their size ($\sim 10^{23}$ microsystems) exhibit behaviour approximating to singular behaviour, in the sense, say, that the maximum in the compressibility of a fluid is experimentally indistinguishable from a singularity, then we have the means to remove the contradiction in the set of statements at beginning of Sect.\ \ref{nftl}.
One way would be to deem it  unnecessary for  PTCP to be treated as singularities in thermodynamics (a denial of  {\bf P--IIB}).
 Although this would allow thermodynamics and statistical mechanics to be modelled in the same way (for {\bf P--IIIB} to be accepted)
 we would argue,  for the reasons given in Sect.\  \ref{ourp}, that it is not a tenable possibility.

The alternative, which is the one discussed in this section, and which is favoured by ourselves, is to accept that thermodynamics
 must represent PTCP in terms  of singularities ({\bf P--IIB}) on the basis that this is
 an appropriate  approximation to real systems. Thus rejecting the assertion that thermodynamics and statistical mechanics must model PTCP in the same way ({\bf P--IIIB}), since statistical mechanics   models phase transitions in finite systems.
Given that real systems are very large (in terms of the number of microsystems) and finite, with phase transition giving the appearance, but not the exact reality of singularities,
 can calculations avoid using the thermodynamic limit? Or,   more generally can recourse to a system where PTCP occur as singularities be avoided? Here  we examine
a proposal of  \citet{main2} which definitely answers the question in the negative and in the next section we propose  an answer which  is more nuanced.

\vspace{0.2cm}

\noindent The definition of a phase transition  provided by \citeauthor{main2} (ibid, p.\ 28)  can\footnote{With some changes of notation to give conformity with our usage.}  be described in the following way.
    For a statistical mechanical system $\frakS_\tmN$    of size $N$ with partition function $Z_2(\zeta_1,\zeta_2,N)$, the free energy $\Phi_2(\zeta_1,\zeta_2,N)$ is given
by (\ref{gibfe})  and satisfies (\ref{althsp31}) and (\ref{althsp41}).\footnote{We have chosen the system with two independent couplings for convenience.}  Suppose that the thermodynamic limit
\begin{equation}
\lim_{N\to\infty}\frac{\Phi_2(\zeta_1,\zeta_2,N)}{N}= \phi_2(\zeta_1,\zeta_2)
\label{entdenlfe}
\end{equation}
exists, with $\phi_2(\zeta_1,\zeta_2)$  the free-energy density of the system $\frakS_\infty$. Then:
\begin{definition}\label{mdeAf}
$(\zeta_1,\zeta_2)$  is a point with a particular criticality for $\frakS_\tmN$ iff  $(\zeta_1,\zeta_2)$ is a point where $\frakS_\infty$ has a singularity associated with this same
criticality.
\end{definition}
  And  \citeauthor{main2} (ibid, p.\ 29) asserts that:\footnote{\label{zeros}In relation to both this assertion and CP--III, the following quibble might not be out of place.
The thermodynamic limit is taken for {\em thermodynamic functions}  which are approximately extensive for large systems and become extensive in the thermodynamic
limit \citep{griff3}.  The partition function is not of this sort, as one can see by using a little `reverse engineering' to define the partition function of $\frakS_\infty$
as $Z_\infty(\zeta_1,\zeta_2,N)\cequals \exp\{-N\phi_2(\zeta_1,\zeta_2)\}$.   Apart from the retained dependence on $N$, a
 singularity, which is an infinity of $\phi_2(\zeta_1,\zeta_2)$ would be a zero of $Z_\infty(\zeta_1,\zeta_2,N)$.  In fact this brings to the fore a problem with CP--III.  Phase transitions
 do not correspond to points ``when the partition function has a singularity''.  For (say a lattice system) the partition function is, for a finite system, a polynomial whose zeros give singularities
 of the free energy, none of which lies on the positive real axis.  In the thermodynamic limit a phase transition corresponds to a point of accumulation of zeros on the real axis.
 The quibble is resolved by replacing `partition function' by `free energy' in CP--III and \citeauthor{main2}'s assertion.}
\begin{quote}
``Rather surprisingly, using this definition it is possible to hold on to all of \citeauthor{cal1}'s four statements [(given above as CP--I to CP--IV)] without contradiction; though only
in a Pickwickian sense --- it is a ``trick'' possible only due to his choice of wording. Namely, the singularity referred to in [CP--III] is one not in the partition function [of $\frakS_\tmN$]
but in [the partition function of $\frakS_\infty$].
\end{quote}
If this is regarded as a positive point in favour of \citeauthor{main2}'s definition, the overall conclusion seems to be more mixed.
 \citeauthor{main2} `worries' that:\footnote{It is convenient to take his worries in reverse order.}
\begin{enumerate}[(1)]
 \item The definition means that a phase transition can be predicted in a finite system, however small it might be  (ibid, p.\ 32).
\item  ``While there exist standard procedures for taking the thermodynamic limit, $\ldots$ these procedures are human inventions, and choices could
 have been made differently. $\ldots$ The definition of a phase transition thus seems arbitrary in a disastrous sense: we can choose whether one is occurring
 or not by modelling it differently, or taking the limit according to a different scheme'' (ibid, p.\ 31).
\item ``[T]he facts we need to decide whether or not [a physical system] is undergoing a phase transition should be physical facts, about actual
states of affairs $\ldots$ They should not exist only in an idealized model on a theoretician's blackboard'' (ibid, p.\ 29).
\end{enumerate}
Although \citeauthor{main2} adds (1) as a final difficulty it is probably the one which would first spring to mind, since the definition would imply a phase transition in
an Ising model of four spins in a square at the critical temperature given by \citeauthor{1-o2}'s solution.  \citeauthor{main2} thinks that ``this bullet can and should be
bitten'' (ibid, p.\ 32),  but the consequences are not, we think, ones which would recommend themselves to any working physicists;  not to put too fine a point on it, they
would bring chaos to discussions of critical phenomena.  The tractable alternative, also suggested by \citeauthor{main2}, is to restrict
 the definition to large systems.\footnote{This is also discussed by \citet{ard1}, who proposes to use the Lee--Yang formulation of phase transitions in terms of the zeros of the partition function to describe, and to understand, the emergence of anomalies of thermodynamic functions in terms of accumulations of Lee--Yang zeros in the vicinity of the critical temperature  on the real axis in the complex temperature plane.  While this does indeed provide a useful intuition, it is not substantially different from exactly solving the statistical mechanics of finite systems, and does not by itself allow us to predict the way in which anomalies approach singularities of the infinite system, when the system size is increased.}
This would seem to us to be an inevitable step, but it also has consequences which we discuss in more detail below.

At one level both (2) and (3) are examples of the standard concern with respect to modelling, namely that we may not have a very good model which is not giving
results which agree with experiment.  And \citeauthor{main2}'s response to this  is, as would be expected, that we should find a better model.  But worry (3)
also contains a second element, namely that his definition contains the use of a counterfactual,  an idealized infinite model. His argument here is more complex and draws
on a strong parallel with \citeauthor{lewd2}'s  (\citeyear{lewd2}) analysis of counterfactuals. On this basis he argues that
\begin{quote}
 ``it is the character of [the real finite system] that determines the nature of the infinite system that we then consider. When we draw conclusions about the nature of the phase transitions, they are
 conclusions about the character of [the real finite system], but by reference to the infinite model we can express them in a concise and illuminating form'' (ibid, p. 30).
\end{quote}
However we have worries of our own which do not seem to concern \citeauthor{main2}.  These can best be described by  considering the transfer matrix treatment in Sect.\ \ref{extls},
where, if we restrict attention to the two-dimensional square-lattice spin-$\frac{1}{2}$ Ising model in zero field, the exact critical temperature is known for the  model on an infinite lattice
(see Appen.\ \ref{tim}).
To apply the transfer matrix method (see Sect.\ \ref{extls}), the square lattice is taken to have $N_{\tH}$ sites in the horizontal direction and $N_{\tV}$ sites in the vertical direction, so that $N=N_\tH N_\tV$.
  Periodic boundary conditions are applied so that the lattice forms a torus with  horizontal rings of $N_\tH$ sites and  rings in a vertical plane of $N_\tV$ sites.
It is assumed that the system is large in the horizontal direction, so that,  parameterized by $N_\tV$, we have a sequence of one-dimensional models of increasing complexity.
 Each  exhibits a maximum in the heat capacity, including the simplest case  $N_\tV=1$   \citep[][p.\ 166]{1-d4}.\footnote{Although we have shown by an exact renormalization group method in
 Sect.\ \ref{wtfn} that when $N_\tV=1$ the critical fixed point is at zero temperature.}
These maxima (although they will differ slightly for all $N_\tV$ however large and finite) are taken as incipient singularities\footnote{As we have indicted in Sect.\  \ref{extls} maxima
in other response functions and also the behaviour of the ratio of the two largest eigenvalues of the transfer matrix can also be used as identifiers of incipient singularities.}
 and for increasing $N_\tV$ show good agreement with the \citeauthor{1-o2} result, which is the case $N_\tV=\infty$.

 However, the prescription to be applied by the \citeauthor{main2} proposal
is that their critical temperatures, {\em for all} $N_\tV$, are the \citeauthor{1-o2} value.  This would seem to us to reverse the order of the way of working of physicists.
We think it is probably true to say that, with notable exceptions like   \citet{kad3,kad1,kad2}, physicists involved in model calculations do not consider whether their interest is in very large systems
 or infinite systems.
  Their concern is whether a phase transition occurs. If they suppose that it does, one tool\footnote{Among others, including duality transformations and series expansions.} to determine its location is to use transfer matrix calculations
 \citep{RC1,RCS1,BN1,OB1,OCB1,lavis10}. The method is to determine incipient singularities for as large a vertical width of system as possible as an
estimate for the transition temperature for a very large/infinite width.  Here one cannot use \citeauthor{main2}'s prescription to assign the infinite-width result to the finite-width systems,
since the former is not known.\footnote{Although one can, in some cases, prove the existence of a phase transition, even if the transition temperature
is not known \citep{2-p8,2-h5,2-h6,2-h7}.}  When, as in the case of the zero-field spin-$\frac{1}{2}$ Ising model, the infinite-width result is known exactly or has been
 determined to a good approximation by series methods, the motivation for determining finite-width results is to test the efficacy of the method, or to cross-check with other results.

 In his discussion of \citeauthor{main2}'s proposal \citet[][p.\ 1130]{butt2} states it in a more restricted form. Again using our notation this is:
 \begin{definition}\label{mbdef}
A phase transition occurs in $\frakS_\tmN$  iff   $\frakS_\infty$ has non-analyticities.
\end{definition}
 This \citeauthor{main2}--\citeauthor{butt2} proposal has the advantage that it doesn't project a result from the infinite system onto finite systems of any size (or maybe onto just large-size systems).
 However, given that it asserts the existence of a phase transition in a finite system of any size $N$, where does this occur?  At the maximum of one of the response functions (heat capacity
 or susceptibility/compressibility), or by  extraction from the behaviour of the ratio of the two largest eigenvalue of the transfer matrix? These will all give different results,
 as will also the results of taking the limits in different ways and for differing numbers of dimensions, all of which in turn will differ with $N$.  If all these values are taken to be estimates
 of some `true' value will this be $N$-dependent or the same for all $N$,   including presumably $N=\infty$, when we would be back with the problems of \citeauthor{main2}'s original proposal?
%-------------------------------------------------------------------------------------------------------
 \section{Phase Transitions in Large Systems: Our Proposal}\label{ourp}
As we shall see, our discussion in previous sections of the structure of  thermodynamics and of statistical mechanics in general, and of PTCP in particular, will allow us to paint a more nuanced and quantitative picture of their relationship than that provided  by previous approaches.   In particular we are concerned with the role in that relationship
played by large finite systems. \citeauthor{main2} suggests that we `bite the bullet' by countenancing the possibility of  phase transitions in small systems. However, we suggest that he is proposing to
 bite the wrong bullet.  The one which should be bitten is the need for a criterion giving a demarkation in system size between small systems and large systems,
 and our proposal,  which uses the discussion of finite-size scaling in Sect.\ \ref{fssc},  is intended to encompass this need.

 \vspace{0.2cm}

\noindent Thermodynamics, on the one hand, characterises PTCP in terms of singularities of thermodynamic functions, which may occur at special values of externally controllable parameters. This characterisation appears, {\em at first sight\/}, to be warranted by the phenomenology of phase transitions as they are observed in nature -- apparent discontinuities of thermodynamic functions at first-order phase transitions, and apparent algebraic singularities of thermodynamic functions including divergent response functions at second-order phase transitions.
In statistical mechanics, on the other hand, singularities of thermodynamic functions can emerge {\em only\/}  in the limit of infinite system size. As realistic systems are clearly of finite size, this creates an internal inconsistency in the list {\bf P--I} to {\bf P--VI} of propositions given above, if indeed the characterisation of PTCP {\em as they occur in nature\/} in terms of singularities (that is proposition {\bf P--IB}) is accepted.

Our aim now is to present an argument, based on the account of finite-size scaling in Sect.\ \ref{fssc},  which shows  that this inconsistency can be resolved within statistical mechanics  {\em and in a fully quantitative manner}.  In  Sect.\ \ref{fssc}, and also here,  discussion is restricted to a system with a thermal coupling $\theta_\vtmT$ and a magnetic coupling $\theta_\vtpH$,
in the cases where (i)  it is fully-finite with thickness $\aleph$ and (ii) it is fully-infinite with $\aleph=\infty$.  In  case (ii) on the zero-field axis $\pH=0$, $\theta_\vtpH=0$ there is a critical  temperature $T=T_\rmc$ with $\theta_\vtmT=0$ where response functions are singular. There is no singularity in the finite system but maxima
appear in the response functions.  We now summarize the relevant conclusions of finite-size scaling:
\begin{description}
\item[\textbf{FSS--I:}] In the thermodynamic limit $\aleph \to \infty$ when $\theta_\vtmT$ is small, but not infinitesimal,  the asymptotic form for the susceptibility at $T=T_\rmc$,
 given by (\ref{varphiThdLim}), has a singular component
with exponent $\curgamma$, but amplitudes which, by virtue of the presence of an irrelevant field $\theta_\star$, are dependent on $\theta_\vtmT$.
\item[\textbf{FSS--II:}]  As $\theta_\vtmT\to 0$, the influence of $\theta_\star$ becomes negligible and the susceptibility exhibits a
 pure power-law singularity at $T=T_\rmc$ as described  by (\ref{Ypowerlaw}).
\item[\textbf{FSS--III:}] When $\aleph$ is finite there is no singular behavior and two temperatures are defined:\footnote{Each will, of course, depend of the particular response function under consideration.}
 the shift temperature $\tiT(\aleph)$ where the susceptibility
has a maximum and the rounding temperature $\scT(\aleph)$ at which the profile of the susceptibility in the finite system begins to diverge from that in the infinite system.
\item[\textbf{FSS--IV:}] Assuming, as in (\ref{roundd}) and  (\ref{rounde}), that $|T_\rmc-\tiT(\aleph)|\sim\mcO(\aleph^{-\curchi})$ and $|\scT(\aleph)-\tiT(\aleph)|\sim\mcO(\aleph^{-\curtau})$,
it can be shown that the shift exponent $\curchi=[\curnu(1-y_\star)]^{-1}$ and the rounding exponent $\tau=\curnu^{-1}$; that is that the rate of convergence of both  the incipient singularity
and  the range of influence of finite-size effects around the incipient singularity are determined by exponents present  {\em in the infinite system}.
\end{description}
This  renormalization group scaling  approach to the description of critical phenomena thus  explains in a quantitative way, how singularities that might occur in infinite systems are
smoothed out by  finite-size effects. This, being fully in line with the fundamental observation that statistical mechanical systems of finite size cannot exhibit any singularities, resolves the inconsistency in the list of propositions {\bf P--I} to {\bf P--VI}. In particular  {\bf FSS--IV}  gives {\em a quantitative measure\/} of the deviations of critical phenomena, as observed in finite systems, from the behaviour expected for  infinite system size.  From  (\ref{scalY}), deviations from critical behaviour characteristic of the infinite system will  be observable  in a narrow region around the infinite system critical point. This, however,  is precisely the region, where  one would stand the chance of observing {\em asymptotic singular behaviour\/}, as only in this region is the influence of irrelevant scaling fields on PCTP expected to be sufficiently small.
In order to observe asymptotic critical singularities it is thus required that  $|\theta_\vtmT|$ be {\em sufficiently small\/} to keep corrections to asymptotic critical singularities due to irrelevant scaling fields under control, {\em but also not too small\/}, in order to prevent finite-size corrections from becoming significant.
As the range of  $\theta_\vtmT$ within which finite-size corrections dominate critical behaviour shrinks with system size $\aleph$ like $\aleph ^{-1/\nu}$, one  has to choose systems {\em sufficiently large in a quantitatively well-defined sense} in order to be able to observe asymptotic critical singularities characteristic of the respective universality class of a system.

In the context of the list {\bf P--I} to {\bf P--VI} of propositions, it is important to realise that the characterisation of PTCP in terms of singularities of thermodynamic functions   constitutes an {\em extrapolation of empirical observations\/}, as properly establishing the existence of a discontinuity of a thermodynamic function would require experimental control of infinite precision, while establishing a divergence of a response function would require an actual measurement of an infinite quantity. Neither requirement can conceivably be met  in any realistic experiment. Given that realistic systems  contain $\mcO(10^{23})$ constituents,  the linear dimension $\aleph$ of such systems, measured in terms of atomic distances, is very large and the temperature range  over which finite-size corrections to singular behaviour would manifest themselves, will be very small. It is thus understandable that such effects have been beyond experimental resolution.\footnote{Except for fairly recently in thin films \citep{LB1,WWZZ1}.} On the other hand, in computer simulations of statistical mechanical systems, one can  handle only relatively small systems, and finite-size roundings of critical singularities are therefore quite prominent. In such situations such  roundings, as predicted (and captured) by finite-size scaling are indeed observed and routinely used to extract asymptotic critical exponents from finite-size data \citep{bin1}.
The renormalization group and its formulation of finite-size scaling theory thus predicts in a quantitative way, both, the emergence of critical singularities, described as pure power-law singularities sufficiently close to an infinite system critical point,\footnote{For the sake of completeness, it should be mentioned that under certain well understood conditions, logarithmic corrections to pure power laws can occur \citep{weg1}. They translate into analogous logarithmic corrections in finite-size scaling relations \citep{m&k1}.}  and their shifting and rounding in systems of finite size.

\vspace{0.2cm}

\noindent According to our definite of a incipient singularity (Def.\ \ref{qcritreg}, above) such will occur in a finite system at certain values  of their external parameters, if  at those values thermodynamic
%---------------------------------------------------------------------------------------------
\begin{figure}[t]
  \includegraphics[width=8cm]{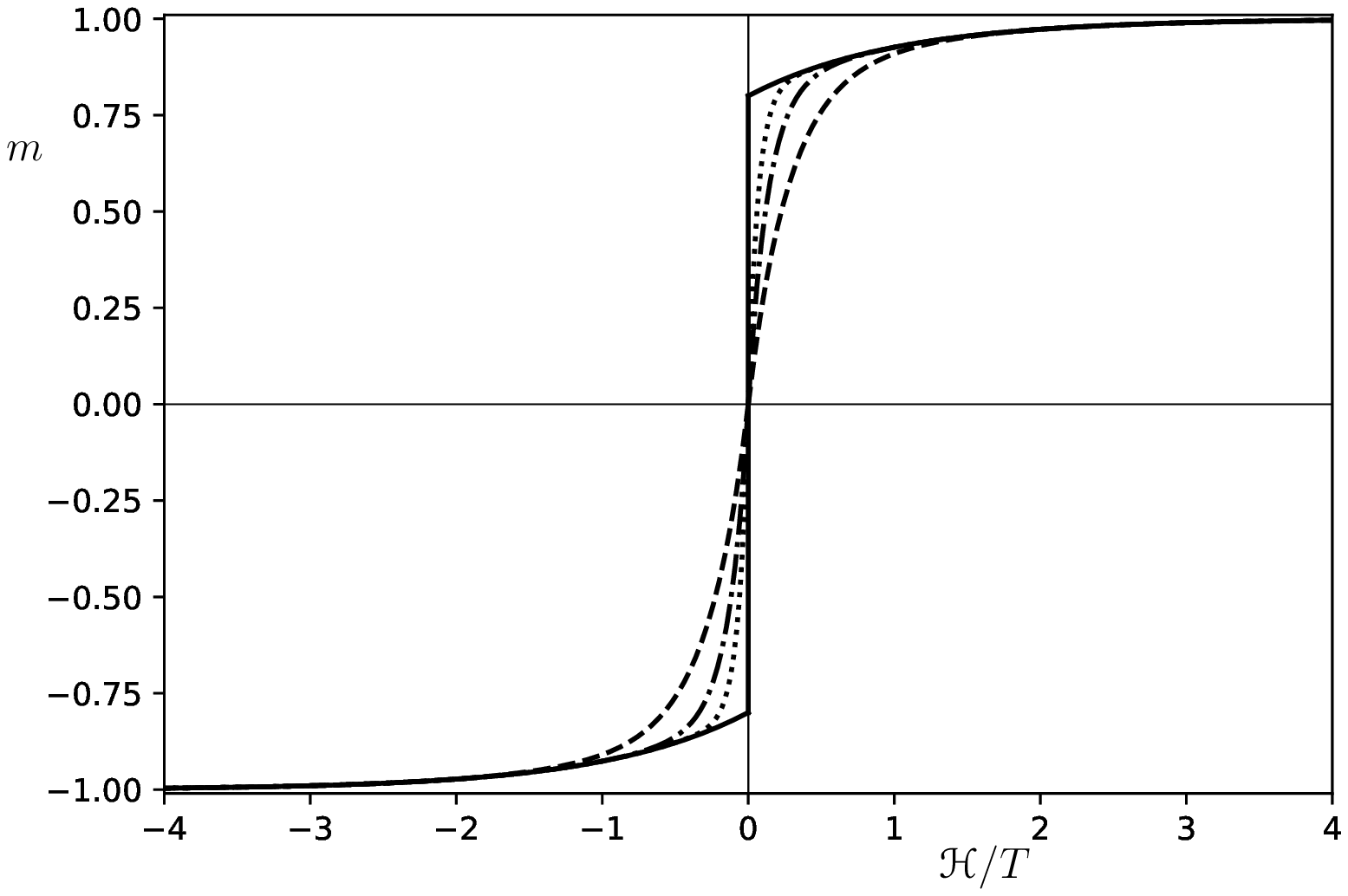}
\caption{Isothermal curves of magnetization density plotted against the field coupling. System size increases from the broken to the chain to the dotted curves with
the infinite system represented by the continuous line.}\label{Fig9}
\end{figure}
%------------------------------------------------------------------------------------------------------
functions exhibit  properties that have no finite limits  as the system size $\aleph$ is increased. This could be a steep increase in the the slope of  magnetization as a function of the external field  across the zero-field axis at low temperatures, as shown in Fig.\ \ref{Fig9}, which is indicative of the possibility of a first-order transition in the infinite system. Or it could be the size-dependent height of the maximum of a response function as shown in (\ref{scalYa})  with $\curomega >0$, which is indicative of the possibility of a second-order transition in the infinite system.
However it is important to note  that an assertion of the occurrence of a incipient singularity in a finite system can never be made with absolute certainty by looking at the behaviour of a single system of any fixed finite size, but {\em only by comparing the behaviour of systems of different sizes.}
That said, our investigations have now provided us with a  well-defined notion of a  large system:
\begin{definition}\label{deflgsys}
\emb{For a system to be counted as large} it must be big  enough to exhibit a range of values of a thermodynamic variable (for example, the temperature) within which the following two phenomena
can both be avoided:
\begin{enumerate}[(i)]
\item the  corrections to scaling (due to the existence of non-zero irrelevant scaling fields) which require the system to be close to an incipient singularity,
\item the   noticeable finite-size corrections in a close neighbourhood of an incipient singularity  (due to a finite value of $\aleph$), which requires the system to be sufficiently far away from an incipient singularity.
\end{enumerate}
\end{definition}
Although, as we saw above, these two conditions pull in opposite directions this tension will become  less acute as the system size increases.
For such systems incipient singularities will be observable in a range of temperatures (or couplings), which are described by the asymptotic critical exponents of infinite systems. These exponents describe  incipient singularities which will never fully materialize in a system of finite extent. They do, however,  provide an economy of description, and lead to a  classification of systems according to their universality class, as described earlier. Quite often the full complexity of the crossover between behaviour described by asymptotic critical exponents and finite-size rounding of thermodynamic functions is far beyond the capabilities of available analytic tools. Taking the thermodynamic limit in a statistical mechanical analysis of a system is also often,\footnote{As in the case of the \citet{1-o2} solution
of the zero-field two-dimensional Ising model.} the only way to carry the calculation through to its end.

\vspace{0.2cm}

\indent The renormalization group approach to PTCP actually plays a dual role in the analysis of critical phenomena.\footnote{This aspect is also highlighted in \citet{HKT1}.} On the one hand it provides micro-reductive methods, firmly embedded in the arsenal of techniques of statistical mechanics, to evaluate critical exponents for given statistical mechanical systems, albeit in most cases only approximately. On the other hand it embodies a new way of looking at such systems, by describing statistical properties of systems at different length scales. It is this radically new way of analysing systems which allows it to put {\em systems with different microscopic properties into a common context\/}, which in turn leads to the identification of fixed points and their basins of attraction as universality classes, thereby revolutionizing the analysis of critical phenomena.

It is perhaps appropriate to add a final twist. Asymptotic critical exponents characterising singularities at phase transitions as they would occur in infinite systems, including exponents that describe corrections to scaling due to irrelevant scaling fields, are obtained from the eigenvalues of a renormalisation group transformation that is linearized in the vicinity of (one of) its fixed points. They are thus obtainable {\em without ever touching or contemplating systems of infinite size}! As we have seen in our discussion above, these critical exponents also govern the way in which finite-size corrections to critical phenomena manifest themselves. In some sense, therefore, it would be fair to say that {\em critical exponents are bona-fide properties of finite systems\/} -- rather than, as mostly discussed, simply properties of potentially infinite systems.

\vspace{0.2cm}

\noindent The aim of our analysis has been to eliminate some of the confusion that has characterised much of the discussion surrounding PTCP in the philosophical (and physics) literature. To summarize our position:
\begin{enumerate}[$\bullet$]
\item It cannot be denied that phase transitions occur in nature. ({\bf P--IA} is accepted).
\item The assertion that they are characterized by singularities is an unwarranted extrapolation of empirical findings. ({\bf P--IB} is rejected). (Asserting the existence of a singularity in an experimental result requires infinitely precise experimental control, or an actual `measurement of the infinite', which is clearly infeasible.)
\item Within thermodynamics, there is {\em no choice\/}  but to describe phase transitions in terms of singularities. (That is,  {\bf P--IIA} and  {\bf P--IIB} are valid statements about the structure of thermodynamics). Equations of state either have unique solutions -- in which case there is no phase transition -- or they may exhibit bifurcations in their solution manifolds, in which case singularities and discontinuities arise.
\item Phase transitions, as they occur in nature, are correctly described by statistical mechanics, the renormalization group and finite-size scaling. Thermodynamics, on the contrary, is fundamentally incapable of an adequate description as it is, from the outset, conceived as a theory of infinitely large systems. ({\bf P--IIIA} is accepted but  {\bf P--IIIB} is rejected).
\item Investigating systems in the limit of infinite system size provides added value in that it allows one to (i) identify exact asymptotic power laws, which the incipient singularities would follow if system sizes could be taken arbitrarily large, (ii) provide a classification of systems according to their universality class.
\end{enumerate}
%-----------------------------------------------------------------------------------------------
\section{After-Thoughts on Reduction and Emergence}\label{emred}
Figure \ref{Fig1}  is a diagrammatic attempt to encapsulate the relationship between thermodynamics including scaling theory,
and the Gibbsian version of statistical mechanics including various approaches to PTCP:   the use of the thermodynamic limit,  the renormalization group and phase transitions
in finite systems. Apart from the formal links spelled out as messages \FSM--\ref{fsm-01}, $\ldots$,
  \FSM--\ref{fsm-04} from statistical mechanics to thermodynamics and the
    connecting relationships \FTD--\ref{ftd-01},  $\ldots$,  \FTD-\ref{ftd-03},
  provided by thermodynamics to statistical mechanics, there is another element of collaborative  interaction, as shown in Fig. \ref{Fig1},  in the direction
from statistical mechanics to thermodynamics; specifically from  the renormalization group to scaling theory.  This has two aspects \emb{substantiation and enrichment}:
\begin{enumerate}[$\bullet$]
\item The Kadanoff scaling relationship (\ref{egs4}) is introduced as a hypothesis, which is {\em substantiated} as
(\ref{egs4rg}) in renormalization group theory.
\item Scaling about an arbitrary origin in a critical region with relevant and irrelevant directions is a consequence of scaling theory.
This picture is {\em enriched} in renormalization group theory, where scaling origins are not arbitrary,
but fixed points of the flow determined by the recurrence relationships,  and corresponding to different universality classes. Relevant and irrelevant directions correspond to directions in which a
 fixed point is repulsive and attractive to the flow. Following a trajectory as it approaches one fixed point, but is finally repulsed towards another, is an example of crossover between different
 types of critical behaviour, that is between different universality classes.
\end{enumerate}
Having spelled out a picture of the anatomy of thermodynamics and statistical mechanics, as well as the relationships between their different parts,
we can now ask what consequences this  has for our understanding of reduction and emergence as regards PTCP. The literature on
 reduction and emergence is vast, even when restricted to the specific area of  PTCP.  So reviewing and discussing this entire literature is beyond the scope of
  this work,\footnote{For surveys of the various positions in the discussion concerning reduction see \citet{HL1} and \citet{vrvg1}.
  For surveys of the discussions of emergence see \citet{hump3} and the contributions of \citet{GHL1}.
  For an overview of discussions  of phase transitions see \citet{she1}.} with a more extensive treatment
  being the subject of a future paper. Our aim in this section is simply to sketch the main contours of the lie of the land in the light
   of the picture we have developed, hoping that this will serve as a springboard for further discussions.

   \vspace{0.2cm}

   \noindent  To aid our account, we introduce the following terminology. Let $ \frakT_{\!\tC}$ and $\frakT_{\!\tF}$ be two theories, where `$\sC$' stands for `coarse', meaning less detailed, and `$\sF$' stands for `finer', meaning more detailed.\footnote{ \citet[][p.\ 928]{butt1} replaces `$\sC$' by   `t'  to stand for top, tangible or tainted and `$\sF$'  by `b' to stand for bottom, basic or best.} Intuitively, $\frakT_{\!\tC}$ is the theory that is supposed to be reduced to $\frakT_{\!\tF}$. In the terminology that has become standard in the philosophical literature on the topic, $\frakT_{\!\tF}$ is supposed to be the \textit{reducing theory}  and $\frakT_{\!\tC}$ is supposed to be the \textit{reduced theory}. We say `supposed to be' because this is what reductionists would expect. The question is whether this expectation bears out, and if so in what sense of reduction.

Accounts of reduction might be divided into two broad families,  called \textit{limit reduction} and \textit{deductive reduction}.\footnote{The distinction goes back to \citet{nic1}.
\citet{bat13} calls them the ``physicist's sense of reduction'' and the ``philosopher's sense of reduction''. We avoid this terminology because, as we will argue, physicists do present deductive reductions.} We now have a look  at each in turn and consider whether they can account for the relation between $\frakT_{\!\tC}$ and $\frakT_{\!\tF}$ that emerges from our account.
%----------------------------------------------------------------------------------------------------------------------------------------------------------------------------------------------------------------
\subsection{Limit Reduction}\label{limred}
 The core idea of limit reduction is that $\frakT_{\!\tC}$ reduces to $\frakT_{\!\tF}$ if the former turns out to be a regular limit of the latter. An example of such a reduction is letting the parameter $c$, the speed of light in the special theory of relativity, tend toward infinity and thereby recovering classical Newtonian mechanics \citep{nic1}.\footnote{We note that the terminology of what reduces to what varies. \citeauthor{nic1} says that in taking the limit $\frakT_{\!\tF}$ reduces to $\frakT_{\!\tC}$; in keeping with the terminology previously introduced, we say that $\frakT_{\!\tC}$ reduces to $\frakT_{\!\tF}$. Nothing in what follows depends on this purely terminological matter.} In general, let us call the relevant parameter $\alpha$; the limit,  denoted as $\lim_\alpha$, can be toward any value
 of $\alpha$, the most frequent cases being  $\alpha \rightarrow 0$ and $\alpha \rightarrow \infty$. \citet{bat13} adds the further requirement that the limit be regular, which means that the relevant formulae in $\frakT_{\!\tF}$ approach the relevant formulae in $\frakT_{\!\tC}$ smoothly as the parameter approach the relevant limit value.\footnote{See \citet{berr1} and for discussions of singular and regular limits see \citet{butt2} and \citet{NgF1}.} Taking these elements together yields the following:
\begin{definition}\label{lim-red}
\textbf{Limit Reduction:}\\
$\frakT_{\!\tC}$ \emb{limit-reduces} to $\frakT_{\!\tF}$ iff $\lim_{\alpha} \frakT_{\!\tF} = \frakT_{\!\tC}$ and the limit is \emb{regular}.
\end{definition}
This definition plays an important role in the discussion about the reduction of PTCP because $\frakT_{\!\tC}$ is commonly associated with thermodynamics and $\frakT_{\!\tF}$ with statistical mechanics.
The failure of the limit to be regular as the number of microsystems tends to infinity is then seen as an indication that reduction fails.

\vspace{0.2cm}

\noindent How does this argument play out in our scheme? To answer this question we first need to identify certain elements in Fig. \ref{Fig1} with $\frakT_{\!\tC}$ and $\frakT_{\!\tF}$. There are two possibilities:
\begin{enumerate}[(a)]
\item  Work within the renormalization group $\textsf{SM4}$. In this case, as described  in Sect.\  \ref{tgs}, the limiting process is implemented by  the renormalization transformation which applies
 a succession of  reductions in the number of lattices sites. This reduces the fluctuations (and correlation length) away from critical regions, but leaves the essential statistical mechanical structure intact.
\item Apply the infinite system limit $\textsf{SM2}\to\textsf{SM3}$.  Away from critical regions this removes fluctuations in the uncontrolled extensive variables, but  leaves
the microstructure and the probability distribution intact.
\end{enumerate}
\noindent However, neither of these is a reduction to a version of thermodynamics.   Both (a) and (b) are procedures lying entirely within statistical mechanics. That having been said,
 (b) is probably the closest to the above idea of reduction. However,  while it uses the thermodynamic limit, that limit does not take the system
 to a thermodynamic system, but to an  {\em infinite statistical mechanical system}  (\textsf{SM3}). To
arrive at  thermodynamics it is necessary to {\em conflate}  \textsf{SM3} with  \textsf{TD3}.
While  \textsf{SM3} like \textsf{TD3} contains the singular characteristics deemed necessary (by some) for the occurrence of phase transitions it also has
a microstructure which is lacking in \textsf{TD3}.

So there is no part of Fig.\  \ref{Fig1} which involves the kind of limit that would ground a limit reduction.
 However, far from being a problem, this is simply irrelevant to the issue of the reduction of PTCP.  As we have indicated in Sect.\ \ref{ourp} the role of the thermodynamic limit
 is, in the first instance,  to provide a condition for maxima in response functions to be incipient singularities; some finite systems do not show PTCP no matter how large they become.
  In the second instance it provides the critical exponents that can be regarded as properties of the real system.  Limits and renormalization group techniques are classification tools that enable us
  to separate phase transitions into different universality  classes.
%----------------------------------------------------------------------------------------------------------------------------------------------------------
\subsection{Deductive Reduction}{\label{dedred}
This notion of reduction  is closely associated with \citeauthor{nag1}. The broad idea  is that $\frakT_{\!\tC}$  is reduced to $\frakT_{\!\tF}$  if the laws of  $\frakT_{\!\tC}$ are deducible from the laws of $\frakT_{\!\tF}$ and some auxiliary assumptions. A mature formulation of this idea, known as the \emb{Generalised Nagel-Schaffner Model of Reduction}, is as follows:\footnote{The original reference is the first (1961) edition of \citet{nag1}, of which  \citet{sch1} provide a reformulation. An alternative account by \citet{butt1,butt2} uses  the notion of a definitional extension. Our presentation here follows that of \citet{DFH1}.}
\begin{definition}\label{nagred}
\noindent \textbf{Deductive Reduction:}\\
\noindent $\frakT_{\!\tC}$ reduces to $\frakT_{\!\tF}$ iff there is a {\em corrected} version $\frakT_{\!\tC}^\star$ of $\frakT_{\!\tC}$ such that:
\begin{enumerate}[(i)]
\item \emb{Connectability}: If $\frakT_{\!\tC}$ contains terms that do not appear in to $\frakT_{\!\tF}$, then for every such term there is a bridge law connecting it to a term in $\frakT_{\!\tF}$.
\item \emb{Derivability}: Given the associations in (i), $\frakT_{\!\tC}^\star$ is derivable from  $\frakT_{\!\tF}$ plus bridge laws  and, possibly, some auxiliary assumptions.
\item \emb{Strong Analogy}: $\frakT_{\!\tC}$ and $\frakT_{\!\tC}^\star$ are {\em strongly analogous} to one another.
\end{enumerate}
\end{definition}
As a simple example, consider the derivation of the perfect gas law $PV=NT$ (given as the second of equations (\ref{prefl1})) from the kinetic theory of gases. Here the perfect gas law is $\frakT_{\!\tC}$ and the kinetic theory is $\frakT_{\!\tF}$. $\frakT_{\!\tC}$ contains the term `temperature', which is not in $\frakT_{\!\tF}$. The bridge law $T\cequals 2U\varepsilon/(3N)$ (which is the first of equations  (\ref{prefl1}))
with the internal energy  $U$ identified as
 the expectation value of the kinetic energy of the gas, connects this term to $\frakT_{\!\tF}$. $\frakT_{\!\tC}^\star$ is the version of the perfect gas law in which, subject to the physical constraints on the system, $P$, $V$, and $T$ are variables that can fluctuate (something they cannot do in $\frakT_{\!\tC}$). $\frakT_{\!\tC}^\star$ and $\frakT_{\!\tC}$ are strongly analogous in that fluctuations are small (to the point of being negligible) in contexts in which $\frakT_{\!\tC}$ is applied.

\vspace{0.2cm}

 \noindent The introduction of $\frakT_{\!\tC}^\star$ is a concession to practice. Ideally one would be able to derive $\frakT_{\!\tC}$ from $\frakT_{\!\tF}$, but that is usually not possible.
  So one rests content with deriving a theory $\frakT_{\!\tC}^\star$  that is not identical with, but strongly analogous to, $\frakT_{\!\tC}$.
What does it mean to be strongly analogous? \citet{sch1} blocks the worry that an appeal to strong analogy is an entry ticket to `anything goes' by imposing the following two conditions:
\begin{enumerate}[(a)]
\item $\frakT_{\!\tC}^\star$ corrects $\frakT_{\!\tC}$ in that $\frakT_{\!\tC}^\star$ makes {\em more accurate predictions} than $\frakT_{\!\tC}$.
\item $\frakT_{\!\tC}$ is {\em explained} by $\frakT_{\!\tF}$ through $\frakT_{\!\tC}^\star$ being a deductive consequence of $\frakT_{\!\tF}$ and $\frakT_{\!\tC}^\star$ being strongly analogous to $\frakT_{\!\tC}$.
\end{enumerate}
With this in place, we can now ask whether the above schema indicates that a deductive reduction is taking place.
For this we first need to know which theories are in play:  what is reduced to what? Since we are interested in a reduction of PTCP,
we should focus on a version of thermodynamics with PTCP in it. So we set $\frakT_{\!\tC}\cequals$\textsf{TD3}.
 Then it might  seem tempting to choose \textsf{SM3} as the reducing theory because,  in Fig. \ref{Fig1},  \textsf{TD3} `communicates' with \textsf{SM3}.
This, however, would be the wrong choice. What we are interested in is the reduction of thermodynamics to a fundamental theory of large systems,
 and this  is \textsf{SM2}. This is because \textsf{SM2} contains the fundamental principles of statistical mechanics with the only added assumption being
 that systems are large; so $\frakT_{\!\tF}$\cequals\textsf{SM2} is appropriate. \textsf{SM3}, by contrast, contains a limit assumption  which does not belong in the
   fundamental theory.
 So the task we set ourselves here    is to check whether the reduction of $\textsf{TD3}$ to $\textsf{SM2}$ fits the mould of deductive reduction. We shall argue that it does
and,  to this end, we now consider this contention in relation how to the  elements (i)--(iii) of Def.\ \ref{nagred} play out in Fig.\ \ref{Fig1}.
\begin{description}
\item[\textbf{For (i)}] connectivity requires a number of bridge laws.  We have avoided this  designation for the relationships
\FTD--\ref{ftd-01}, \FTD--\ref{ftd-02} and \FTD--\ref{ftd-03}  in Fig.\ \ref{Fig1}, preferring to call them `inter-theory connections'.  However, now we shall consider the possibility
that they can assume the role of bridge laws as required in the present context.  The paradigmatic example of a bridge law in the philosophical literature is
provided, as indicated above, by the perfect gas. There the bridge law identifies the temperature  in statistical mechanics using the underlying identification of the expectation value of  kinetic energy
of the gas with its internal energy. But, on closer examination, this example glosses over two other identifications, of volume and pressure.\footnote{And also of the number
of particles of the gas if it is enclosed in a permeable container.}  In a perfect gas contained in a cylinder closed by a movable piston,
the piston position will fluctuate; that is to say, from a statistical mechanical point of view, the volume of the gas is a fluctuating quantity. So, just as the internal energy must be identified
with the expectation value of the kinetic energy,  the thermodynamic volume must be identified with the expectation value of the statistical mechanical volume.  Other instances
of the same kind are provided by other systems and they are all covered by \FTD--\ref{ftd-03}, which in the current context plays the role of a bridge law.   In the case of the perfect fluid the identification of internal
energy and the expectation value of the kinetic energy and of the thermodynamic volume with the expectation value of the statistical mechanical volume is sufficient to provide a bridge for
temperature, pressure and for entropy via the Sackur-Tetrode formula and consequentially for all other thermodynamic variables, as described by the connecting relationships
\FTD--\ref{ftd-01}  and \FTD--\ref{ftd-02}.  These could, therefore, be regarded as consequences of the underlying bridge law  \FTD--\ref{ftd-03}, rather than as bridge laws in their own right.
In more complicated situations, where there is a need to connect a larger set of thermodynamic and statistical mechanical variables, it is a reasonable economy to regard them, together with
 \FTD--\ref{ftd-03} as comprising an exhaustive set of bridge laws.
 \vspace{0.2cm}
\item[\textbf{For (ii)}] by definition $\frakT_{\!\tC}^\star$ is a corrected version of $\frakT_{\!\tC}$ that can be derived from $\frakT_{\!\tF}$ plus bridge laws. In the current context $\frakT_{\!\tC}^\star$ is a version of \textsf{TD3} in which the relevant quantities are allowed to fluctuate, and the fluctuations show roughly the pattern given in \textsf{SM2} (but without $\frakT_{\!\tC}^\star$ containing any of the microstructure of matter specified in statistical mechanics). It is obvious that $\frakT_{\!\tC}^\star$ thus defined is a deductive consequence of \textsf{SM2}: it is obtained simply by applying the bridge laws to \textsf{SM2}.\footnote{Terminological note: the term `corrected', which is customary in the discussion of reduction, is somewhat ill-chosen, because it might suggest that that \textsf{TD3} is in some way faulty, which it is not.  It is in fact one of the most successful and enduring models in physics. The term `corrected' here should be read in a unemphatic (and non-pejorative) way, as simply indicating that conditions (a) and (b), listed above, are met.}
    \vspace{0.2cm}
\item[\textbf{For (iii)}]  we need to show that $\frakT_{\!\tC}^\star$ and $\frakT_{\!\tC}$ stand in the proper strong analogy relationship.
In effect the derivation of \textsf{SM3} from \textsf{SM2} through the thermodynamic limit and the fact that \textsf{SM3} corresponds to \textsf{TD3}
amounts to saying that there is a strong analogy between \textsf{SM2} and \textsf{TD3}. However, a more detailed analysis is useful and for this
 we check whether \citeauthor{sch1}'s two criteria are satisfied:
 \vspace{0.1cm}
\begin{description}
\item[\underline{For (a)}] the messages \FSM--\ref{fsm-01} and \FSM--\ref{fsm-02} are relevant.
  \FSM--\ref{fsm-01}  asserts that uncontrolled thermodynamic variables fluctuate  with  variances of $\mcO(N)$ related to response functions. This means that the variances
  of the corresponding density variables are $\mcO(1/N)$.  That these  fluctuations are small for large systems is related to, but not exactly equivalent to the fact, asserted in \FSM--\ref{fsm-02},
 that extensivity is an approximate property of large systems. So  $\frakT_{\!\tC}^\star$ modifies $\frakT_{\!\tC}$ by replacing equality in the basic relationship with approximate equality, valid
 when the system is large.  It also contains fluctuation--response function relationships between fluctuations, which are recognised in $\frakT_{\!\tC}^\star$ but not in $\frakT_{\!\tC}$,
and response functions which appear in both. Thus $\frakT_{\!\tC}^\star$  makes more adequate predictions than $\frakT_{\!\tC}$ because real systems \textit{do} show fluctuations.
\item[\underline{For (b)}] the way that $\frakT_{\!\tC}\cequals$\textsf{TD3} is explained by $\frakT_{\!\tF}\cequals$\textsf{SM2} follows straightforwardly from Sect.\ \ref{ourp}
once the bridge laws are accepted and we have in place the definition of an incipient singularity (Def. \ref{qcritreg}).  Maxima in response functions are identified as incipient singularities
if they map into real singularities in the thermodynamic limit, which is the step from \textsf{SM2} to \textsf{SM3}.  And, as we have already noted,
\textsf{TD3} communicates with \textsf{SM3} in the sense that it communicates its understanding of the singularities in \textsf{SM3} to \textsf{TD3}.
\end{description}
\end{description}
From the above we conclude, that \textsf{TD3} reduces to \textsf{SM2} in the sense of deductive reduction.
However, the structure of Fig. \ref{Fig1} prompts a consideration
of the possibility of a further relationship of deductive reduction higher in the figure.  In particular does $\frakT_{\!\tC}\cequals$\textsf{TD4} and   $\frakT_{\!\tF}\cequals$\textsf{SM4}
satisfy the required conditions?\footnote{Replacing \textsf{SM4} with  \textsf{SM5}, would rather complicate the situation, since \textsf{TD4} has extensivity in all $d$ dimensions, whereas this is the case for only $\frakd$ dimensions (which includes the fully-finite case $\frakd=0$) in \textsf{SM5}.}
It is straightforward to see that connectability and derivability, where  $\frakT_{\!\tC}^\star$ is a version of \textsf{TD4} that has certain of the features of \textsf{SM4} built into it, are satisfied as before.
Scaling in \textsf{TD4} is a phenomenological means of capturing the structure of the way thermodynamic functions in critical regions depend on variables (in the form of homogeneous functions of controllable variables).  It can in a sense be regarded as being built from  renormalization group theoory with the scaffolding removed. This is what we referred to above as the substantiation
of scaling theory by the renormalization group. On the other hand the {\em values} of critical exponents and the interpretation of the origin of scaling as the fixed point of a semi-group transformation
is absent from \textsf{TD4} but present in \textsf{SM4}.  In that sense the later provides an {\em explanation} or {\em enrichment} of the former.
%-------------------------------------------------------------------------------------------------------------------------------------------------------------------------------------
\subsection{Emergence}{\label{emerg}
  In the case of emergence things are even more difficult than with reduction. As \citeauthor{hump3} notes in a recent review of the field, not only is there no unified framework or account of emergence, there is not even a generally agreed set of core examples of emergent phenomena on which a discussion could build \citep{hump3}. Our aim here is not, therefore,  to comprehensively review the field; we rather discuss some senses of emergence that have played a role in the debate and assess whether, in the light of our analysis, PTCP are emergent in these senses.

\vspace{0.2cm}

\noindent For \citeauthor{butt1}, whose view of reduction is essentially Nagelian, there is no conflict between reduction and emergence. The view that reduction and emergence are compatible is based on an understanding of emergence as there being  ``properties or behaviour of a system which are \textit{novel} and \textit{robust} relative to some appropriate comparison class'' (\citeyear{butt1}, p.\ 921, orig. emph.). He adds the comment that this is intended to cover the case where a system consists of parts, where the idea is that a composite system's ``properties and behaviour are novel and robust compared to those of its component systems, especially its microscopic or even atomic component'' (op. cit.). We agree that thus understood, there is emergence in the large but finite systems we are studying and PTCP can be regarded as both emergent and reduced. Illustrative of this is the transfer matrix approach where
 maxima in response functions and the correlation length (or critical properties if $\frakd>d_{\tL\tC}$), calculated for a lattice which is infinite in $\frakd$ dimensions, converge towards the critical properties of a $(\frakd+1)$-dimensional system as the size in that dimension is increased. This account,  affords a understanding of dimensional crossover between universality classes, with  the  `gradual emergence'  of critical behavior.

\vspace{0.2cm}

\noindent  \citet{hump3} introduces the triplet of conceptual emergence, ontological emergence and epistemological emergence, which we now consider:
\begin{enumerate}[(1)]
\item We have \emb{conceptual emergence} ``when a reconceptualization of the objects and properties of some domain is required in order for effective representation, prediction, and explanation to take place'' (op. cit. p.\ 762). This is close to \citeauthor{butt1}'s notion of reduction, and  there is emergence in this sense because various notions that are not native to statistical mechanics, have been introduced into the theory in order to deal with PTCP, both through inputs from thermodynamics (\textsf{FTD-1}, \textsf{FTD-2}, and \textsf{FTD-3}) and through the introduction of the notion of a large system at level \textsf{SM2}.  As we have argued in Sect.\ \ref{ourp} and in our discussion of transfer matrix methods, it is precisely in such large systems that PTCP are manifested
    in the form of incipient singularities.
\item \emb{Ontological emergence} amounts to the following:``$\textit{A}$ ontologically emerges from $\textit{B}$ when the totality of objects, properties, and laws present in $\textit{B}$ are insufficient to determine $\textit{A}$'' (op. cit. p.\ 762). As we have seen in Sect.\ \ref{ourp}, the properties of a system’s micro-constituents together with the laws that govern them are sufficient to determine PTCP; in fact they can be shown to happen in finite systems. So PTCPs are not ontologically emergent.
\item \emb{Epistemological emergence} is present when the limitations in our knowledge prevent us from predicting the relevant phenomenon. As \citeauthor{hump3} puts it, $\textit{A}$ epistemically emerges from $\textit{B}$ ``when full knowledge of the domain to which $\textit{B}$ belongs is insufficient to allow a prediction of $\textit{A}$ at the time associated with $\textit{B}$'' (op. cit. p.\ 762). This is also the notion of emergence that \citeauthor{mor1} appeals to when she notes that  ``what is truly significant about emergent phenomena is that we cannot appeal to microstructures in explaining or predicting these phenomena, even though they are constituted by them'' \citep[][p.\ 143]{mor1}.\footnote{For detailed discussion of \citeauthor{mor1}'s position see \citet{HKT1}.} We submit that PTCP are not epistemically emergent because, as we have seen in Sect.\  \ref{ourp}, they in fact can be deduced and predicted from the underlying micro-theory. What is important here is PTCP appear in {\em finite} systems.
\end{enumerate}

\vspace{0.2cm}

\noindent \citeauthor{bat6}'s  account of emergence \citep{bat6}, centres around the application of the renormalization group. As we have seen in Sect.\ \ref{tl&rg}
he (and Kadanoff) regard the use of renormalization group as a wholly different type of approach to PTCP from which novel properties emerge.  In particular
the fixed points of the renormalization transformation which allocate the universality classes.  We agree with this except for two reservations:
\vspace{0.2cm}
\begin{enumerate}[(i)]
\item  Batterman takes the thermodynamic limit as an essential feature of this method.  As we have indicated in Sect.\ \ref{tl&rg} we do not regard this as
being necessary.
\item  There is nothing automatic about setting up a renormalization group analysis of a system.  It does not arise in a straightforward algorithmic way from
the basic structure of statistical mechanics. Indeed {\em physical insight} is required both in the
 the choice of the lattice scaling $\mcN\to\wtmcN$ and  of the weight function. These must be compatible with the nature of the ordered state and the critical phenomena to be
 explored.  The recurrence relationships are determined by these choices, and the fixed points `emerge' as properties of the recurrence relationships.  These in turn have exponents which give the universality classes of the various critical regions. As we have already indicated,
most renormalization schemes involve some degree of approximation, with a consequent variation in fixed points and their exponents.\footnote{ An example of such variations, in the case of the Ising model on a triangular lattice, for the exponents $y_\vtmT$ and $y_\vtpH$ and for the location of the fixed point (the Curie temperature)
 is provided by Table IV on page 482 of   \citet{2-n11}.}
However, weight-function dependent variations can also occur
 even when  no approximation is involved.  An example of this is the one-dimensional Ising model with the scheme described in Sect.\
\ref{wtfn} with $\lambda\cequals2$, but with $J<0$, that is the antiferromagnetic case.  In principle one expects a fixed point associated with antiferromagnetism,
 but, although the free-energy density is correctly computed the fixed point is missing.  For this to appear, as is shown by \citet{2-n12},  one needs to take $\lambda\cequals 3$; that is blocks of three sites.
That, in general, different fixed points and hence different universality classes emerge from different choices of lattice scaling and weight function {\em for the same system}  means that  this is a qualified type of emergence.
 \end{enumerate}

 \vspace{0.2cm}

\noindent Finally, emergence is often characterised as the failure of reduction \citep[][p.\ 21]{kim3}. That is, reduction and emergence are taken to be mutually exclusive
 and a property is emergent only if it fails to be reducible. PTCP are not emergent in this sense because, as we have seen above, they are reducible in the sense of a deductive reduction.

\vspace{0.2cm}

%------------------------------------------------------------------------------------------------------------
\section{Conclusions}\label{concl}
We have presented a picture of the way that thermodynamics and statistical mechanics coexist and collaborate within the envelope of thermal physics.
We showed that the relationship between the two developments, represented by the columns in Fig.\ \ref{Fig1} depends,
on the one hand, on  inter-theory connecting relationships from thermodynamics to statistical mechanics, one of which,  \FTD--\ref{ftd-03}, can, in the context of deductive reduction
be regarded as a bridge law, with the remaining two, \FTD--\ref{ftd-01} and \FTD--\ref{ftd-02}, being consequences of \FTD--\ref{ftd-03}.
On the other hand, from statistical mechanics to thermodynamics, there is
 also  a sequence of `messages' that are effectively warnings about the idealized nature of thermodynamics.

We address the problem that real systems are finite, and singular behaviour associated with PTCP can occur only in infinite systems, using finite-size scaling and a
clear specification of a large system.  This enables us to develop a picture of the way that PTCP in finite systems can be defined in terms of incipient
singularities.  Within this picture the role of the infinite system is threefold: (a)  the existence of  a critical region in the thermodynamic limit   is a necessary condition for there to be
a region of incipient singularity in the real finite system,
(b) as one (but not the only) way   to  determine quantitative properties like the value of critical exponents of the real system (c) to
simplify calculations.  In these senses the infinite system is an indispensable, idealized  approximation to the real finite system.

The usual arguments for limit reduction are based on an unwarranted  conflation between a thermodynamic system with critical behaviour (\textsf{TD3}) and an infinite statistical mechanical
system (\textsf{SM3}).  On the other hand, the arguments for the deductive reduction of  \textsf{TD3} to the statistical mechanics of a large system (\textsf{SM2}) are valid.
Next we argue that PTCP are neither ontologically or epistemologically emergent, but they are conceptually emergent.  Rather less frequently remarked upon are the ways that statistical mechanics
both substantiates and enriches  the picture of PTCP in thermodynamics.

\vspace{0.5cm}

\noindent \textbf{\large Acknowledgements}

\vspace{0.2cm}

\noindent  We are grateful to the reviewers of this work for their detailed comments and helpful criticisms.

%------------------------------------------------------------------------------------------------
\appendix
%---------------------------------------------------------------------------------------------------------

\vspace{0.5cm}

\noindent \textbf{\large Appendices}
\addcontentsline{toc}{section}{Appendices}
%-----------------------------------------------------------------------------------------------------------------------
\section{Response Functions and Critical Exponents}\label{refnce}
In terms of densities and fields the response functions are
\begin{equation}
c_{{x}}  \cequals
k_{\tB} {T} {\left({\frac{\partial s}{\partial{T}}}\right)}_{\ms{4}{x}},\smallsep
c_{\xi}  \cequals
k_{\tB} {T} {\left({\frac{\partial s}{\partial{T}}}\right)}_{\ms{4}\xi},
\smallsep
\varphi_{\vtmT}  \cequals  T
{\left({\frac{\partial {x}}{\partial \xi}}\right)}_{\ms{4}\vtmT}, \smallsep
\alpha_{\xi} \cequals {T} {\left( {\frac{\partial {x}}{\partial{T}}}
\right)}_{\ms{4}\xi}.
\label{cp35a}
\end{equation}
And in the coupling--density representation the densities are given by
\begin{equation}
u=\frac{\partial\phi_2}{\partial \zeta_1}, \pairsep x=-\frac{\partial\phi_2}{\partial \zeta_2}
\label{cpden}
\end{equation}
with the response functions  by
\begin{equation}
\begin{array}{l}
\displaystyle{ c_{\xi} = - k_{\tB} \left( {{\zeta_1}}^2 {\frac{\partial^2{\phi_2}}{\partial {{\zeta_1}}^2}} + 2 {{\zeta_1}}{{\zeta_2}}
{\frac{\partial^2 {\phi_2}}{\partial {\zeta_1} \partial
{\zeta_2}}} + {{\zeta_2}}^2{\frac{\partial^2 {\phi_2}}{\partial {{\zeta_2}}^2}}\right),}\\[0.5cm]
\displaystyle{\varphi_{\vtmT}  = - {\frac{\partial^2 \phi_2}{\partial {\zeta_2}^2}},}
\tripsep\displaystyle{\alpha_{\xi}= {\zeta_1} {\frac{\partial^2 {\phi_2}}{\partial {\zeta_2} \partial {\zeta_1}}}
+{\zeta_2} {\frac{\partial^2 {\phi_2}}{\partial {\zeta_2}^2}}.}
\end{array}
\label{cdrsfn}
\end{equation}
The simplest way to treat $c_{x}$ is to use the standard formula
\begin{equation}
 c_{\xi}-c_{{x}} = k_{\tB}  {{{\alpha}_{\xi}^2}/{\varphi_{\vtmT}}}.
\label{cp43}
\end{equation}
In the standardized notation of PTCP
\citep[][]{2-b2}\label{2-b2-4-001} the critical exponents $\curalpha$,
$\curalpha^\prime$, $\curbeta$, $\curdelta$, $\curgamma$ and
$\curgamma^\prime$ are defined by {\eqspace\jot=0.4cm
\begin{eqnarray}
c_{{x}}  &\sim&  \left\{ \begin{array}{@{}ll}
{({T} - {T}_{\rmc})}^{-\curalpha}\, ,
&\mbox{\ along the critical isochore, ${T} > {T}_{\rmc}$,}
\\[0.3cm]
{({T}_{\rmc} - {T})}^{-\curalpha^\prime}\, ,
&\mbox{\ along the coexistence curve, ${T} < {T}_{\rmc}$,}
\end{array}
\right.
\label{cp60}\\
{x} - {x}_{\rmc}  &\sim&  {({T}_{\rmc} - {T})}^{\curbeta}\, ,
\mbox{\ along the coexistence curve, ${T} <{T}_{\rmc}$,}
\label{cp48}\\
\varphi_{\vtmT}  &\sim& \left\{ \begin{array}{@{}ll}
{({T} - {T}_{\rmc})}^{-\curgamma}\, ,
&\mbox{\ along the critical isochore, ${T} > {T}_{\rmc}$,}\\[0.3cm]
{({T}_{\rmc} - {T})}^{-\curgamma^\prime}\, ,
&\mbox{\ along the coexistence curve, ${T} < {T}_{\rmc}$,}
\end{array}
\right.
\label{cp51}\\
\xi - \xi_{\rmc}  &\sim&  ({x} - {x}_{\rmc}){|{x} - {x}_{\rmc} |}
^{\curdelta - 1}\, ,
\mbox{\ along the critical isotherm.}
\label{cp49}
\end{eqnarray}
It} also convenient to define the
exponents $\cursigma$ and $\cursigma^\prime$ according to
\begin{eqnarray}
c_{\xi}  &\sim&  \left\{ \begin{array}{@{}ll}
{({T} - {T}_{\rmc})}^{-\cursigma}\, ,
&\mbox{\ along the critical isochore, ${T} > {T}_{\rmc}$,}
\\[0.3cm]
{({T}_{\rmc} - {T})}^{-\cursigma^\prime}\, ,
&\mbox{ along the coexistence curve, ${T} < {T}_{\rmc}$.}
\end{array}
\right.
\label{cp55}
\end{eqnarray}
In addition to these purely thermodynamic critical exponents three more exponents
 $\nu$, $\nu'$ and $\eta$  arise, from statistical mechanics,  for the pair correlation function and
correlation length defined in Sect.\ \ref{corcoln}. For the correlation length the exponents $\nu$ and $\nu'$ are given by
\begin{eqnarray}
\curxi  &\sim& \left\{ \begin{array}{@{}ll}
{({T} - {T}_{\rmc})}^{-\curnu}\, ,
&\mbox{\ along the critical isochore, ${T} > {T}_{\rmc}$,}\\[0.3cm]
{({T}_{\rmc} - {T})}^{-\curnu^\prime}\, ,
&\mbox{\ along the coexistence curve, ${T} < {T}_{\rmc}$.}
\end{array}
\right.
\label{excl}
\end{eqnarray}
which encapsulates the asymptotic behaviour of the correlation length in a neighbourhood of that critical point.
The situation for the correlation function is rather more complicated since we are concerned not only with its dependence on the couplings
near to a critical region but also on its asymptotic form for a pair of widely separated lattice sites.
However, the result
\begin{equation}
\curGamma(\obr;\zeta_1,\zeta_2)=\frac{f_d(|\obr|/\curxi)}{ |\obr|^{d-2-\cureta}},\label{excf}
\end{equation}
 from Ginzburg--Landau theory \cite[see, for example,][Chap. 5]{lavisnew} in which dependence on the couplings is mediated  through the correlation length
  is believed to have wide applicability.
  %-----------------------------------------------------------------------------------------------------------------------
\section{The Ising Model}\label{tim}
For simplicity we consider a $d$-dimensional hypercubic lattice $\mcN_d$ of $N$ sites with periodic boundary conditions.
At   $\br\in\mcN_d$ there is a microsystem  $\sigma(\br)$ with values $\pm 1$\footnote{These states are usually thought of as spin directions up (+1) and down (-1).} so that the microstate of the system is $\bsigma\cequals\{\sigma(\br)\}$
and the Hamiltonian is
\begin{equation}
\hH(\zeta_\vtmT,\zeta_{\vtpH};\bsigma)\cequals
 -\zeta_\vtmT\sum^{\pop}_{\{{\br},{\br}'\}}\sigma({\br})\sigma({\br}') - \zeta_\vtpH\sum_{\{{\br}\}}\sigma({\br})\,
\label{eu0}
\end{equation}
where the first summation is over all first-neighbour pairs of lattice sites and the thermal and field couplings are
\begin{equation}
\zeta_\vtmT\cequals{J}/{T},\pairsep \zeta_\vtpH\cequals {\pH}/{T},
\label{isingcop}
\end{equation}
respectively, $J$ being an energy parameter, so that $J>0$ corresponds to ferromagnetic behavior, where the states of all first-neighbour pair of sites are aligned,  and $J<0$ corresponds to antiferromagnetic behaviour, where the states of all first-neighbour pairs are anti-aligned;\footnote{A perfect arrangement like this is possible for hypercubic lattices, but not for so-called `close-packed' lattices like the plane triangular lattice and the face-centred cubic lattice.} $\pH$ is a magnetic field.
It will be noted that these are the couplings introduced for the simple magnetic system in Sect.\ \ref{scathr} except that there we used $\varepsilon$ instead of $J$ and it was assumed that
$\varepsilon>0$. This two-state model which was first  solved for
$d=1$ by \citet{2-i1},\footnote{In fact it was suggested to \citeauthor{2-i1}
by his research director Wilhelm Lenz and historians of science like
\citet{3-b2} and \citet{3-n1} often call it the {\em Lenz-Ising model}.}
%----------------------------------------------------------------------------------------------------------
\begin{figure}[h]
  \includegraphics[width=4cm]{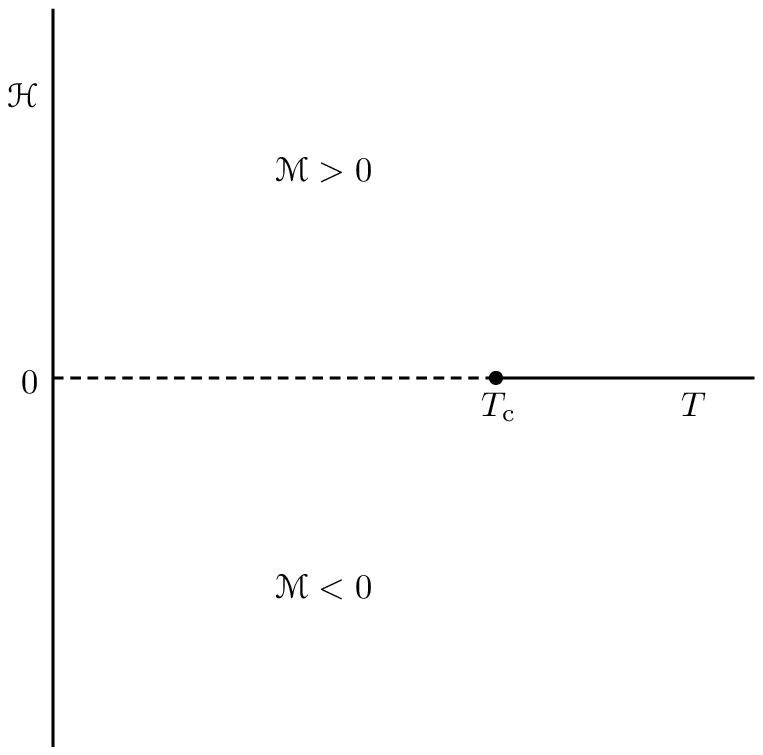}
\caption{The phase diagram for the Ising model. The  first-order
transition is shown as a broken line.}\label{Fig10}
\end{figure}
%-------------------------------------------------------------------------------------------------------------------
 is usually now called the \emb{spin-$\frac{1}{2}$ Ising model}.\footnote{The corresponding three-state model with states $0,\pm 1$ being the spin-1 Ising model.}
On the basis of this solution, and a wealth of other results for larger dimensions (including the exact  zero-field solution for a square lattice by \citet{1-o2})  the ferromagnetic phase diagram in the
space of the temperature $T$ and the field $\pH$ is known to have the form  shown in Fig.\ \ref{Fig10},  where, for $d=1$,  $T_\rmc=0$, for $d=2$, $T_\rmc=2.2692\, J$ and, for $d=3$,
$T_\rmc=4.5108\,J$.  Apart from in the one-dimensional case where the critical point is at zero temperature (see Sect.\ \ref{wtfn}(a)),
 there is a line of first-order transitions along the interval $[0,T_\rmc)$ of the zero-field axis
across which the magnetization $\pM$ changes between equal and opposite values.  The universality class of the second-order transition at the critical point depends on the dimension of the system.  The critical exponents for $d=2$ are $\curalpha=0$ (a logarithmic singularity), $\curbeta=\frac{1}{8}$, $\curgamma=\frac{7}{4}$, $\curdelta=15$.  For $d=3$ the exponents are obtained to a high level of accuracy by series methods with $\curalpha=0.11008$,   $\curbeta=0.326419$, $\gamma=1.237075$ and $\delta=4.78984$.
When $d\ge d_{\tU\tC}=4$ the critical exponents take their classical values  $\curalpha=0$, $\curbeta=\frac{1}{2}$, $\curgamma=1$ and $\curdelta=3$.\footnote{For  a definition of the
 upper-critical dimension see footnote   \ref{ucd}.}
%-----------------------------------------------------------------------------------------------------------

%---------------------------------------------------------------------------
\end{document}